\documentclass{article}

\usepackage{arxiv}

\usepackage[T1]{fontenc}    
\usepackage[numbers]{natbib}
\usepackage{hyperref}       
\usepackage{url}            
\usepackage{booktabs}       
\usepackage{amsfonts, amsmath, amssymb}       
\usepackage{nicefrac}       
\usepackage{microtype}      
\usepackage{graphicx}

\usepackage{algorithm}
\usepackage{algpseudocode}
\usepackage{tikz}
\usepackage{standalone}
\usetikzlibrary{decorations.pathreplacing}
\usepackage{pgfplots}
\usepgfplotslibrary{dateplot}
\pgfplotsset{compat=1.18}
\usepackage{amsthm}

\newcommand{\Prob}{\mathbb{P}}
\newtheorem{lemma}{Lemma}
\newtheorem{theorem}{Theorem}

\newcommand{\R}{\mathbb{R}}

\newcommand{\sla}[1]{\textcolor{black}{#1}}
\newcommand{\slb}[1]{\textcolor{black}{#1}}
\newcommand{\slc}[1]{\textcolor{black}{#1}}

\makeatletter
\newcommand*{\textlabel}[2]{%
  \edef\@currentlabel{#1}
  \phantomsection
  #1\label{#2}
}
\makeatother

\title{Data-adaptive structural change-point detection via isolation}

\author{Andreas Anastasiou \\
	Department of Mathematics and Statistics \\ 
    University of Cyprus\\
	\texttt{anastasiou.andreas@ucy.ac.cy} \\
	\And
	Sophia Loizidou \\
	Department of Mathematics\\
    University of Luxembourg\\
	\texttt{sophia.loizidou@uni.lu} \\
}

\date{}

\renewcommand{\headeright}{}
\renewcommand{\undertitle}{Published in Statistics and Computing}
\renewcommand{\shorttitle}{Data-adaptive structural change-point detection via isolation}

\begin{document}
\maketitle

\begin{abstract}
	In this paper, a new data-adaptive method, called DAIS (Data Adaptive ISolation), is introduced for the estimation of the number and the location of change-points in a given data sequence. The proposed method can detect changes in various different signal structures; we focus on the examples of piecewise-constant and continuous, piecewise-linear signals. 
    The novelty of the proposed algorithm comes from the data-adaptive nature of the methodology. At
    each step, and for the data under consideration, we search for the most prominent change-point in a targeted neighborhood of the data sequence that contains this change-point with high probability. Using a suitably chosen contrast function, the change-point will then get detected after being isolated in an interval. The isolation feature enhances estimation accuracy, while the data-adaptive nature of DAIS is advantageous regarding, mainly, computational complexity. 
    The methodology can be applied to both univariate and multivariate signals.
    The simulation results presented indicate that DAIS is at least as accurate as state-of-the-art competitors and in many cases significantly less computationally expensive.
\end{abstract}



\section{Introduction} \label{introduction}

Change-point detection, also known as data segmentation, is the problem of finding abrupt changes in data when, at least, one of their properties over time changes. 
It has attracted a lot of interest over the years, mostly due to its importance in time series analysis and the wide range of applications where change-point detection methods are needed. 
These include genomics (\cite{10.1093/biomet/asv031}), neuroscience (\cite{seeded_bs}), seismic data (\cite{se-12-2717-2021}), astronomy (\cite{chan2022}), and finance (\cite{10.2307/24310145}).

There are two directions; sequential (or online) and a posteriori (or offline) change-point detection. 
We focus on the latter, where the goal is to estimate the number and locations of changes in given data.
We work under the model
\begin{equation} \label{model}
    X_t = f_t + \sigma \epsilon_t, \quad t = 1, 2, \ldots, T,
\end{equation}
where $X_t$ are the observed data and $f_t$ is a one-dimensional deterministic signal with structural changes at certain points.
We highlight, however, that the case when $X_t$ is multivariate is also discussed in our paper. We focus on the case of piecewise-constant and continuous, piecewise-linear signals $f_t$, meaning that the changes are in the mean or the slope. However, the proposed algorithm can be extended to other, possibly more complicated, signal structures. 
Although detecting changes in the mean is a simple case, as noted by \cite{brodsky2000non}, more complex change-point problems that allow changes in properties other than the mean, can be reduced to the segmentation problem we are studying. 
This can be done by applying a suitable transformation to the data that reveals the changes as those in the mean of the transformed data. 
For example, \cite{cho2015multiple} and \cite{ccid} provide methods for the detection of multiple change-points in the second-order (i.e. autocovariance and cross-covariance) structure of possibly high dimensional time series, using Haar wavelets as building blocks. 
The signal is transformed into a piecewise-constant one with change-point locations identical to those in the original signal. 
We want to emphasize that the novelty of our paper does not come from the structure of changes that the method can be applied to.
Upon correct choice of a contrast function, we can cover more scenarios such as the case of piecewise-quadratic signals (see \cite{baranowski2019narrowest} for more details).
The importance and novelty of our work lies mainly in the data-adaptive nature of the proposed algorithm itself, which, at each step, starts searching for change-points in the areas that most prominently include one.

The various methods available for solving the problem of change-point detection 
can be mainly split into two categories, according to whether they are optimization-based methods or they use an appropriately chosen contrast function. 
The former group includes methods that look for the optimal partition of the data, by performing model selection using a penalization function to avoid overfitting.
The latter category's methods do not search for the globally optimal partition of the signal. 
Instead, the change-point locations are chosen as the most probable location at each step of the algorithm.
Only some of the already existing methods are mentioned in this introduction but comprehensive overviews and more detailed explanations can be found in \cite{TRUONG2020107299} and \cite{yu2020review}.

Starting with the optimization-based methods for detecting changes in the mean when $f_t$ is piecewise-constant, one of the most common penalty functions is the Schwarz Information Criterion (\cite{10.2307/2958889}), which was used by \cite{YAO1988181} for change-point detection, assuming Gaussian random variables $\epsilon_t$ in \eqref{model}.
Under the same assumption, \cite{yao1989least} studies an estimator based on least squares.
Relaxing this assumption and allowing instead for more general exponential family distributions, \cite{HAWKINS2001323} introduces a dynamic programming algorithm which uses maximum likelihood estimates of the location of the change-points. 
\cite{NINOMIYA2005237} introduces an AIC-type criterion for change-point models and \cite{7938741} shows that while an AIC-like information criterion does not give a strongly consistent selection of the optimal number of change-points, a BIC-like information criterion does.
\cite{1381461} employs dynamic programming to guarantee that the exact global optimum, in terms of creating the segments, is found, while
\cite{doi:10.1080/01621459.2012.737745} proposes an algorithm, called PELT, whose computational cost is linear in the number of observations.
\cite{rigaill2015pruned} introduces the pDPA algorithm, which includes a pruning step towards complexity reduction.
\cite{maidstone2017optimal} combines the ideas from PELT and pDPA leading to two new algorithms, FPOP and SNIP, which have low computational complexity.
\sla{Finally, \cite{frick_multiscale_2014} introduces SMUCE, a new estimator for the change-point problem in exponential family regression and \cite{Verzelen_optimal} proposes two procedures achieving optimal rate, a least-squares estimator with a new multiscale penalty and a two-step multiscale post-processing procedure.} 

Focusing now on the category where a contrast function is used, for $f_t$ being piecewise-constant the relevant function is the absolute value of the CUSUM statistic, which is defined in Section~\ref{piecewise constant}. 
A method that has received a lot of attention is the Binary Segmentation algorithm, as introduced in \cite{vostrikova1981detecting}.
It starts by searching the whole data sequence for one change-point.
Subsequently, at each step, the data sequence is split according to the already detected change-points.
However, because of checking for a single change-point in intervals that could have more than one change-points, Binary Segmentation has suboptimal accuracy; many methods have been developed with the scope to improve on such drawbacks.
One of these methods is proposed by \cite{fryzlewicz2014wild}, called WBS, which calculates the value of the contrast function on a large number of randomly drawn intervals,
which allows detection of change-points in small spacings.
\cite{fryzlewicz2020detecting} proposes WBS2 which draws new intervals each time a detection occurs.
A different approach is adopted by \cite{baranowski2019narrowest} with the NOT algorithm. 
By choosing the narrowest interval for which the value of the chosen contrast function exceeds a pre-defined threshold, 
there is exactly one change-point in each interval with high probability. 
The ID method in \cite{anastasiou2022detecting} achieves, first, isolation and then detection of each change-point through an idea based on expanding intervals in a sequential way, starting from both the beginning and the end of the data sequence, in an interchangeable way. 
The isolation of each change-point maximizes the detection power.
\cite{seeded_bs} proposes SeedBS which uses a deterministic construction of search intervals which can be pre-computed while \cite{chu1995mosum}, \cite{10.3150/16-BEJ887} and \cite{10.1214/22-EJS2101} consider moving sum statistics.
\cite{10.1214/17-AOS1662} achieves a multiscale decomposition of the data with respect to an adaptively chosen unbalanced Haar wavelet basis, using a transformation of the data called TGUH transform and \cite{chan2022}, through a `bottom-up' approach, explores reverse segmentation, which involves the creation of a `solution path' by deleting the change-point with the smallest CUSUM value in the segment determined by its closest left and right neighbors, in order to obtain a hierarchy of nested models. 
A method that does not fall in any of the two categories described above is FDRSeg, proposed by \cite{10.1214/16-EJS1131}, which controls the false discovery rate in the sense that the number of falsely detected change-points is bounded linearly by the number of true jumps.

Methods for detecting more general structural changes can again be split into the same two categories.  Focusing on optimization-based methods, \cite{10.2307/2998540} considers the estimation of linear models based on the least squares principle. 
\cite{doi:10.1137/070690274} introduces a trend filtering (TF) approach to produce piecewise-linear trend estimates using an $\ell_1$ penalty. 
\cite{doi:10.1080/10618600.2018.1512868} presents a locally dynamic approach, called CPOP, that finds the best continuous piecewise-linear fit to the data using penalized least squares and an $\ell_0$ penalty. 
\cite{10.1214/aos/1176347963} proposes the MARS method for flexible regression using splines functions with the degree and knot locations determined by the data. 
Two methods for optimizing the knot locations and smoothing parameters for least-squares or penalized splines are introduced by \cite{doi:10.1080/00949655.2011.647317}. 
\cite{WIGGINS2015346} introduces the frequentist information criterion for change-point detection which can detect changes in the mean, slope, standard deviation or serial correlation structure of a signal whose noise can be modeled by Gaussian, Wiener, or Ornstein-Uhlenbeck processes. 
\cite{SOSACOSTA20182044} introduces PLANT, a bottom-up type algorithm that finds the points at which there is potential variation of the slope using a likelihood-based approach, and then recursively merges the adjacent segments. 
Considering algorithms using a suitably chosen contrast function, the algorithm NOT can be applied to both piecewise-linear and quadratic signals, while ID can be applied to piecewise-linear and can be extended to piecewise-quadratic signals.
Finally, \cite{maeng2023detecting} extends TGUH to piecewise-linear signals and \cite{mosum_linear} proposes a moving sum methodology.


In this paper, we propose a data-adaptive change-point detection method, called Data Adaptive ISolation (labelled DAIS), that attempts to isolate the change-points before detection, by also taking into account the potential true locations of the change-points. 
This means that the algorithm does not start from a random interval, or the beginning/end of the signal, as most algorithms in the literature, but instead starts checking around a point that there is reason to believe it could be around a true change-point.
The idea behind the data-adaptive nature of the algorithm, and the belief that the algorithm will start close to a change-point, arise from the fact that, for example, in the case that the unobserved true signal $f_t$ in \eqref{model} is piecewise-constant, taking pairwise differences between consecutive time steps will reveal a constant signal with the value 0, with spikes where the change-points occur. 
This means that the largest spike, in absolute value, occurs at the location of the change-point with the largest jump in the sequence, therefore the most prominent change-point. 
Similarly, in the case that $f_t$ is continuous and piecewise-linear, differencing the signal twice will again reveal a signal with spikes near to the true locations of the change-points. 
We try to use this fact in the observed signal, $X_t$. 
In the algorithm, we identify the location of this spike, which we will be referring to as the `largest difference' for the rest of the paper, and test around it for possible change-points.
A discussion about the connection between the isolation aspect of the algorithm and the location of the largest difference can be found in Section~\ref{discussion}, while Section~\ref{discussion prob} complements this discussion with theoretical results in the case that $\epsilon_t$ is Gaussian.
We want to emphasize that since the largest difference is only used as the starting point of the search for change-points in the data sequence, the true change-points are detected even when the largest difference is far from them.
The novelty of our work lies in the data-adaptive nature of the algorithm as just described. 

The DAIS algorithm uses left- and right-expanding intervals around the location of the largest difference (in absolute value) found, denoted by $d_{s,e}$, in the interval $[s,e]$ in order to identify the potential change-point that might have caused the spike and so must be close to $d_{s,e}$. 
This is done in a deterministic way around $d_{s,e}$, expanding once either only to the left or only to the right at each step, in an alternate way. 
Using expanding intervals, we achieve isolation of the change-points, which is desirable as the detection power of the contrast function is maximized in such cases
(see Sections \ref{piecewise constant} and \ref{piecewise linear} for the choice of the contrast function). 
Due to the alternating sides of the expansions, the location of the largest difference, which, as explained before, we have a reason to believe to be around the location of a change-point, is at the midpoint of the interval being checked (or close to it) after an even number of steps has been performed. 
This increases the power of the contrast function and so gives an advantage to the detection power of the method.
Using an expansion parameter $\lambda_T$ (more details in Section \ref{parameters}) at most $\lceil (e-s+1) / \lambda_T \rceil + 1$ steps are necessary to check the whole length of the interval $[s,e]$ around the largest difference. 
As soon as a change-point is detected, DAIS restarts on two disjoint intervals, one ending at the start-point of the interval where the detection occurred and one starting from the end-point of that same interval. 
A more detailed explanation of the algorithm can be found in Section \ref{methodology}.

The data-adaptive nature of DAIS is what differentiates it from its competitors. 
A key methodological difference with other algorithms is the way the intervals are being checked.
\sla{WBS and NOT, as a first step, randomly draw a number of intervals, and derive the maximum of the contrast function values within those intervals. For those values exceeding a certain, predefined threshold, the change-point is detected for WBS at the location where the maximum, among the maximums within each interval, is attained and for NOT at the location corresponding to that maximum value (that surpasses the threshold) that belongs to the smallest interval drawn. WBS2 draws new random intervals every time a change-point is detected.}
\sla{ID uses one-sided expanding intervals starting from the end-points of the data sequence under consideration in order to achieve isolation before detection, and so the intervals are checked in a deterministic way. All aforementioned algorithms will follow the described method for any data sequence, 
drawing random intervals (WBS, WBS2, NOT) or starting from a fixed point (ID).
In contrast to this, DAIS adapts its starting-point at every step to the nature of the data, as it takes into account the location and magnitude of the changes when calculating the location of the largest difference.
Expanding intervals around the location of the largest, in absolute value, difference, are created and thus detection of the change-point in an interval in which it is isolated occurs with high probability.
This way of expanding around $d_{s,e}$ increases the detection power, as explained in Section~\ref{sec: DAIS_algorithm}.
}
The data-adaptivity also gives an advantage in terms of computational complexity, as is explained in Section~\ref{computational complexity}.

The paper is organized as follows. 
Section \ref{methodology} provides details on the methodology of DAIS. 
We present a simple example followed by a detailed explanation of the algorithm and, 
in Section~\ref{discussion}, we discuss the connection between the location of the largest difference calculated at each step of the algorithm and the guarantee that the change-point can be detected in an interval in which it is isolated.
In Section \ref{theory}, we provide theoretical results regarding the consistency of the number of change-points detected and the accuracy of their estimated locations as well as some theoretical results concerning the discussion of Section \ref{discussion}. 
Section \ref{Computational complexity and practicalities} includes an explanation of how some parameters of DAIS are selected
and a note on computational complexity, which is also compared to that of competitors.
In Section \ref{simulations}, we provide results on simulated data and a comparison to state-of-the-art competitors. 
Section \ref{sec: DAIS extensions} discusses extensions to the algorithm in more complicated signal structures, which include temporal correlation, relaxing the Gaussianity assumption on the noise term $\epsilon_t$, and multivariate data sequences.
In Section \ref{Real data}, two examples using real world data are presented. 
The first one is about crime data reported daily and the second one involves weekly data on the Euro to British pound exchange rate.
Section \ref{Conclusions} concludes the paper with a summary of the most important findings. 
The proof of Theorem~\ref{consistency_theorem} and the outline of the proof of Theorem \ref{thm: consistency_multivariate} can be found in Appendix~\ref{proofs} and \ref{appendix: proof_multivariate}, respectively.
Finally, further simulation results, \sla{an investigation of the impact of the choice of parameter values in the performance of the algorithm,} and the proof of Theorem \ref{consistency_theorem_slope} can be found in the supplementary material.

\section{Methodology} \label{methodology}

\subsection{Simple example}
We begin by presenting a toy example of how DAIS works in practice in the case of $f_t$ in \eqref{model} being a piecewise-constant signal. 
We have a sequence with length $T = 100$, standard deviation $\sigma=1$, and a change-point at location $r=65$ with a jump of magnitude equal to $1.5$. The data are presented in the left plot of Figure \ref{fig:toy example} with the true underlying signal plotted in red.
We first calculate the absolute difference of consecutive observations, $X_t$ and $X_{t+1}$ for $t=1,2,\dots, T-1$. 
As can be seen by the plot on the right of Figure \ref{fig:toy example}, in this example the largest difference occurs at the location of the change-point ($d_{1,100}=65$), which means that $\textrm{argmax}_{t = 1,2, \ldots, 99}\lvert X_{t+1}-X_{t}\rvert = 65$. 
As shown in the left plot of Figure \ref{fig:tikz3}, DAIS creates right- and left-expanding intervals around this point, where, for the sake of presentation of this example, the expansion parameter is chosen to be $\lambda_T = 10$.
The intervals in the order in which they are checked are $[65,74]$, $[55,74]$ and $[55,84]$, where detection occurs. 
Note here that the algorithm does not detect the change-point in the first two intervals being checked, most probably either because the magnitude of the jump compared to $\sigma$ is small, or because the number of observations used is not sufficiently large, or a combination of the two.
The change-point is identified at location 65 and then the algorithm restarts in the intervals $[1,55]$ and $[84,100]$. 
As these two intervals contain no change-points in our example, the largest differences are just at points where the absolute value of the difference between the noise of two consecutive observations is large.
In this case, the location of the largest difference is arbitrary and only acts as the starting point for searching the whole interval for change-points to confirm that there is none.
The largest differences are now detected at points $d_{1,55} = 46$ and $d_{84,100} = 85$ in the two intervals, respectively. 
The algorithm starts checking the interval $[1,55]$, with the expansions being performed in a similar way as to the smaller interval $[84,100]$, on which we focus for simplicity.
For the interval $[84,100]$ the left- and right-expanding subintervals checked are: $[85,94]$, $[84,94]$, and $[84,100]$, as shown in the right plot of Figure \ref{fig:tikz3}. 
Since no change-point is detected in either of the intervals, the algorithm terminates.

\begin{figure}[htbp]
\centering
  \includegraphics[scale = 0.6]{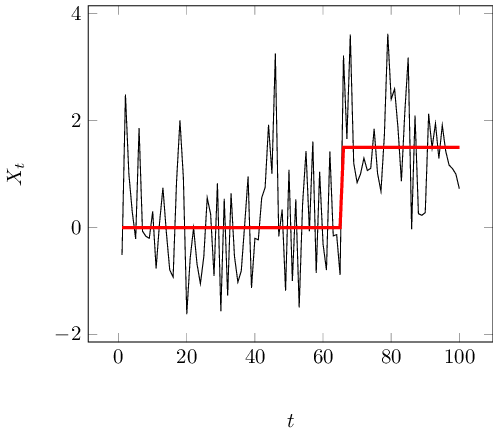}
  \includegraphics[scale = 0.6]{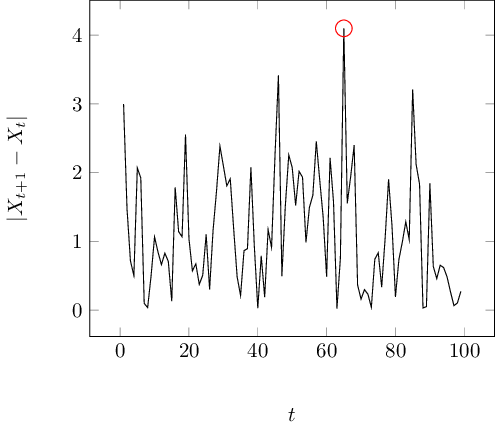}
  \caption{\textbf{Left plot:} The data (in black) and the underlying signal (in red) used for the toy example. The change-point is at location 65. \textbf{Right plot:} The absolute values of consecutive pairwise differences. The maximum, with location $d_{1,100} = 65$, is marked with a red circle.}
  \label{fig:toy example}
\end{figure}

\begin{figure}[htbp]
\centering
  \includegraphics[scale = 0.6]{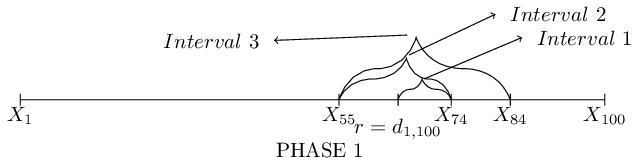}
  \includegraphics[scale = 0.6]{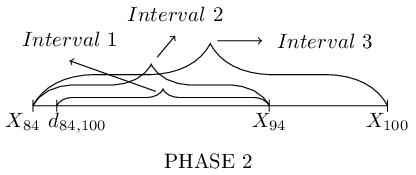}
  \caption{Example with one change-point at $r=65$ with jump 1.5. \textbf{Left plot:} The largest difference is at $d_{1,100}=65$ and the change-point is detected when the third expanding interval is checked. After detection the next intervals that are checked are $[1,55]$ and $[84, 100]$. \textbf{Right plot:} The left- and right-expanding intervals around the location of the largest difference, $d_{84,100}=85$, when checking $[84, 100]$. No change-point is detected.}
  \label{fig:tikz3}
\end{figure}

We emphasize that the differences are only used as a way to start the algorithm in an area close to the change-point. 
In cases with high variance, $\sigma^2$, of the noise, the largest difference might not be around the change-point but under the assumptions specified in Section \ref{theory}, the change-points will still be detected with high probability.
The location of the largest difference that guarantees detection of the change-point in an interval in which it is isolated is discussed in Section~\ref{discussion}, and further results are provided in Section~\ref{discussion prob} for the case that $\epsilon_t$ in \eqref{model} is Gaussian.

\subsection{The DAIS algorithm} \label{sec: DAIS_algorithm}
In the following paragraphs, the method is presented in more generality. For the rest of the paper, we use $s, e$ to denote the start- and the end-point, respectively, of the interval under consideration, while $T \in \mathbb{N}$ denotes the total length of the given signal and $1\leq s<e\leq T$. We denote by $N \in \mathbb{N}_0$ and by $r_i$ for $i=1,2,\dots,N$ the number and the locations of the true change-points, respectively, with the latter ones sorted in increasing order, while $\hat{N}$ and $\hat{r}_i$ for $i=1,2,\dots,\hat{N}$ are their estimated values. The true, unknown number, $N$, of change-points is allowed to grow with the sample size. We denote by $\lambda_T \in \mathbb{N}$ the expansion parameter and $\zeta_T \in \mathbb{R^+}$ the threshold, which is used to decide whether a time-point is a change-point; more details on the choice of their values are
given in Section \ref{parameters}. 
In addition, $d_{s,e}$ is the location of the largest difference, in absolute value, in the interval $[s,e]$  and is defined as
\begin{equation}\label{def_largest_diff}
    d_{s,e} = \left\{\begin{array}{ll}
    \underset{t \in \{s, s+1, \dots, e-1\}}{\textrm{argmax}} \{|X_{t+1} - X_t|\}, & f_t \textrm{ is piecewise-constant}, \\
    \underset{t \in \{s, s+1, \dots, e-2\}}{\textrm{argmax}}\{\lvert X_{t+2} - 2X_{t+1} + X_{t} \rvert\}, & f_t \textrm{ is continuous, piecewise-linear.}
    \end{array}\right.
\end{equation}
We denote by $C^b_{s, e}(\boldsymbol{X})$ the contrast function applied to the location $b$ using the values $\{X_s, X_{s+1},\ldots,  X_e\}$ for $b\in[s,e)$. The functions used for the cases of piecewise-constant and continuous, piecewise-linear signals are explained in Sections \ref{piecewise constant} and \ref{piecewise linear}, respectively. 
For $K^l = \lceil \frac{d_{s,e} - s + 1}{\lambda_T} \rceil$, $K^r = \lceil \frac{e - d_{s,e} + 1}{\lambda_T} \rceil$, $K^{\max} = \max\{K^l, K^r\}$, $K^{\min} = \min\{K^l, K^r\}$, and $K = K^{\max} + K^{\min}$, then the start- and end-points of the intervals $[c_m^l, c_k^r]$ in which the algorithm tests for the existence of a change-point, are given by
\begin{align*}
    & c_{m}^l =
    \max\{d_{s,e}-m\lambda_T, s\}, \quad m =0, 1, \dots,K^{\max}
    , \\
    & c_{k}^r = \min\{d_{s,e}+k\lambda_T-1,e\}, \quad k =1,2, \dots,K^{\max}.
\end{align*}
These are collected in the ordered sets, $L_{d_{s,e},s,e}$ and $R_{d_{s,e},s,e}$, defined in \eqref{end-points}, which consist of the left and right end-points, respectively, of the intervals used in DAIS. 
\sla{It holds that for $m, k\leq K^{\min}$, $s \leq c_{m}^l < c_{m-1}^l$ and $c_{k-1}^r < c_{k}^r \leq e$, while for $m,k \in \{K^{\min}+1, \ldots,K^{\max}\}$, either $s = c_{m}^l = c_{m-1}^l$ and $c_{k-1}^r < c_{k}^r \leq e$ or $s \leq c_{m}^l < c_{m-1}^l$ and $c_{k-1}^r = c_{k}^r = e$ (so one end-point remains fixed).}
Since the expansions occur in turn once from the right and once from the left, the points appear twice in each set, up until the moment when one end-point is equal to either $s$ or $e$, \sla{which occurs when $K^{\min}$ expansions have occurred on both sides, and so $2K^{\min}$ expansions in total}. 
The left- and right-expanding intervals checked are collected in the ordered set $I_j$ in \eqref{intervals}. 
The ordered sets of the end-points for the expansions are:
\begin{align} \label{end-points}
    & L_{d_{s,e},s,e} = \bigl\{c^l_0, c^l_1, c^l_1, c^l_2, \dots, c^l_{K^{\min}-1}, c^l_{K^{\min}}, c^l_{K^{\min}+1}, c^l_{K^{\min}+2}, \ldots, c^l_{K^{\max}}\bigr\}, \nonumber\\
    & R_{d_{s,e},s,e} = \bigl\{c^r_1, c^r_1, c^r_2, c^r_2, \dots, c^r_{K^{\min}}, c^r_{K^{\min}}, c^r_{K^{\min} + 1}, c^r_{K^{\min}+2}, \ldots, c^r_{K^{\max}}\bigr\},
\end{align} 
while the ordered set of the intervals checked can be written as:
\begin{align}
    I_{j} = \bigl[s_j, e_j\bigr] = \Bigl[ L_{d_{s,e},s,e}[j], R_{d_{s,e},s,e}[j]\Bigr], \textrm{ for }  j \in \{1, 2, \ldots, K\} \label{intervals},
\end{align}
where A[i] denotes the $i^{th}$ element of the ordered set A.
It is easy to show that $\lvert L_{d_{s,e},s,e} \rvert = \lvert R_{d_{s,e},s,e} \rvert = K$, where $\lvert A \rvert$ denotes the cardinality of the set A.
Note that the largest number of expansions required in order to check the whole length of the interval $[s,e]$ around $d_{s,e}$ is $\lceil \frac{e-s+1}{\lambda_T}\rceil + 1$.
The intervals $I_j$ defined in \eqref{intervals} are not deterministic. The end-points depend on the location of the largest difference, which depends on the structure of the unknown signal and the noise of the data sequence under consideration.

\begin{algorithm}[htbp]
\caption{DAIS}\label{alg:DAIS}
\begin{algorithmic}
\State \textbf{function} DAIS($s, e, \lambda_{T}, \zeta_T$)
\If{$e-s<3$}
    \State STOP
\Else
    \State Set $d_{s,e}$ as in \eqref{def_largest_diff}
    \State For $j \in \{1,2,\dots,K\}$ let $I_j=[s_j, e_j]$ as in \eqref{intervals}
    \State $i=1$
    \State \textbf{(Main part)}
    \State $b_{i} = {\rm{argmax}}_{t\in [s_{i}, e_{i})} C_{s_{i}, e_{i}}^t\left(\boldsymbol{X}\right)$
    \If{$C_{s_{i}, e_{i}}^{b_{i}}\left(\boldsymbol{X}\right) > \zeta_T$}
        \State add $b_{i}$ to the list of estimated change-points
        \State DAIS($s, s_{i}, \lambda_{T}, \zeta_T$)
        \State DAIS($e_{i}, e, \lambda_{T}, \zeta_T$)
    \Else
        \State $i=i+1$
        \If{$i \leq K$}
            \State Go back to \textbf{(Main part)} and repeat
        \Else
        \State STOP
        \EndIf
    \EndIf
\EndIf
\State \textbf{end function}
\end{algorithmic}
\end{algorithm}

Algorithm \ref{alg:DAIS} provides a pseudocode that briefly explains the workflow of DAIS.
The first step is to identify, within the observed data, $X_t$, $t=1,2,\ldots, T$, the location of the largest difference $d_{1, T}$, as defined in \eqref{def_largest_diff}. 
This is the step that makes the algorithm data-adaptive, differentiating it from the competitors, as the algorithm will check different intervals $I_j$ for the same underlying signals (in terms of the type of the change, the length of the sequence and the change-point locations) if they have different observed noise, depending on the value of $d_{1, T}$.
Applying the appropriate transformation (as in \eqref{def_largest_diff} based on the signal structure) to the unobserved true signal $f_t$ in \eqref{model}, will reveal a constant signal at 0, with non-zero values at the locations where the change-points occur, so we expect that the largest difference for the observed data $X_t$ will be at a point near the change-point with the largest magnitude of change. 
A proper explanation of what we consider to be `near' is discussed in Section \ref{discussion} in detail.
After calculating $d_{1, T}$, expanding intervals are used to check for possible change-points. 
The algorithm expands the intervals once from the right and once from the left of the point where the maximum absolute difference occurs at every step by an expansion parameter of magnitude $\lambda_T$. 
By doing this, we identify the location of the change-point that is closest to the detected largest difference.
This expansion is performed for at most $\lceil \frac{T}{\lambda_T} \rceil + 1$ steps in total. 
Upon the detection of a change-point, the algorithm restarts after discarding the data that have already been checked. 
If, for example, the change-point was detected in the subinterval $[s^{\ast}, e^{\ast}]$, then DAIS will be reapplied to $[1, s^{\ast}]$ and $[e^{\ast}, T]$. 
When checking any interval $[s,e]$ the maximum value attained by the contrast function at a point $b$ within the interval, is compared to the chosen threshold $\zeta_T$.
If this value exceeds $\zeta_T$, we take $b$ to be a change-point, for $b\in\{s,s+1,\ldots,e-1\}$ in the case of piecewise-constant signals, and $b\in\{s,s+1,\ldots,e-2\}$ in the case of continuous, piecewise-linear signals. 
If no change-point is detected, then the algorithm terminates. 

As mentioned in the Introduction, the algorithm can be applied to more general signals, for example \sla{continuous}, piecewise-polynomial signals of order $p$.
In this case, the location of the largest difference needs to be generalized.
In the more general case, it can be defined as the value of $t$ for which the \sla{absolute value of the} $(p+1)-$times differentiated signal is maximized.
\sla{Differentiating the signal a $p+1$ amount of times eliminates a degree $p$ polynomial trend (\citep{Akilagun_Fast_Optimal}).}

There are two main advantages to the algorithm due to its data-adaptive nature. 
Firstly, since at each step DAIS starts checking for potential change-points around the largest difference, which as already explained we have reason to believe that it is around the most prominent change-point,
then this change-point is detected fast. 
If the starting point is not near the change-point, detection still occurs with a slower speed.
Secondly, when the change-point is at the location of the largest difference, or close to it, due to the way the expansions occur, the change-point lies in the middle of the intervals being checked when an even number of expansions have been performed. This provides advantages in detection power due to the change-point being detected in some cases in balanced intervals.
Such near balance in the distances of the change-points from the left and right end-points of the intervals under consideration does not necessarily appear in many state-of-the-art competitors, such as Binary Segmentation, ID, NOT, WBS, WBS2 and SeedBS.
Thus, the data-adaptive nature of our algorithm enhances its speed and accuracy in estimating the locations of the change-points.

\subsection{Location of the largest difference and the isolation aspect}\label{discussion}

As already explained in detail in Section \ref{sec: DAIS_algorithm}, for a given data sequence $X_t$, the algorithm calculates the location of the largest difference $d_{s,e}$, as defined in \eqref{def_largest_diff}, in the interval $[s,e]$ under consideration at each step and checks for change-points around it.
For different data sequences, DAIS will start searching for change-points in different neighborhoods of the data, based on where the most prominent change-point (if any) lies. This is exactly the novel data-adaptive nature of the proposed algorithm, compared to the existing algorithms in the literature that exhibit a fixed workflow, not adaptive to the data structure regarding the location of the change-points, if any exist.
\sla{
In this section, we now explain the role of the location of the largest difference in the detection of a change-point in an interval where isolation is guaranteed, a discussion which is related to Assumption (A1) of Section~\ref{piecewise constant}.
}

\sla{
For the rest of this section, we are working under the framework of Assumptions (A2) and (A3), for piecewise-constant and continuous, piecewise-linear signals, respectively.
The assumptions can be found in Sections~\ref{piecewise constant} and \ref{piecewise linear} and need to be satisfied for the change-points to be detected.}
\sla{The detection of a change-point $r_j$, for $j \in \{1,\ldots,N\}, N \in \mathbb{N}$ will certainly occur if both end-points are far enough from the true location of the change-point, but of course it can happen at any smaller interval where the value of the contrast function exceeds the pre-specified threshold $\zeta_T$.}
\sla{We define the smallest distance between two consecutive change-points by $\delta_T$, as in \eqref{eq: def distance between cpts}. 
The constant $n \geq 3/2$, also introduced in this section, acts as a scaling parameter and 
in the worst case scenario, meaning that the interval in which detection occurs is the largest possible, both end-points are at least at a distance $\delta_T/2n$ from the true change-point (as also explained in the proof of Theorem \ref{consistency_theorem} in Appendix \ref{proofs}).}
\sla{Larger values of $n$ imply that the end-points of the intervals will be closer to the change-points when detection occurs.}
In the worst-case scenario, 
\sla{as just explained}, it holds that
\sla{not only} both end-points are at least a distance $\delta_T/2n$ from the change-point, 
\sla{but also} at least one of the end-points will be in one of the intervals:
\begin{equation} \label{intervals_discussion}
    I^L_j = \left(r_j-\frac{\delta_{T}}{n}, r_j-\frac{\delta_{T}}{2n}\right], \quad
    I^R_j = \left[r_j + \frac{\delta_{T}}{2n}, r_j + \frac{\delta_{T}}{n}\right).
\end{equation}
Note that in \cite{fryzlewicz2014wild} and \cite{anastasiou2022detecting} the choice $n=3/2$ was made while in \cite{baranowski2019narrowest} the value of $n=3$ was used when defining similar intervals for the location of the end-points of the interval which allow for detection to occur.
Since the length of the intervals in \eqref{intervals_discussion} is $\delta_T/2n$, the condition $\lambda_T \leq \delta_T/2n$ is required, so that there is always at least one end-point that lies within $I^L_j$ and $I^R_j$,
\slb{where, as previously explained, $\lambda_T$ is the expansion parameter}.
Since $\lambda_T \geq 1$, we implicitly require $\delta_T \geq 2n$ \sla{and, using the lower bound of $n$, $\lambda_T \leq \delta_T/3$}.

In the cases of either $\sla{d_{s,e} \leq r_{1}}$ or $\sla{d_{s,e} \geq r_N}$, it is apparent that any value of $d_{s,e}$ guarantees isolation, as the left and right expanding intervals will not include any other change-point, besides $r_1$ and $r_N$, respectively, before detection occurs.
Considering now the more complicated case, where the location of the largest difference lies between two change-points, define \sla{$\Tilde{\delta}_j$ as in \eqref{eq: def distance between cpts} to be the distance between the consecutive change-points $r_j$ and $r_{j+1}$.}
In addition, $\delta_{s,e}^j = \min\{r_{j+1}-d_{s,e}, d_{s,e}-r_j\}$ for $j=0,1,\ldots,N$.
Without loss of generality, suppose that $d_{s,e} \in \left[ r_J, r_{J+1} \right]$ for some $J\in\{ 1,\ldots,N-1 \}$.
\sla{Requiring}
\begin{align} \label{delta upper bound}
    0 \leq \delta_{s,e}^J \leq \frac{\Tilde{\delta}_J}{2} - \frac{3\delta_T}{4n},
\end{align}
guarantees detection in an interval where the change-point is isolated and also allows for the detection of the neighboring change-points.
This is because the end-point of the interval the algorithm will be applied to after detection of the first change-point (either $r_J$ or $r_{J+1}$), will have a distance of at least $\delta_T/2n$ from all the undetected change-points; see \eqref{proof_of_detection} and \eqref{slope_observed} in the supplementary material for a proof that $\delta_T/2n$ is the minimum distance of the end-points from the change-point that guarantees detection.
Note that the lower bound of $n$ ensures that the right hand side of the second inequality in \eqref{delta upper bound} is non-negative.
\sla{
As is explained later in detail, for any $n>3/2$, the upper bound in \eqref{delta upper bound} increases as $T\rightarrow\infty$.
In the case that $n=3/2$ and $\Tilde{\delta}_J = \delta_T$, the upper bound becomes $0$, which means that $d_{s,e}$ has to lie exactly at the location of the change-point for \eqref{delta upper bound} to be satisfied.
However, this does not mean that the method fails in this case; it only implies that the change-point is not guaranteed to be isolated when it is detected.
More discussion about this is given after Assumption (A1) in Section~\ref{piecewise constant}.} 

\sla{
It is worth mentioning that \eqref{delta upper bound} is derived such that all change-points can be detected in an interval in which they are isolated, for any $\lambda_T\leq \delta_T/2n$.
If we fix $\lambda_T=1$, then in the worst-case scenario, which is for detection to happen when both end-points are necessarily at least at a distance $\delta_T/2n$ from the true change-point location, one of the end-points will be exactly at a distance $\lceil \delta_T/2n\rceil$. 
So, for $\lambda_T=1$, \eqref{delta upper bound} can be relaxed to
\begin{align}\label{delta upper bound relaxed}
    0 \leq \delta_{s,e}^J \leq \frac{\Tilde{\delta}_J}{2} - \Big\lceil\frac{\delta_T}{2n} \Big\rceil.
\end{align}
For the choice $n=3/2$ and $\Tilde{\delta}_J = \delta_T$, using \eqref{delta upper bound relaxed}, $d_{s,e}$ has to be within $\delta_T/2 - \Big\lceil\delta_T/3 \Big\rceil$ of the true location of the change-point.
}

\sla{
From the consistency result as in Theorem~\ref{consistency_theorem}, $\delta_T$ must be at least of order $O(\log T)$. 
This is because $\delta_T$ needs to be larger than the distance between the estimated and true locations of the change-points, which is at most of order $O(\log T)$, as can be seen from \eqref{main result mean}; otherwise we would not be able to match the estimated change-points with the true ones.
Also, by definition, $\Tilde{\delta}_j$ is at least of order $\mathcal{O}(\delta_T)$ for any $j=0,1,\ldots,N$.
This means that as $T\rightarrow\infty$, the upper bound of $\delta_{s,e}^J$ as given in \eqref{delta upper bound relaxed} (and \eqref{delta upper bound} for $n>3/2$) increases.
Additionally,} if $\underline{f}_T$ increases, the signal-to-noise ratio becomes larger and so the probability that the largest difference is exactly at the location of the change-point increases.
\sla{In both cases, the probability that $d_{s,e}$ lies within $\delta_{s,e}^J$ of any change-point increases.}
More details on this, in the case that $\epsilon_t, t=1,2,\ldots,T$ follow the Gaussian distribution, are provided in Section~\ref{discussion prob}.

\sla{We define 
\begin{equation}\label{eq: max_largest_diff}
    d_{s,e}^{\max} = \left\{ \begin{array}{cc}
        T-1, & \textrm{when } f_t \textrm{ is piecewise-constant,} \\
        T - 2, & \textrm{when } f_t \textrm{ is continuous, piecewise-linear,}
    \end{array} \right.,
\end{equation}
to be the maximum value that $d_{s,e}$ can obtain.
In the case of continuous, piecewise polynomial signals of order $p$, it can be extended to $d_{s,e}^{\max} = T-p-1$.}
Considering now any $d_{s,e} \in \{1,2,\ldots,\sla{d_{s,e}^{\max}}\}$ and using \eqref{delta upper bound}, \sla{in order to have detection while the change-point is isolated,} it must hold that
\begin{equation} \label{eq: assumption_largest_diff}
    d_{s,e} \in \left[1, r_1 + \frac{\Tilde{\delta}_1}{2} - \frac{3\delta_T}{4n} \right] 
    \cup 
    \left[r_N - \frac{\Tilde{\delta}_{N-1}}{2} + \frac{3\delta_T}{4n}, \sla{d_{s,e}^{\max}} \right] \cup \left\{ 
    \bigcup_{j=2}^{N-1} B_{n;j}  \right\}
    \end{equation}
where
\begin{equation}\label{eq: interval_lergest_diff}
    B_{n;j} = \left[r_j - \frac{\Tilde{\delta}_{j-1}}{2} + \frac{3\delta_T}{4n}, r_j + \frac{\Tilde{\delta}_j}{2} - \frac{3\delta_T}{4n} \right]
\end{equation}
\sla{and the union over $j$ in \eqref{eq: assumption_largest_diff} is the empty set for $N\leq2$.}
Note that the length of the interval where $d_{s,e}$ cannot lie under this assumption, \sla{at each step of the algorithm,} is 
\sla{$2\lceil 3\delta_T/4n \rceil - 1$}
around the midpoint of every pair of consecutive change-points.
\sla{Similar intervals to \eqref{eq: assumption_largest_diff} can be defined using \eqref{delta upper bound relaxed}, but we focus on the general $\lambda_T\leq\delta_T/2n$ case.}

\sla{
Now, let us consider the case where the signal-to-noise ratio is so small that it is impossible for the location of the largest difference to be informative (near anyone of the true change-points); in such cases, and based on the noise, the largest difference could be anywhere in the set $\{1, 2, \ldots d_{s,e}^{{\rm max}}\}$, with $d_{s,e}^{{\rm max}}$ as in \eqref{eq: max_largest_diff}. Assumptions (A2) or (A3), which can be found in Sections ~\ref{piecewise constant} and \ref{piecewise linear}, respectively, and are needed for detection of the change-points to be possible, are considered to be satisfied. 
In this case,
the probability that the location of the largest difference does not lie in \eqref{eq: assumption_largest_diff} is 
\slb{$(N-1) \left(2\lceil \frac{3\delta_T}{4n} \rceil - 1\right)/d_{s,e}^{\max}
$}.
Therefore, even in the worst case that the signal-to-noise ratio is such that the location of the largest difference is uninformative, as long as $N\delta_T$ is of order less that $O(T)$, the probability of the result in \eqref{eq: assumption_largest_diff} will still go to 1 as $T \to \infty$. This is a strong indication that the requirement expressed in \eqref{eq: assumption_largest_diff}, which is, in fact, Assumption (A1) provided in Section \ref{theory}, is not restrictive and it only used to have guaranteed isolation and detection at the same time.}

We highlight that in the specific case of $\epsilon_t$ in \eqref{model} being i.i.d. random variables following the $\mathcal{N}(0,1)$, a discussion on the calculation of the exact probability that the largest difference is at a location that guarantees detection in an interval where the change-point is isolated, for all change-points, can be found in Section \ref{discussion prob}.

\section{Theory} \label{theory}
DAIS requires knowledge of the standard deviation, $\sigma$, of the data generation mechanism in \eqref{model}. 
If $\sigma$ is unknown, it can be estimated using the Median Absolute Deviation method (MAD), with the estimator defined as $\hat\sigma := 1.4826\times \text{median} (\left| \boldsymbol{x} - \text{median}(\boldsymbol{x}) \right|)$ for $\boldsymbol{x} = (x_1, x_2,\dots x_T)$, as proposed by \cite{Hampel}. \sla{The obtained estimator has a simple explicit formula, needs little computation time, is very robust with a bounded influence function and a 50\% breakdown point, and, in the case of Gaussianity is a consistent estimator of the population standard deviation (\citep{doi:10.1080/01621459.1993.10476408}).}
\sla{Due to the expected presence of change-points in the data, we use this estimator on the first and second order differenced data, for piecewise-constant and continuous, piecewise-linear signals, respectively; see for example \cite{baranowski2019narrowest} for a similar approach.
The exact expressions for the piecewise-constant and continuous, piecewise-linear cases are provided in Sections~\ref{piecewise constant} and \ref{piecewise linear}, respectively. }
\sla{We take $\sigma$ to be known} and, without loss of generality, we let $\sigma=1$. The model becomes
\begin{equation} \label{model_sigma1}
    X_t = f_t + \epsilon_t, \quad t = 1, 2, \ldots, T,
\end{equation}
\sla{where $\epsilon_t, 
t= 1, 2, \ldots, T$ are i.i.d. random variables following the $\mathcal{N}(0,1)$ distribution. 
This assumption is commonly used in the literature to prove theoretical consistency results.}


\sla{Before proceeding, we introduce some notation that is used throughout the paper.
With $f_t$ the underlying signal as in \eqref{model}, then for $j=1,2,\ldots,N$, the minimum magnitude of change in the data sequence is defined as
\begin{equation} \label{eq: def minimum magnitude of change}
    \underline{f}_T = \min_j \Delta^f_j,
\end{equation}
where
\begin{equation} \label{eq: def magnitude of change}
    \Delta^f_j = 
    \left\{
    \begin{array}{ll}
        \left|f_{r_j+1}-f_{r_j}\right|, & \textrm{when } f_t \textrm{ is piecewise-constant,} \\
        \\
        \lvert 2f_{r_j} - f_{r_j+1} - f_{r_j-1} \rvert, & \textrm{when } f_t \textrm{ is continuous, piecewise-linear.}
    \end{array}
    \right.
\end{equation}
Furthermore, for $r_0=0$ and $r_{N+1}=T$, then the minimum distance between consecutive change-points is defined as
\begin{equation}\label{eq: def distance between cpts}
    \delta_T = \min_{j=0,\ldots,N} \Tilde{\delta}_j, 
    \textrm{ where } \Tilde{\delta}_j = r_{j+1} - r_{j}.
\end{equation}
}

In Sections \ref{piecewise constant} and \ref{piecewise linear} we provide the theoretical results on the consistency of the algorithm for the cases of piecewise-constant and continuous, piecewise-linear signals, respectively. 
The proofs of Theorems \ref{consistency_theorem} and \ref{consistency_theorem_slope} can be found in Appendix \ref{proofs} and the supplementary material, respectively.
In Section~\ref{sec: multivariate_extension} we also provide theoretical results for the consistency in the case of multivariate signals and the outline of the proof can be found in Appendix~\ref{appendix: proof_multivariate}.
In Section \ref{discussion prob}, we return to the discussion of Section \ref{discussion} and present some results in the case when \eqref{eq: assumption_largest_diff} holds, and $\epsilon_t$ in \eqref{model} are i.i.d. from the $\mathcal{N}(0,1)$ distribution.

\subsection{Piecewise-constant signals} \label{piecewise constant}
Under piecewise-constancy, $f_t = \mu_j$ for $t = r_{j-1}+1, \ldots, r_j$ and $f_{r_j} \neq f_{r_j+1}$. \sla{The estimator of the standard deviation used is $\hat\sigma_C := 1.4826\times \text{median} (\left| \boldsymbol{Y} - \text{median}(\boldsymbol{Y}) \right|)/\sqrt{2}$, where $\boldsymbol{Y} = (X_2-X_1, \ldots, X_T - X_{T-1})$}. 
The statistic used as a contrast function in this case is the absolute value of the CUSUM statistic, which is defined as
\begin{equation}\label{CUSUM}
    C_{s,e}^b\left(\boldsymbol{X}\right) = \left| \sqrt{\frac{e-b}{\ell(b-s+1)}}\sum^b_{t=s}X_t - \sqrt{\frac{b-s+1}{\ell(e-b)}}\sum^{e}_{t=b+1}X_t \right|,
\end{equation}
where $s \leq b < e$ and $\ell=e-s+1$. 
The value of $C_{s,e}^b\left(\boldsymbol{X}\right)$ is small if $b$ is not a change-point and large otherwise.

To prove the consistency of our method in accurately estimating the number and locations of the estimated change-points, we work under the following assumptions: 
\begin{enumerate}
    \item[{(A1)}] The location of the largest difference $d_{s,e}$ as defined in \eqref{def_largest_diff}, satisfies \eqref{eq: assumption_largest_diff} every time the algorithm is applied to an interval $[s,e]$ which contains at least one change-point.

    \item[{(A2)}] The minimum distance, $\delta_T$, between two successive change-points, and the minimum magnitude of jumps 
    \sla{$\underline{f}_T$, as defined in \eqref{eq: def distance between cpts} and \eqref{eq: def minimum magnitude of change}, respectively,}
    are connected by $\sqrt{\delta_T}\underline{f}_T \geq \underline{C} \sqrt{\log T}$ for a large enough constant $\underline{C}$.
\end{enumerate}

The first assumption is a formal statement of what is discussed in Section \ref{discussion} and is required to ensure that every change-point can be detected in an interval where isolation is guaranteed.
It is not a strict assumption, for more details on this see Sections \ref{discussion} and \ref{discussion prob}.
The exact probability of the consecutive pairwise difference being larger in absolute value at a change-point location compared to a point where no change occurs, in the case where $\epsilon_t$ are assumed to be Gaussian, is provided in Section \ref{discussion prob}. 
(A2) is a typical assumption in the literature and characterizes the difficulty of the detection problem. It is worth mentioning that if $\delta_T$ is of order greater than $O(\log T)$, then $\underline{f}_T$ is allowed to decrease as $T$ increases. Further discussion for (A2) is given after Theorem~\ref{consistency_theorem}.
\sla{Assumptions (A1) and (A2), together, guarantee that all change-points will be detected in intervals in which they are isolated.
However, even when (A1) fails, the change-points can still potentially get detected in practice. 
This can happen in one of two ways. 
The first is that detection occurs in an interval in which the change-point is isolated, whose end-points (one or both) are at a distance of less than $\delta_T/2n$ from the true location of that change-point. 
The second way is in an interval in which the change-point is not isolated. 
In both cases the detection power, which is based on the value of the contrast function, is reduced, but change-points can still be detected.
The detection power is maximized when the change-point is isolated, in the midpoint of the interval under consideration with end-points far from it, and this is what we attempt to achieve using Assumption (A1).}

Below, we provide the relevant theorem for the consistency of DAIS in accurately estimating the true number and the locations of the change-points in the case of piecewise-constant signals.
\sla{The proof of the theorem can be found in Appendix~\ref{proofs}.}

\begin{theorem}\label{consistency_theorem}
Let $\{X_t\}_{t=1,2,\ldots, T}$ follow model \eqref{model_sigma1}, with $f_t$ being a piecewise-constant signal and assume that the random sequence $\{\epsilon_t\}_{t=1,2,\ldots, T}$ is independent and identically distributed (i.i.d.) from the normal distribution with mean zero and variance one and also that Assumptions (A1) and (A2) hold. 
Let $N$ and $r_j, j=1, 2, \ldots N$ be the number and location of the change-points, while $\hat{N}$ and $\hat{r}_j, j=1, 2, \ldots \hat{N}$ their estimates, sorted in increasing order. 
Then, with $\Delta_j^f$, $\underline{f}_T$, and $\delta_T$ as in \eqref{eq: def magnitude of change}, \eqref{eq: def minimum magnitude of change}, and \eqref{eq: def distance between cpts}, respectively, there exist positive constants $C_1, C_2, C_3, C_4$ independent of $T$ such that for $C_1 \sqrt{\log T} \leq \zeta_{T} < C_2 \sqrt{\delta_T}\underline{f}_T$ and for sufficiently large $T$, we obtain
\begin{equation} \label{main result mean}
\mathbb{P}\Biggl( \hat{N} = N, \max_{j=1, 2, \ldots, N} \biggl( \left| \hat{r}_j - r_j \right| \left( \Delta^{f}_j \right)^{2} \biggl) 
\leq C_{3} \log T \Biggl) \geq 1-\frac{C_{4}}{T}.
\end{equation}
\end{theorem}

The lower bound of the probability in \eqref{main result mean} is $1-\mathcal{O}(1/T)$.
From Theorem \ref{consistency_theorem}, we can conclude that in order to be able to match the estimated change-points with the true ones, $\delta_T$ needs to be larger than $\max_{j=1, 2, \ldots, N} \left| \hat{r}_j - r_j \right|$, which means that $\delta_T$ must be at least of order $\mathcal{O}(\log T)$. 
Assumption (A2) ensures that the rate attained for $\delta_T \underline{f}_T^2$, which characterizes the complexity of the problem, is $\mathcal{O}(\log T)$ and, as \cite{10.2307/24310529} argues, the lowest possible $\delta_T \underline{f}_T^2$ that allows detection of the change-points, in the case of piecewise-constant signals is $\mathcal{O}(\log T - \log\log T)$.
Therefore, the rate obtained by DAIS is optimal up to a rather negligible double logarithmic factor. 
As mentioned in \cite{yu2020review}, the rate we obtain has been established in various different papers.
Rates for other methods which are the same or comparable to ours can be found in \cite{CHO202476}.

\subsection{Continuous, piecewise-linear signals} \label{piecewise linear}
For the case of continuous, piecewise-linear signals,  $f_t = \mu_{j,1} + \mu_{j,2}t$ for $t = r_{j-1}+1, \ldots, r_j$ and $f_{r_j-1} + f_{r_j+1} \neq 2f_{r_j}$ with the additional constraint $\mu_{j,1}+\mu_{j,2}r_j = \mu_{j+1,1}+\mu_{j+1,2}r_j$ so that the signal is continuous.
\sla{The estimator of the standard deviation used is $\hat\sigma_L := 1.4826 \times \text{median} (\left| \boldsymbol{Y} - \text{median}(\boldsymbol{Y}) \right|)/\sqrt{6}$, where $\boldsymbol{Y} = (X_1-2X_2+X_3, \ldots, X_{T-2} - 2X_{T-1}+X_{T})$}. 
Through a log-likelihood approach, \cite{baranowski2019narrowest} shows that contrast function is $C_{s,e}^b\left(\boldsymbol{X}\right) = \left| \langle \boldsymbol{X}, \boldsymbol{\phi^b_{s,e}}\rangle\right|$, where the contrast vector $\boldsymbol{\phi^b_{s,e}} = (\phi^b_{s,e}(1),\dots, \phi^b_{s,e}(T))$ is given by
\begin{equation} \label{phi_definition}
    \phi^b_{s,e}(t)=\left\{\begin{array}{lll}
    \alpha^b_{s,e}\beta^b_{s,e} \big[(e+2b-3s+2)t - (be+bs-2s^2+2s)\big], & t \in\{s, \dots, b\}, \\
    -\frac{\alpha^b_{s,e}}{\beta^b_{s,e}} \big[(3e-2b-s+2)t - (2e^2+2e-be-bs)\big], & \hspace{-0.4cm} t \in\{b+1, \ldots, e\},\\
    0, & \text{otherwise,}
\end{array}\right.
\end{equation}
where $\ell=e-s+1$ and
\begin{align*}
    & \alpha^b_{s,e} = \Biggl(\frac{6}{\ell(\ell^2-1)(1+(e-b+1)(b-s+1)+(e-b)(b-s))}\Biggr)^{\frac{1}{2}}, \\
    & \beta^b_{s,e} = \Biggl( \frac{(e-b+1)(e-b)} {(b-s+1)(b-s)} \Biggr) ^ \frac{1}{2}.
\end{align*} 
Under the Gaussianity assumption, $C_{s,e}^b\left(\boldsymbol{X}\right)$, as defined above, is maximized at the same point as the generalized log-likelihood ratio for all possible single change-points within $[s,e)$.

Before presenting the theoretical result for the consistency of the number and location of the estimated change-points by DAIS in the case of continuous, piecewise-linear signals, we require the following assumption, which is equivalent to (A2) in Section \ref{piecewise constant}, but now under the current setting of changes in the slope.

\begin{enumerate}
    \item[{(A3)}] The minimum distance, $\delta_T$, between two successive change-points and the minimum magnitude of jumps \sla{$\underline{f}_T$, as defined in \eqref{eq: def distance between cpts} and \eqref{eq: def minimum magnitude of change}, respecitvely,}
    are connected by $\delta_T ^ {3/2}\underline{f}_T \geq C^{\ast} \sqrt{\log T}$ for a large enough constant $ C^{\ast}$.
\end{enumerate}

\begin{theorem}\label{consistency_theorem_slope}
Let $\{X_t\}_{t=1,2,\ldots, T}$ follow model \eqref{model_sigma1}, with $f_t$ being a continuous, piecewise-linear signal and assume that the random sequence $\{\epsilon_t\}_{t=1,2,\ldots, T}$ is independent and identically distributed (i.i.d.) from the normal distribution with mean zero and variance one and also that (A1) and (A3) hold. 
Let $N$ and $r_j, j=1, 2, \ldots N$ be the number and location of the change-points, while $\hat{N}$ and $\hat{r}_j, j=1, 2, \ldots \hat{N}$ are their estimates, sorted in increasing order. 
In addition, $\Delta^f_j$, \sla{as defined in \eqref{eq: def magnitude of change}} is the magnitude of \sla{each change in the slope of} 
$f_t$, \sla{and} $\underline{f}_T$ and $\delta_T$ \sla{are defined in \eqref{eq: def minimum magnitude of change} and \eqref{eq: def distance between cpts}, respectively}. 
Then there exist positive constants $C_1, C_2, C_3, C_4$ independent of $T$ such that for $C_1 \sqrt{\log T} \leq \zeta_{T} < C_2 \delta_T^{3/2}\underline{f}_T$ and for sufficiently large $T$, we obtain
\begin{equation} \label{main result slope}
\mathbb{P}\Biggl( \hat{N} = N, \max_{j=1, 2, \ldots, N} \biggl( \left| \hat{r}_j - r_j \right| \left( \Delta^{f}_j \right)^{2/3} \biggl) \leq C_{3} (\log T)^{1/3} \Biggl) \geq 1-\frac{C_{4}}{T}.
\end{equation}
\end{theorem}

The proof of Theorem \ref{consistency_theorem_slope} can be found in the supplementary material.
Again, the lower bound of the probability is $1-\mathcal{O}(1/T)$ and the difficulty of the problem, $\delta_T ^ {3/2} \underline{f}_T$ is analogous to $\sqrt{\delta_T}\underline{f}_T$, which appears in (A2) for the case of piecewise-constant signals.
For $\underline{f}_T \sim T^{1/2}$, the change-point detection accuracy is $O\left( T^{-1/3} \left( \log T \right)^{1/3} \right)$, which differs from the optimal rate $O(T^{-1/3})$ derived by \cite{Raimondo1998} only by the logarithmic factor.
This means that the change-point detection rate in DAIS is almost optimal.

\subsection{Probability of guaranteed isolation of the change-point at detection} \label{discussion prob}

\sla{In this section, we work under Assumptions (A2) and (A3), for piecewise-constant and continuous, piecewise-linear signals, respectively.}
We return to the discussion of Section \ref{discussion}, that concerns the connection between the location of the largest difference and the guarantee that the detection of the change-point closest to it occurs in an interval in which it is isolated.
The probability of this happening is given by
\begin{equation}\label{eq: prob of isolation}
    \Prob \left( d_{s,e} \in \left[1, r_1 + \frac{\Tilde{\delta}_1}{2} - \frac{3\delta_T}{4n} \right] 
    \cup 
    \left[r_N - \frac{\Tilde{\delta}_{N-1}}{2} + \frac{3\delta_T}{4n}, \sla{d_{s,e}^{\max}} \right] 
    \cup 
    \left\{ \bigcup_{j=2}^{N-1} B_{n;j} \right\} \right)
\end{equation}
where $B_{n;j}$ is defined in \eqref{eq: interval_lergest_diff} \sla{and the union over $j$ is the empty set for $N\leq2$}.
Under the assumption of $\epsilon_t$ being i.i.d. random variables following the $\mathcal{N}(0,1)$ distribution, and in the case that $f_t$ is piecewise-constant, we note that in order to calculate the probability that the largest difference is exactly at the location of a change-point, we need to calculate the following:
\begin{equation} \label{probability}
    \Prob\Biggl(\bigcup_{j=1}^{N} \bigcap_{t=1}^{T-1}|X_{r_{j}+1} - X_{r_{j}}| \geq |X_{t+1} - X_t| \Biggr).
\end{equation}
It is easy to see that $\eqref{eq: prob of isolation} \geq \eqref{probability}$, \sla{as \eqref{probability} is the probability of $d_{s,e}$ being exactly at the location of any change-point, while \eqref{eq: prob of isolation} in a small area around any change-point}.
As a first step, we are interested in the probability that the largest difference occurs at the location of a specific change-point, namely $r_J$, for $J \in \{1,\ldots,N\}$. 
We will work on
\begin{align} \label{reduced prob discussion}
    \Prob\Biggl(|X_{r_{J}+1} - X_{r_{J}}| \geq |X_{t+1} - X_t| \Biggr),
\end{align}
which is the probability that the observations \sla{around} the change-point $r_J$ have a larger absolute difference than a pair of observations in the signal where no change occurs. 
\sla{The focus is put on $t \notin \{r_J-1, r_J, r_J+1\}$ such that $|X_{r_{J}+1} - X_{r_{J}}|$ and $|X_{t+1} - X_t|$ are independent.}
Using \eqref{model}, the fact that $\epsilon_t\sim\mathcal{N}(0,1)$, and that $f_{t'+1}-f_{t'} = 0$ for $t' \neq r_J$, \eqref{reduced prob discussion} can be written as
\begin{align} \label{reduced prob discussion 2}
    \Prob\Biggl(|X_{r_{J}+1} - X_{r_{J}}| \geq |X_{t+1} - X_t| \Biggr) & =  \Prob\Biggl(|f_{r_{J}+1} - f_{r_{J}} + \sigma Z_2| \geq | \sigma \Tilde{Z}_2| \Biggr) \\
    & =  \Prob\Biggl(\left|\sla{\frac{f_{r_{J}+1} - f_{r_{J}}}{\sigma}} + Z_2 \right| \geq |\Tilde{Z}_2 | \Biggr) \nonumber
\end{align}
where 
$Z_2, \Tilde{Z}_2\sim\mathcal{N}(0,2)$ are independent. 
In the above equation, $|\sla{ \left( f_{r_{J}+1} - f_{r_{J}} \right) / \sigma} + Z_2|$ follows the Folded Normal (FN) distribution with mean $\sla{ \left( f_{r_{J}+1} - f_{r_{J}} \right) / \sigma}$ and variance 2.
Similarly, $|\Tilde{Z}_2|$, follows the same distribution with mean 0 and variance 2, which reduces to the simpler case of the Half Normal (HN) distribution with variance 2.
So, \eqref{reduced prob discussion 2} can be equivalently written as $\Prob(Y_1 \geq Y_2)$
for $Y_1 \sim FN(\sla{ \left( f_{r_{J}+1} - f_{r_{J}} \right) / \sigma}, 2)$ and $Y_2 \sim HN(0, 2)$. 
Denoting by $f_{Y_1}$ the probability density function of $Y_1$ and $F_{Y_2}$ the cumulative distribution function of $Y_2$, \eqref{reduced prob discussion 2} can be written as
\begin{align*}
    \Prob(Y_1 \geq Y_2)
    & = \int_0^{\infty} \Prob(Y_2 \leq y) f_{Y_1}(y) dy
    = \int_0^{\infty} F_{Y_2}(y) f_{Y_1}(y) dy \\
    & = \ \mathbb{E}_{Y_1}\Bigl[F_{Y_2}(Y_1)\Bigr] 
    = \ \mathbb{E}_{Y_1} \Bigl[\textrm{erf}\Bigl(\frac{Y_1}{2}\Bigr)\Bigr]
\end{align*}
where $\textrm{erf}(z)=\frac{2}{\sqrt{\pi}}\int_0^z e^{-t^2} \rm{d}t$ is the Gauss error function. 
Using \cite{ng1969table} (equation (38) on p.10), we can calculate the above
as
\begin{equation}\label{eq: prob_isolation_final}
    \Prob(Y_1 \geq Y_2) = \frac{1}{2} \biggl[ 1 + \textrm{erf} \left( \frac{\sla{ f_{r_{J}+1} - f_{r_{J}} }}{2\sqrt{2} \ \sla{\sigma}}\right)^2 \biggr],
\end{equation} 
so we have an explicit expression for \eqref{reduced prob discussion 2}.
It is easy to see now that this probability goes to 1 as $\sla{ \left( f_{r_{J}+1} - f_{r_{J}} \right)}$ gets larger.
However, calculating or bounding from below the probability in \eqref{probability}, and subsequently in \eqref{eq: prob of isolation}, are not as straightforward to handle tasks due to the fact that $|X_{t}-X_{t-1}|$ and $|X_{t-1}-X_{t-2}|$ are dependent. 
We provide results based on simulations.

We want to note that when $N=1$, the interval in Assumption (A1) becomes $[1,\sla{d_{s,e}^{\max}}]$, which implies that isolation of the change-point in the interval in which detection occurs is guaranteed for all possible values {of the location} of the largest difference.
Also, in a data sequence with $N\geq 2$ change-points, we expect the order that they are detected to be with decreasing magnitude of change.
This means that the change-point with the smallest magnitude of change will be detected last, in an interval in which it is certainly isolated, so the second smallest magnitude of change in the data sequence is the one that needs to be `large enough' to ensure that all change-points will be isolated in an interval where detection can occur. 
\sla{For this reason, the following simulation study considers data sequences with two change-points.}

\sla{Before investigating \eqref{eq: prob of isolation} further, it is worth emphasizing that for any location of the largest difference, isolation of the change-point is guaranteed for at least one interval if $\lambda_T = 1$.
However, as explained in Section~\ref{discussion}, in order to guarantee that detection will occur, the endpoints of the interval have to be at a distance of at least $\frac{\delta_T}{2n}$ from the change-point, and Assumptions (A2) or (A3) need to be satisfied.}

\sla{
For the simulation study, we focus on piecewise-constant signals. 
More specifically, we consider a data sequence with  $\sigma = 1$ and two change-points at locations $r_1 = 10 + \lfloor \log T \rfloor$ and $r_2 = T - 10 - \lfloor \log T \rfloor$, with equal magnitudes of change, which are chosen to be equal to $\sigma$.
These values ensure that Assumption (A2) is satisfied and are also aligned with the conclusion of the consistency theorem, Theorem~\ref{consistency_theorem}, which states that $\delta_T$ has to be at least of order $O(\log T)$ in order to be able to match the estimated change-points with the true ones.
We use $n=3/2$, which is the most unfavorable choice, as it minimizes the total length of the union of the intervals in \eqref{eq: prob of isolation}, which becomes $\left[1, \frac{T}{2} - \delta_T/2 \right] \cup \left[ \frac{T}{2} + \delta_T/2, T-1 \right]$.
In Table~\ref{table:simulated_prob}, the proportion of times that $d_{1,T}$ lies in this union of intervals is reported as $T$ increases, using 10,000 repetitions for each value of $T$.
The results show that the probability of the largest difference occurring at a point that guarantees isolation in the interval where the change-point will be detected increases with the length of the signal. 
This is expected in this example, as the length of the interval in which $d_{1,T}$ should not lie is equal to $\delta_T$, and therefore, its proportion to the length, $T-1$, of the whole differenced sequence goes to 0.
}


\sla{
\begin{table}[htbp]
\caption{Proportion of times that the location of the largest difference was in an interval where detection while the change-points are isolated is guaranteed.}
\label{table:simulated_prob}
\centering
\begin{tabular}{rccrccrr}
  \hline
    T &&& $\delta_T$ &&& proportion \\ 
  \hline
  60 &&& 14 &&& 0.762 \\
80 &&& 14 &&& 0.818 \\
100 &&& 14 &&& 0.850 \\
150 &&& 15 &&& 0.905 \\
200 &&& 15 &&& 0.928 \\
250 &&& 15 &&& 0.942 \\
500 &&& 16 &&& 0.966 \\
750 &&& 16 &&& 0.976 \\
1000 &&& 16 &&& 0.982 \\
1500 &&& 17 &&& 0.990 \\
   \hline
\end{tabular}
\end{table}
}

\section{Computational complexity and practicalities} \label{Computational complexity and practicalities}

\subsection{Computational complexity}\label{computational complexity}

We will now explain why DAIS needs to check at most $\lceil \frac{T} {\lambda_T} \rceil + 2N + 1$ intervals before the algorithm is terminated, where, as explained before, $N$ is the total number of change-points of the signal, $T$ is the length of the data sequence and $\lambda_T$ is the chosen expansion parameter.
If there are no change-points, the algorithm will check at most $\lceil \frac{T} {\lambda_T} \rceil + 1$ left- and right-expanding intervals around the location of the largest difference, as explained in Section \ref{methodology}.
In the case that the signal has change-points, since at every time that detection occurs the already used observations are discarded and the algorithm is restarted on two disjoint subintervals, then the maximum total number of expansions required increases by two (one for each of the newly created subintervals). 
Therefore, the maximum number of iterations depends on the number of change-points and DAIS needs to check at most $K' = \lceil \frac{T} {\lambda_T} \rceil + 2N + 1$ intervals. 
Since we take $\lambda_T < \delta_T$, then $K' > \lceil T/\delta_T\rceil + 2N + 1$.
Combining this with the definition of $\delta_T$, which implies that $N\leq T/\delta_T$, it can be concluded that the lower bound of $K'$ is $\mathcal{O}(T/\delta_T)$. 
As explained by \cite{anastasiou2022detecting}, the lower bound for the maximum number of intervals that need to be checked in ID is $\mathcal{O}(T/\delta_T)$ and in WBS and NOT it is $\mathcal{O}(T^2/\delta_T^2)$ up to a logarithmic factor, which means that DAIS is faster than the latter two algorithms and is of the same order as ID. 

Now, focusing more on the algorithms ID and DAIS, we need to highlight that
the data-adaptive nature of DAIS leads to computational advancements. 
More specifically, 
due to the methodology of ID, when the data sequence under consideration has a small number of true change-points, the number of intervals that need to be checked will be close to the maximum possible for ID, which is $2\lceil \frac{T} {\lambda_T} \rceil$.
In contrast, DAIS needs to check only $K'$ intervals, which will be much smaller compared to ID.
For example, when the signal has no change-points, using the same expansion parameter $\lambda_T$, DAIS checks almost half as many intervals as ID, which leads to substantial computational gain.
This can be seen in the simulation results of Signals {(S12)} and {(S13)} in Table~\ref{supp:table:justnoise, long signal}, which have 0 and 1 change-points, respectively, when compared with ID\_th, which uses the same $\lambda_T$ as DAIS.
When many change-points are present, 
using again the definition of $\delta_T$ and the inequality $\delta_T \geq 2n\lambda_T$, for which the explanation is in Section \ref{discussion}, we get
    $N \leq T/(2n\lambda_T) < T/(2\lambda_T) - 1.$
So, $K'$ will always be smaller than the maximum number of intervals that are checked by ID.
However, it needs to be noted that the data-adaptivity comes at an additional, but small, computational cost, since, at each step of the algorithm, the location of the largest difference, in absolute value, needs to be re-calculated.

\subsection{Parameter selection} \label{parameters}
\textbf{Choice of the threshold $\zeta_T$}:
In Theorems \ref{consistency_theorem} and \ref{consistency_theorem_slope}, the rate of the lower bound of the threshold $\zeta_T$ is $\mathcal{O}(\sqrt{\log T})$ and so we use
\begin{equation} \label{def_threshold}
    \zeta_T = C\sqrt{\log T}.
\end{equation}
For the choice of the positive threshold constant $C$, we ran an extensive simulation study using various signal structures and Gaussian noise. 
The best behavior occurred for $C=1.7$ in the case of piecewise-constant signals and $C=2.1$ in the case of piecewise-linear signals. 
These are the constants that were used for the simulations in Section \ref{simulations}. 
\sla{However, our algorithm performs well for a range of values for $C$.
Simulation results for $C\in\{1.5,1.6,1.7,1.8,1.9\}$ in the case of piecewise-constant and $C\in\{1.9,2,2.1,2.2,2.3\}$ in the case of continuous, piecewise-linear signals are provided in the supplementary material in \ref{sec: impact_params} in Tables \ref{tab: thr_const_const} and \ref{tab: thr_const_lin}.}
\cite{anastasiou2022detecting} proves in Corollary 1 that, in the case of piecewise-constant signals, as $T \rightarrow \infty$, the threshold can be taken to be at most $\sqrt{3\log T}$; our choice of 1.7 for the threshold constant, $C$ as in \eqref{def_threshold}, does not violate this result.
\newline
\\
\textbf{Choice of the expanding parameter $\lambda_{T}$}:
As can be seen by the proofs of Theorems \ref{consistency_theorem} and \ref{consistency_theorem_slope}, in a given signal, detection will occur for any $\lambda_T \leq \delta_T/2n$, where \slb{$n\geq 3/2$ is a constant}.
Using $\lambda_T=1$ ensures that each change-point is isolated for as many intervals as possible, but this increases the computational time.
In practice, we use $\lambda_T=3$, which is small enough to have very good accuracy in the detection of the change-points, but not too small which could magnify the computational cost.
\sla{Our algorithm is quite robust to changes in the value of $\lambda_T$; for a small simulation study in the case that $\lambda_T \in \{1,3, 5, 10, 15, 20\}$, see Table~\ref{tab: increasing_lambda}.
If the value of $\lambda_T$ is large, then there is a risk of underestimating the number of change-points in cases where the distance between change-points is smaller than $\lambda_T$.}
\newline

\section{Simulations}\label{simulations}
This section compares the performance of DAIS with state-of-the-art competitors.
We only focus on algorithms that estimate both the number and the locations of the change-points in the given univariate signal in an offline manner.
The competitors used in the simulation study can be found in Table \ref{table:competitors}, along with the relevant \textsf{R} packages.
For ID, WBS and SeedBS we report the results for both the Information Criterion (IC) and the thresholding stopping rules. The notation is ID\_ic, WBS\_ic, SeedBS\_ic and ID\_th, WBS\_th, SeedBS\_th, respectively. {\sla{Regarding the threshold constant, the default value $C = 1.15\sqrt{2}$, as in the {\textsf{breakfast}} {\textsf{R}} package (\citep{breakfast_package}), was taken in the piecewise-constant case, while $C = 1.4\sqrt{2}$ was used for ID in the continuous, piecewise-linear case)}}. 
{\sla{We highlight, that even though for SeedBS we used code (that is partly based on the {\textsf{wbs}} {\textsf{R}} package (\citep{wbs_package})) from the relevant GitHub repository, we decided, for a fair comparison, to keep the well-calibrated threshold constant of the {\textsf{breakfast}} {\textsf{R}} package. This choice led to a substantially better performance for SeedBS compared to results that were relied on the constant used in the relevant GitHub code.}} \sla{Furthermore, for the NOT method of \cite{baranowski2019narrowest} results are presented based on the Schwarz information
criterion, while for the WBS2 method of \cite{fryzlewicz2020detecting} results are provided when the algorithm is combined with the Steepest Drop to Low Levels (SDLL) model selection procedure, introduced also in \cite{fryzlewicz2020detecting}.}
In order to allow the competitors achieve their best performance, in some cases the inputs of the relevant functions, as developed in \textsf{R}, were adjusted.
The competitors ID, WBS, WBS2, PELT and MARS were applied using the default values based on the relevant {\textsf{R}} packages provided in Table \ref{table:competitors}. For DP\_univar, the values proposed in the examples of the corresponding \textsf{R} function are used.
For MOSUM, the function multiscale.localPrune was used, with the default values for all the parameters apart from the minimal mutual distance of change points that can be expected (the d.min input argument) which was set equal to 3 instead of 10. This enhances MOSUM's performance compared to the default d.min = 10 value because in our simulation study, we use signals with $\delta_T < 10$.
The maximum number of features allowed, $K_{\rm max}$, is required for the function that employs the Information Criterion model selection procedure within the {\textsf{breakfast}} {\textsf{R}} package. Therefore, for the NOT, WBS\_ic and ID\_ic methods we needed to choose $K_{\rm max}$. The default value in the relevant {\textsf{R}} function is $K_{\rm max} = \min\{100, T/\log(T)\}$. In signals, where the true number of change-points surpasses this default value, we set $K_{\rm max} = \lceil T/\delta_T\rceil$, where $\delta_T$ is the minimum distance between consecutive change-points in the signal tested.
For MSCP, the value of the minimal window considered is chosen to be $\min\{\lceil\delta_T/2\rceil, 20\}$, where 20 is the default value. 
\sla{
SMUCE requires Monte Carlo simulations for calculating $q$, the vector of critical values at level $\alpha$, which is the asymptotic probability of overestimating the number of change points. 
The time for these simulations, which can be quite high, is not included in the values reported below and we use the default choice, which is $\alpha = 0.5$.}
\sla{\cite{Liehrmann22112024} proposes an implementation for MsFPOP and MsPELT, which are extensions of FPOP and PELT to the multiscale penalty proposed by \cite{Verzelen_optimal}.
As the results for the number and the locations of the estimated change-points are the same (they only differ at their computational times), we only report the results of MsFPOP.
The calibration constants are set to the values chosen by simulations in \cite{Liehrmann22112024} such that the percentage of the false positives is less than 0.05.}
The estimation of the standard deviation for CPOP is set to be done using the MAD method as described in Section \ref{theory}. 
TF, SeedBS \sla{and MsFPOP} are employed based on the implementations found in \url{https://github.com/hadley/l1tf}, \url{https://github.com/kovacssolt/SeedBinSeg/blob/master/SeedBS.R} \sla{and \url{https://github.com/aLiehrmann/MsFPOP}}, respectively.
Change-point detection using local\_poly involves performing cross validation for every time series, for the optimal values to be chosen. The degree of the polynomial is set to 1, so that the fitted signals are piecewise-linear.
Code on the simulations, as well as on how to implement DAIS, can be found in \url{https://github.com/Sophia-Loizidou/DAIS}.
For the simulations we have used a modification of the DAIS algorithm described in Section \ref{methodology}, in order to make it more accurate. 
More specifically, DAIS restarts from the detected change-point instead of the start- and end-points of the interval where the detection occurred. This leads to an improvement in the accuracy of the algorithm without affecting the speed.
This is also our recommendation on how to use the algorithm in practice.

\begin{table}[htbp]
\caption{Methods used in the simulations} \label{table:competitors}
\centering
        \begin{tabular}{|l|l|l|l|}
        
            \hline
            Type of signal & Method Notation & Reference & \textsf{R} package\\
            \hline
            & ID & \cite{anastasiou2022detecting} & {\textsf{breakfast}} (\citep{breakfast_package})\\
            & PELT & \cite{doi:10.1080/01621459.2012.737745} & {\textsf{changepoint}} (\citep{pelt_package})\\
            & WBS & \cite{fryzlewicz2014wild} & {\textsf{breakfast}} (\citep{breakfast_package})\\
            & WBS2 & \cite{fryzlewicz2020detecting} & {\textsf{breakfast}} (\citep{breakfast_package})\\
            & NOT & \cite{baranowski2019narrowest} & {\textsf{breakfast}} (\citep{breakfast_package})\\
            Piecewise-constant & MOSUM & \cite{10.3150/16-BEJ887} & {\textsf{mosum}} (\citep{JSSv097i08}) \\
            & MSCP & \cite{10.1214/22-EJS2101} & {\textsf{mscp}} (\citep{mscp_package})\\
            & SeedBS & \cite{seeded_bs} & - \\
            & DP\_univar & \cite{10.1214/20-EJS1710} & {\textsf{changepoints}} (\citep{changepoints_package})\\
            & \sla{SMUCE} & \cite{frick_multiscale_2014} & \sla{{\textsf{stepR}}} (\citep{stepr_package})\\
            & \sla{MsFPOP} & \cite{Liehrmann22112024} & \sla{-} \\
            \hline
            & ID & \cite{anastasiou2022detecting} & {\textsf{breakfast}} (\citep{breakfast_package})\\
            & CPOP & \cite{doi:10.1080/10618600.2018.1512868} & {\textsf{cpop}} (\citep{cpop_package}) \\
            Continuous, & NOT & \cite{baranowski2019narrowest} & {\textsf{breakfast}} (\citep{breakfast_package})\\
            piecewise-linear & MARS & \cite{10.1214/aos/1176347963} & {\textsf{earth}} (\citep{earth_package})\\
            & TF & \cite{doi:10.1137/070690274} & -\\
            & TS & \cite{maeng2023detecting} & {\textsf{trendsegmentR}} (\citep{trendSegmentR_package})\\
            & local\_poly & \cite{yu_localising_2022} & {\textsf{changepoints}} (\citep{changepoints_package})\\
            \hline
        \end{tabular}
\end{table}

As a measure of the accuracy of the estimated number of change-points, we present the difference between the number of change-points detected and the true number of change-points ($\hat{N} - N$). 
As a measure of accuracy of the location of the detected change-points, we give Monte-Carlo estimates of the mean squared error, 
\begin{equation*}
    \textrm{MSE} = T^{-1} \sum_{t=1}^{T} \mathbb{E} \left(\hat{f_t} - f_t \right) ^2,
\end{equation*}
where $T$ is the total length of the sequence and $\hat{f_t}$ is the ordinary least squares approximation of $f_t$ between two successive change-points.
We also report the scaled Hausdorff distance, defined as
\begin{equation*}
    d_H = n_s^{-1} \max {\Big\{\max_j \min_k \left| r_j - \hat{r}_k \right|, \max_k \min_j \left| r_j - \hat{r}_k \right| \Big\}}
\end{equation*}
where $n_s = \max_{j=1,\ldots,N} \Tilde{\delta}_j$ is the length of the largest segment between successive change-points. 
The Hausdorff distance, $d_H$, is given for all signals except for signals that do not contain any change-points; in such cases, $d_H$ is uninformative. 
Finally, the average computational time for each method is also reported. 
The top performing algorithm, in terms of accuracy in the number of change-points detected, is presented in bold, along with all methods that performed within $5\%$.
\sla{The lowest MSE and $d_H$ values are also in bold, as well as all values that are within 10\% of the lowest ones.}
The signals used can be found in Appendix \ref{simulations_supplement} and further simulation results are in the supplementary material.

The signals included in this section contain more than one change-point. For signals with no or one change-point, please refer to Tables~\ref{supp:table:justnoise, long signal} and \ref{table:justnoise_wave} in the supplementary material. 
Due to their difficult structure, we first put the focus of the discussion on Signals (S1) and (S2).

Signal (S1) (Table \ref{table:small_dist, small_dist2}) is a particularly difficult structure for most algorithms because it includes a pair of consecutive change-points that shift the data sequence to opposite directions, which masks the existence of the change-points. \sla{Signals with this type of difficult settings make the appreciation of the isolation aspect more apparent; for a relevant discussion of the importance of localization in such signals see also Section 2 in \cite{fryzlewicz2014wild}. Therefore, detection is harder when the change-points are close to each other, `work against each other', and their isolation, while of great importance, is for some algorithms unlikely or, for others, occurs only in a limited number of intervals. The best performing methods in (S1), in terms of accuracy for the estimated number of change-points, are DAIS and ID\_th; this is unsurprising based on the fact that these two methods are expected to certainly achieve isolation of the change-points prior to their detection.}

\sla{
Signal (S2) contains two change-points at a distance of just 5, at locations $r_1=30$ and $r_2=35$. The magnitudes of the changes regarding $r_1$ and $r_2$ are 2.3 and 5.7, respectively. Due to the large magnitude of change for $r_2$, DAIS starts checking for a change-point around $r_2$, which can be detected while it is still isolated in a small interval (due to the distance with $r_1$ being small). Then, $r_1$, which has a smaller magnitude of change, is also detected while isolated. In contrast to this, algorithms that are also based on thresholding and check the data sequence in a sequential way (while attempting to obtain change-point isolation), either through keeping one of the end-points fixed (for example, ID) or through a moving sum approach (for example, MOSUM) may struggle with such a signal. This is because, in such algorithms, $r_1$ will appear first while scanning the data, and if not detected while still isolated, then there is the risk of $r_2$ being first detected (due to its larger magnitude of change) before $r_1$, in an interval that also contains $r_1$. Discarding the data points already checked (as, for example, in ID), or carrying on the moving sum process further away from $r_1$ (as, for example, in MOSUM), will then mean that $r_1$ can not get detected in any following steps of such algorithms. Moving now away from thresholding based approaches, we highlight that WBS2 (which is based on the SDLL model selection procedure), as well as many optimization based algorithms, more specifically, ID\_ic, WBS\_ic, NOT, SeedBS\_ic, and MsFPOP, have a very good behavior in (S2).
}

Regarding the results for Signals {(S3)} and {(S4)} as shown in Table \ref{table:stairs, mix, many_cpts_mix}, many algorithms, including our proposed method DAIS, perform very well. However, ID\_ic, WBS, SeedBS, and NOT tend to slightly overestimate the true number of change-points. In contrast, the MSCP algorithm seems to underestimate the number of the change-points in those two signals. 
DP\_univar \sla{and MsFPOP} exhibit heavy underestimation in Signal {(S3)}, while they heavily overestimate the true number of change-points in {(S4)}. \sla{SMUCE underestimates in Signal {(S3)} and performs very well in {(S4)}}.
For Signals {(S5)}, {(S6)} and {(S7)} shown in Tables~\ref{table:stairs, mix, many_cpts_mix} and \ref{table:many_cpts, many_cpts_long}, a great number of methods struggle to detect the change-points, due to the limited spacings between them.
Signal {(S5)} 
includes change-points of different magnitudes and with slightly different distances between them. DAIS, ID, WBS\_th, WBS2, PELT, and MsFPOP exhibit very good performance in the signal both in terms of accuracy (with respect to the number and the locations of the estimated change-points) and speed. {\sla{Signals (S6) and (S7) consist of a large number of regularly occurring change-points. DAIS, ID\_th, WBS2, SeedBS\_th, and MsFPOP exhibit a very good behaviour in such complex structures. Notice that algorithms which were tested with both the thresholding and IC stopping rules, such as ID, WBS, and SeedBS, tend to perform well in general but neither is 
\sla{consistently accurate} in terms of the choice of the stopping rule in the various scenarios tested that involve large deviations on the number of change-points as well as their distance. For example, notice the dependency on the model selection method (either IC or thresholding) when the performances related to ID, SeedBS, and WBS for IC and thresholding in signal (S2), are compared to their performances in Signals (S6) and (S7). Such model selection based behaviors are expected as the thresholding approach is more appropriate for a large number of change-points, potentially close to each other, while the information criterion is better suited for a moderate number of change-points with larger spacings. This difference in accuracy between the threshold- and SIC-based stopping rules is what motivated in \cite{anastasiou2022detecting} the introduction of a hybrid version of ID that combines these two stopping rules.}}
Regarding the MOSUM algorithm, it performs well for signals with a small number of change-points; however it struggles in Signals {(S5)}, {(S6)} and (S7), \sla{which are more complex structures involving a large number of change-points with limited spacings}.
WBS2 is a very well performing algorithm in terms of accuracy, but is more computationally expensive compared to DAIS, with DAIS being at least one order of magnitude faster than WBS2 in the simulations.
Based on the simulations, we conclude that overall, DAIS performs very well in various different piecewise-constant signal scenarios regarding the number of change-points, the distance between them, as well as the magnitudes of the changes, when most competitors struggle to accurately detect the changes in some of the cases tested.

We now proceed with a discussion of the performance of DAIS under the continuous, piecewise-linear framework. The simulation results are shown in
Tables \ref{table:wave1} and \ref{table:wave2, wave3}. We conclude that DAIS performs at least as accurately as the algorithms ID and CPOP, and substantially better than other competitors available in the literature, regarding accuracy in terms of both the estimated number and the estimated locations of the change-points. Furthermore, the average computational time for DAIS remains very low in all signals tested; in fact, it is the lowest among the times for the top performing methods.

\begin{table}[htbp]
\caption{Distribution of $\hat{N} - N$ over 100 simulated data sequences of the Signals (S1) and \sla{(S2)}. 
    The average MSE, $d_H$ and computational times are also given.} \label{table:small_dist, small_dist2}
\centering
        \begin{tabular}{|l|c|c|c|c|c|c|c|c|c|c|c|}
            \hline
            &&\multicolumn{6}{|c|}{} &&& \\ 
            &&\multicolumn{6}{|c|}{$\hat{N} - N$} & & & \\
            Method & Signal & $-2$ & $-1$ & 0 & 1 & 2 & $\geq 3$ & MSE & $d_H$ & Time (s) \\
            \hline
            \textbf{DAIS} && 16 &     0 &    \textbf{80} &     2 &     2 &     0 & \textbf{0.014} & 0.089 & 0.004 \\ 
            \textbf{ID\_th} & &  10 &    0 &   \textbf{79} &    8 &    3 &    0 & \textbf{0.013} & 0.081 & 0.007 \\  
            ID\_ic & &  33 &    0 &   64 &    2 &    1 &    0 & 0.018 & 0.179 & 0.013 \\ 
            WBS\_th & &   8 &    7 &   63 &   16 &    5 &    1 & 0.015 & 0.105 & 0.010 \\
            WBS\_ic && 30 &    0 &   68 &    2 &    0 &    0 & 0.016 & 0.166 & 0.032 \\ 
            WBS2 &  & 10 &    2 &   75 &    7 &    1 &    5 & 0.015 & 0.091 & 0.563 \\ 
            PELT && 77 &    0 &   23 &    0 &    0 &    0 & 0.026 & 0.397 & 0.001 \\ 
            NOT & (S1) & 30 &    0 &   68 &    2 &    0 &    0 & 0.016 & 0.167 & 0.037 \\ 
            MOSUM && 15 &    0 &   74 &    6 &    4 &    1 & \textbf{0.014} & 0.108 & 0.011 \\
            MSCP & &75 &   23 &    2 &    0 &    0 &    0 & 0.030 & 0.394 & 2.600 \\ 
            SeedBS\_th && 11 &   18 &   55 &   13 &    2 &    1 & 0.017 & 0.103 & 0.022 \\
            SeedBS\_ic & & 32 &    0 &   66 &    2 &    0 &    0 & 0.017 & 0.176 & 0.022 \\
            DP\_univar & & 0 &    0 &    6 &    3 &   10 &   81 & 0.045 & 0.334 & 0.138 \\ 
            \sla{SMUCE} & & \sla{4} & \sla{36} & \sla{53} & \sla{7} & \sla{0} & \sla{0} & \sla{0.020} & \sla{\textbf{0.057}} & \sla{0.005} \\
            \sla{MsFPOP} & & \sla{29} & \sla{0} & \sla{70} & \sla{1} & \sla{0} & \sla{0} & \sla{0.016} & \sla{0.157} & \sla{0.004} \\
            \hline
            \sla{\textbf{DAIS}} &  &  \sla{0} &    \sla{7} &   \sla{\textbf{86}} &    \sla{4} &    \sla{3} &    \sla{0} & \sla{\textbf{0.063}} & \sla{0.039} & \sla{0.001} \\ 
            \sla{ID\_th} &&    \sla{0} &   \sla{37} &   \sla{49} &   \sla{10} &    \sla{2} &    \sla{2} & \sla{0.125} & \sla{0.065} & \sla{0.001} \\ 
            \sla{\textbf{ID\_ic}} &  &  \sla{0} &    \sla{9} &   \sla{\textbf{83}} &    \sla{6} &    \sla{1} &    \sla{1} & \sla{0.068} & \sla{0.026} & \sla{0.004} \\ 
            \sla{WBS\_th} & &    \sla{0} &    \sla{1} &   \sla{64} &   \sla{17} &    \sla{8} &   \sla{10} & \sla{0.075} & \sla{0.154} & \sla{0.002} \\ 
            \sla{\textbf{WBS\_ic}} & &   \sla{0} &    \sla{4} &   \sla{\textbf{86}} &    \sla{7} &    \sla{2} &    \sla{1} & \sla{\textbf{0.061}} & \sla{0.025} & \sla{0.005} \\ 
            \sla{WBS2} &  & \sla{0} &   \sla{24} &   \sla{67} &    \sla{3} &    \sla{3} &    \sla{3} & \sla{0.087} & \sla{0.048} & \sla{0.055} \\ 
            \sla{PELT} & &   \sla{0} &   \sla{18} &   \sla{81} &    \sla{1} &    \sla{0} &    \sla{0} & \sla{0.068} & \sla{\textbf{0.011}} & \sla{0.001} \\ 
            \sla{\textbf{NOT}} & \sla{(S2)} &   \sla{0} &    \sla{4} &   \sla{\textbf{87}} &    \sla{6} &    \sla{2} &    \sla{1} & \sla{\textbf{0.062}} & \sla{0.024} & \sla{0.013} \\ 
            \sla{MOSUM} & &   \sla{0} &   \sla{83} &   \sla{14} &    \sla{3} &    \sla{0} &    \sla{0} & \sla{0.185} & \sla{0.050} & \sla{0.002} \\ 
            \sla{MSCP} &   & \sla{0} &   \sla{100} &    \sla{0} &    \sla{0} &    \sla{0} &    \sla{0} & \sla{0.197} & \sla{0.069} & \sla{0.066} \\ 
            \sla{SeedBS\_th} & & \sla{0} &    \sla{5} &   \sla{75} &   \sla{12} &    \sla{7} &    \sla{1} & \sla{\textbf{0.064}} & \sla{0.073} & \sla{0.026} \\ 
            \sla{\textbf{SeedBS\_ic}} &&     \sla{0} &    \sla{6} &   \sla{\textbf{87}} &    \sla{6} &    \sla{1} &    \sla{0} & \sla{\textbf{0.060}} & \sla{0.021} & \sla{0.026} \\ 
            \sla{DP\_univar} &  &  \sla{0} &   \sla{15} &   \sla{59} &   \sla{13} &   \sla{11} &    \sla{2} & \sla{0.183} & \sla{0.175} & \sla{0.001} \\ 
            \sla{SMUCE} & &   \sla{0} &   \sla{36} &   \sla{59} &    \sla{5} &    \sla{0} &    \sla{0} & \sla{0.111} & \sla{0.033} & \sla{0.001} \\ 
            \sla{\textbf{MsFPOP}} & &   \sla{0} &   \sla{12} &   \sla{\textbf{86}} &    \sla{2} &    \sla{0} &   \sla{0} & \sla{\textbf{0.066}} & \sla{0.015} & \sla{0.001} \\  
            \hline
\end{tabular}
\end{table}
\begin{table}[!ht]
\caption{Distribution of $\hat{N} - N$ over 100 simulated data sequences of the Signals {(S3)}, {(S4)} and {(S5)}. 
    The average MSE, $d_H$ and computational times are also given.} \label{table:stairs, mix, many_cpts_mix}
\centering
        \begin{tabular}{|l|c|c|c|c|c|c|c|c|c|c|c|}
            \hline
            &&\multicolumn{7}{|c|}{} &&& \\ 
            &&\multicolumn{7}{|c|}{$\hat{N} - N$} & & & \\
            Method & Signal & $\leq 3$ & $-2$ & $-1$ & 0 & 1 & 2 & $\geq 3$ & MSE & $d_H$ & Time (s) \\
            \hline
            \textbf{DAIS} && 0 & 0 & 2 & \textbf{95} & 3 & 0 & 0 & 0.023 & 0.014 & 0.002 \\
            \textbf{ID\_th} & &   0 &    0 &    1 &   \textbf{94} &    5 &    0 &    0 & 0.022 & \textbf{0.013} & 0.003 \\ 
            ID\_ic &  &   0 &    0 &    0 &   78 &   17 &    4 &    1 & 0.023 & 0.015 & 0.006 \\ 
            WBS\_th &  &   0 &    0 &    1 &   86 &   12 &    1 &    0 & 0.025 & 0.016 & 0.003 \\ 
            WBS\_ic && 0 &    0 &    0 &   62 &   26 &    6 &    6 & 0.024 & 0.017 & 0.005 \\ 
            {\textbf{WBS2}} & & 0 &    0 &    2 &   {\textbf{92}} &    5 &    1 &    0 & 0.025 & 0.015 & 0.039 \\ 
            \textbf{PELT} && 0 & 0 & 7 & \textbf{93} & 0 & 0 & 0 & 0.022 & 0.014 & 0.001 \\ 
            NOT & {(S3)} & 0 &    0 &    0 &   81 &   12 &    2 &    5 & 0.022 & 0.014 & 0.053 \\ 
            \textbf{MOSUM} && 0 & 0 & 2 & \textbf{97} & 1 & 0 & 0 & \textbf{0.018} & \textbf{0.012} & 0.004 \\ 
            MSCP && 100 &    0 &    0 &    0 &    0 &    0 &    0 & 0.949 & 0.182 & 0.072 \\ 
            SeedBS\_th && 0 &    3 &   10 &   79 &    8 &    0 &    0 & 0.030 & 0.020 & 0.016 \\ 
            SeedBS\_ic && 0 &    0 &    0 &   70 &   19 &    6 &    5 & 0.025 & 0.016 & 0.016 \\
            DP\_univar && 100 & 0 & 0 & 0 & 0 & 0 & 0 & 0.240 & 0.067 & 0.001\\
            \sla{SMUCE}  &  & \sla{47} & \sla{26} & \sla{8} & \sla{19} & \sla{0} & \sla{0} & \sla{0} & \sla{0.115} & \sla{0.044} & \sla{0.001} \\
            \sla{MsFPOP} & & \sla{100} & \sla{0} & \sla{0} & \sla{0} & \sla{0} & \sla{0} & \sla{0} & \sla{0.344} & \sla{0.090} & \sla{0.001} \\
            \hline
            \textbf{DAIS} & & 0 & 0 & 0 & \textbf{96} & 4 & 0 & 0 & 1.699 & 0.012 & 0.002 \\
             ID\_th &  &  0 &    0 &    0 &   92 &    8 &    0 &    0 & 1.710 & 0.013 & 0.002 \\ 
            ID\_ic &  &  0 &    0 &    0 &   88 &   10 &    1 &    1 & 1.760 & 0.014 & 0.007 \\ 
            WBS\_th & &   0 &    0 &    0 &   79 &   19 &    2 &    0 & \textbf{1.490} & 0.018 & 0.004 \\ 
            WBS\_ic & & 0 &    0 &    0 &   94 &    6 &    0 &    0 & \textbf{1.470} & 0.009 & 0.010 \\ 
            WBS2 & & 0 &    0 &    0 &   94 &    5 &    0 &    1 & \textbf{1.480} & 0.012 & 0.118 \\ 
            \textbf{PELT} & & 0 & 0 & 0 & \textbf{100} & 0 & 0 & 0 & \textbf{1.475} & \textbf{0.008} & 0.001 \\ 
            \textbf{NOT} & {(S4)} & 0 &    0 &    0 &   \textbf{97} &    3 &    0 &    0 & \textbf{1.550} & 0.009 & 0.018 \\ 
            MOSUM & & 0 & 0 & 0 & 90 & 8 & 2 & 0 & \textbf{1.513} & 0.014 & 0.007 \\ 
            MSCP & & 49 &   33 &   16 &    2 &    0 &    0 &    0 & 10.770 & 0.109 & 0.289 \\ 
            SeedBS\_th & & 0 &    0 &    1 &   83 &   14 &    2 &    0 & 1.630 & 0.018 & 0.017 \\ 
            \textbf{SeedBS\_ic} & & 0 &    0 &    0 &   \textbf{96} &    4 &    0 &    0 & \textbf{1.600} & 0.010 & 0.017 \\ 
            DP\_univar & & 0 & 0 & 0 & 0 & 0 & 0 & 100 & 4.544 & 0.112 & 0.004\\
            \sla{\textbf{SMUCE}} &  & \sla{0} & \sla{0} & \sla{2} & \sla{\textbf{98}} & \sla{0} & \sla{0} & \sla{0} & \sla{\textbf{1.560}} & \sla{\textbf{0.008}} & \sla{0.002} \\
            \sla{MsFPOP} & & \sla{0} & \sla{0} & \sla{0} & \sla{0} & \sla{0} & \sla{0} & \sla{100} & \sla{11.80} & \sla{0.153} & \sla{0.001} \\
            \hline
            \textbf{DAIS} & & 0 &     0 &     0 &    \textbf{98} &     2 &     0 &     0 & 0.253 & \textbf{0.004} & 0.001 \\ 
              \textbf{ID\_th} &  &  0 &    0 &    0 &   \textbf{99} &    1 &    0 &    0 & \textbf{0.246} & 0.005 & 0.003 \\    
            ID\_ic &  &  0 &    0 &    0 &   91 &    9 &    0 &    0 & 0.259 & 0.007 & 0.004 \\ 
            WBS\_th &   & 0 &    0 &    0 &   94 &    6 &    0 &    0 & \textbf{0.239} & 0.005 & 0.001 \\
            WBS\_ic & & 0 &    0 &    0 &   71 &   19 &    9 &    1 & 0.270 & 0.012 & 0.003 \\
            \textbf{WBS2} & & 0 &    0 &    0 &   \textbf{99} &    1 &    0 &    0 & \textbf{0.235} & \textbf{0.004} & 0.024 \\ 
            PELT & & 4 &     4 &     0 &    92 &     0 &     0 &     0 & 0.380 & 0.018 & 0.001  \\ 
            NOT & {(S5)} & 0 &    0 &    0 &   72 &   17 &    8 &    3 & 0.271 & 0.012 & 0.004 \\ 
            MOSUM & & 100 &     0 &     0 &     0 &     0 &     0 &     0 & 6.460 & 0.893 & 0.001 \\ 
            MSCP & & 100 &     0 &     0 &     0 &     0 &     0 &     0 & 6.270 & 0.446 & 0.026 \\ 
            SeedBS\_th & & 0 &    1 &   14 &   80 &    5 &    0 &    0 & 0.348 & 0.016 & 0.019 \\
            SeedBS\_ic & & 0 &    0 &    0 &   73 &   20 &    5 &    2 & 0.265 & 0.011 & 0.019 \\ 
            DP\_univar & & 100 &     0 &     0 &     0 &     0 &     0 &     0 & 5.790 & 0.154 & 0.001  \\ 
            \sla{SMUCE} & & \sla{14} & \sla{10} & \sla{19} & \sla{57} & \sla{0} & \sla{0} & \sla{0} & \sla{1.110} & \sla{0.034} & \sla{0.001} \\
            \sla{\textbf{MsFPOP}} & &  \sla{0} & \sla{0} & \sla{0} & \sla{\textbf{100}} & \sla{0} & \sla{0} & \sla{0} & \sla{\textbf{0.228}} & \sla{\textbf{0.004}} & \sla{0.001} \\
            \hline
            \end{tabular}
\end{table}

\begin{table}[tbp]
\caption{Distribution of $\hat{N} - N$ over 100 simulated data sequences of the Signals {(S6)} and {(S7)}. 
    The average MSE, $d_H$ and computational times are also given.} \label{table:many_cpts, many_cpts_long}
\centering
        \begin{tabular}{ |l|c|c|c|c|c|c|c|c|}
            \hline
            &&\multicolumn{3}{|c|}{} &&& \\ 
            &&\multicolumn{3}{|c|}{$\hat{N} - N$} & & & \\
            Method& Signal & $-99$& $[-98,-11]$ & $[-10,10]$ & MSE & $d_H$ & Time (s) \\
            \hline
            \textbf{DAIS} && 0 & 5 & \textbf{95}  & 0.435 & 0.011 & 0.016 \\
            \textbf{ID\_th} & &   0 &    3 &   \textbf{97} & 0.403 & 0.009 & 0.017 \\   
            ID\_ic & &  98 &    0 &    2  & 3.930 & 0.970 & 0.053 \\ 
            WBS\_th &  &  0 &   73 &   27 &  1.230 & 0.012 & 0.007 \\
             WBS\_ic && 99 &    0 &    1  & 3.960 & 0.980 & 0.036 \\ 
             \textbf{WBS2} & &  0 &    0 &  \textbf{100} &  \textbf{0.280} & \textbf{0.003} & 0.336 \\ 
            PELT && 2 & 93 & 5 & 2.490 & 0.251 & 0.007 \\
            NOT & {(S6)} & 100 &    0 &    0 &  4.000 & 0.990 & 0.034 \\ 
            MOSUM && 100 & 0 & 0  & 4.000 & 0.990 & 0.003 \\ 
            MSCP && 0 &  100 &   0 & 3.930 & 0.057 & 1.400 \\ 
            SeedBS\_th && 0 &   26 &   74 &  0.829 & 0.011 & 0.020 \\
            SeedBS\_ic && 100 &    0 &    0 &    4.000 & 0.990 & 0.020 \\ 
            DP\_univar && 0 & 100 & 0 & 3.690 & 0.026 & 0.048 \\ 
            \sla{SMUCE} &  &  \sla{0} & \sla{100} & \sla{0} & \sla{3.570} & \sla{0.030} & \sla{0.002} \\
            \sla{\textbf{MsFPOP}} &  & \sla{0} & \sla{0} & \sla{\textbf{100}} & \sla{0.346} & \sla{0.008} & \sla{0.001} \\
            \hline
            \textbf{DAIS} && 0 & 0 & \textbf{100} & 0.234 & 0.003 & 0.013 \\
            \textbf{ID\_th} &  &  0 &    0 &  \textbf{100} & 0.245 & 0.003 & 0.021 \\ 
            ID\_ic & & 100 &    0 &    0 & 6.250 & 0.992 & 0.061 \\ 
            WBS\_th & &   0 &   97 &    3 & 2.060 & 0.011 & 0.015 \\ 
            WBS\_ic && 100 &    0 &    0 & 6.250 & 0.992 & 0.025 \\
            \textbf{WBS2} & & 0 &    0 &  \textbf{100} & \textbf{0.199} & \textbf{0.001} & 0.307 \\ 
            PELT && 2 & 55 & 43 & 1.670 & 0.082 & 0.001 \\ 
            NOT & {(S7)} & 100 &    0 &    0 & 6.250 & 0.992 & 0.031 \\ 
            MOSUM && 100 &0 &0 & 6.250 & 0.992 & 0.003 \\ 
            MSCP &&  99 &  1 &    0 & 6.220 & 0.160 & 1.090 \\ 
            \textbf{SeedBS\_th} &&  0 &    0 &  \textbf{100} & 0.402 & 0.007 & 0.023 \\ 
            SeedBS\_ic && 100 &    0 &    0 & 6.250 & 0.992 & 0.023 \\ 
            DP\_univar && 0 & 100 &0 & 5.710 & 0.025 & 0.030 \\ 
            \sla{SMUCE} & &  \sla{0} & \sla{100} & \sla{0} & \sla{5.430} & \sla{0.020} & \sla{0.002} \\
            \sla{\textbf{MsFPOP}} &  & \sla{0} & \sla{0} & \sla{\textbf{100}} & \sla{\textbf{0.209}} & \sla{0.002} & \sla{0.001} \\
            \hline
        \end{tabular}
\end{table}

\begin{table}[tbp]
\centering
\caption{Distribution of $\hat{N} - N$ over 100 simulated data sequences of the Signal {(S9)}. 
    The average MSE, $d_H$ and computational times are also given.} \label{table:wave1}
        \begin{tabular}{|l|l|c|c|c|c|c|c|c|c|c|}
            \hline
            &&\multicolumn{6}{|c|}{} &&& \\ 
            &&\multicolumn{6}{|c|}{$\hat{N} - N$} &&& \\
            Method & Signal & $\leq -2$ & $-1$ & 0 & 1 & $2$ & $\geq3$ & MSE & $d_H$ & Time (s) \\
            \hline
            \textbf{DAIS} && 0 & 0 &  \textbf{99} &    1 &    0 &    0 & 0.030 & 0.085 & 0.020 \\ 
            \textbf{ID\_th} &&  0 &    0 &   \textbf{95} &    5 &    0 &    0 & 0.030 & 0.086 & 0.024 \\ 
             \textbf{ID\_ic} && 0 & 0 &   93 &    7 &    0 &    0 & 0.034 & 0.099 & 0.093 \\ 
            \textbf{CPOP} && 0 & 0 &  \textbf{99} &    1 &    0 &    0 & \textbf{0.013} & \textbf{0.050} & 5.550 \\ 
            \textbf{NOT} & {(S9)} & 0 &    0 &  \textbf{100} &    0 &    0 &    0 & 0.015 & \textbf{0.055} & 0.720 \\ 
            MARS && 0 & 0 &   5 &   39 &   47 &    9 & 0.025 & 0.194 & 0.008 \\ 
            TF && 0 &  0 &  0 &    0 &    0 &  100 & 0.018 & 0.440 & 0.768 \\ 
            \textbf{TS} && 0 & 1 &   \textbf{99} &    0 &    0 &    0 & 0.096 & 0.186 & 1.020 \\ 
            local\_poly && 0 & 0 &   1 &    5 &   12 &   82 & 0.054 & 0.503 & 125.400 \\ 
            \hline
             \end{tabular}
\end{table}

\begin{table}[h]
\centering
\caption{Distribution of $\hat{N} - N$ over 100 simulated data sequences of the Signals {(S10)} and {(S11)}. 
    The average MSE, $d_H$ and computational times are also given.} \label{table:wave2, wave3}
        \begin{tabular}{|l|l|c|c|c|c|c|c|c|c|c|c|}
            \hline
            && \multicolumn{7}{|c|}{} &&& \\ 
            && \multicolumn{7}{|c|}{$\hat{N} - N$} &&& \\
            Method & Signal & $\leq -15$ & $(-15,-2]$ & $-1$ & 0 & 1 & $[2,15)$ & $\geq15$ & MSE & $d_H$ & Time (s) \\
            \hline
            \textbf{DAIS} && 0 &    0 &    0 &  \textbf{100} &    0 &    0 &    0 & 0.267 & 0.296 & 0.018 \\ 
            \textbf{ID\_th} & &   0 &    0 &    1 &   \textbf{97} &    2 &    0 &    0 & 0.233 & 0.273 & 0.046 \\ 
            ID\_ic && 0 &    0 &    0 &   90 &   10 &    0 &    0 & 0.264 & 0.302 & 0.478 \\ 
            CPOP && 0 &    0 &    0 &   91 &    9 &    0 &    0 & 0.381 & \textbf{0.247} & 0.439 \\ 
            NOT & {(S10)} & 0 &   57 &   16 &   17 &   10 &    0 &    0 & 0.808 & 0.871 & 1.030 \\ 
            MARS && 100 &    0 &    0 &    0 &    0 &    0 &    0 & 4.700 & 98.50 & 0.003 \\ 
            TF && 0 &    0 &    0 &    0 &    0 &    0 &  100 & \textbf{0.205} & 0.395 & 0.645 \\ 
            TS && 0 &   33 &   37 &   30 &    0 &    0 &    0 & 0.854 & 0.580 & 1.430 \\ 
            local\_poly && 100 &    0 &    0 &    0 &    0 &    0 &    0 & 4.740 & 0.827 & 157.80 \\ 
            \hline
            \textbf{DAIS} && 0 &    0 &    0 &  \textbf{100} &    0 &    0 &    0 & \textbf{0.039} & \textbf{0.196} & 0.016 \\ 
            \textbf{ID\_th} & & 0 &    0 &    0 &  \textbf{100} &    0 &    0 &    0 & \textbf{0.037} & \textbf{0.216} & 0.036 \\ 
             \textbf{ID\_ic} && 0 &    0 &    0 &   \textbf{98} &    2 &    0 &    0 & \textbf{0.041} & 0.243 & 0.433 \\ 
            \textbf{CPOP} &&  0 &    0 &    0 &   \textbf{99} &    1 &    0 &    0 & 0.252 & 0.291 & 0.062 \\ 
            NOT & {(S11)} & 100 &    0 &    0 &    0 &    0 &    0 &    0 & 1.060 & 119.000 & 0.815 \\ 
            MARS && 100 &    0 &    0 &    0 &    0 &    0 &    0 & 1.060 & 118.000 & 0.002 \\
            TF && 0 &    0 &    0 &    0 &    0 &    1 &   99 & 0.191 & 0.307 & 0.300 \\ 
            TS &&  100 &    0 &    0 &    0 &    0 &    0 &    0 & 1.060 & 119.000 & 0.480 \\ 
            local\_poly && 100 &    0 &    0 &    0 &    0 &    0 &    0 & 1.060 & 76.700 & 16.900 \\ 
            \hline
            \end{tabular}
\end{table}


    \section{DAIS extensions} \label{sec: DAIS extensions}

    \subsection{Temporal dependence} \label{sec: temporal_dependence}
    We first consider relaxing the assumption of independence for the error terms $\epsilon_t$. 
    Temporal dependence is therefore introduced in the observed data sequence.
    In such cases, methods that reduce the autocorrelation in the time series can be applied.
    This can be done using a subsampling technique, which is briefly described below. 
    
    For a chosen integer $s$ we subsample the observed time series $X_t$ by choosing every $s^\text{th}$ observation. 
    This creates $s$ different data sequences,
    for which the autocorrelation is reduced compared to the original data sequence. 
    If, for example, $s=5$ and \linebreak $T=1000$, the new data sequences are of the form
    $\{X_1, X_{6}, \ldots, X_{996}\}$, $\{X_2, X_{7}, \ldots, X_{997}\}$, $\{X_3, X_{8}, \ldots, X_{998}\}$, $\{X_4, X_{9}, \ldots, X_{999}\}$, $\{X_5, X_{10}, \ldots, X_{1000}\}$.
    DAIS is applied to each one of these smaller data sequences and we obtain $s$ sets of detected change-points. 
    We can then apply a majority voting rule, which discards any change-points that are not detected in at least $\eta$ time series, where $\eta\leq s$ is a positive integer.
    The change-points are then transformed to represent locations in the original time series.
    The purpose of taking every $s^\text{th}$ observation is to not include the observations which have the highest correlation between them. 
    That is, the observations just before or just after each data point.
    Thus, for larger values of $s$, we manage to reduce the autocorrelation in the resulting disjoint data sequences, but we obtain a larger number of smaller data sequences, which has a negative impact on the detection accuracy.
    More details on this can be found in \cite{ccid}.
    A different way to reduce autocorrelation is using a pre-averaging technique, such as the one described in Section~\ref{sec: heavy_tailed}, which involves averaging the data sequence before applying the algorithm in order to reduce the correlation between consecutive observations.

    \subsection{Heavy-tailed noise} \label{sec: heavy_tailed}

    In this subsection we relax the assumption that the error term $\epsilon_t$ is Gaussian and we instead consider a heavy-tailed distribution for $\epsilon_t$.
    To improve the performance of DAIS in such cases, the data sequence can be pre-processed
    by averaging the time series to bring the noise closer to Gaussianity, as described in Section 4.5 of \cite{anastasiou2022detecting}.
    Using this technique, we take advantage of the Central Limit Theorem, the noise of the observations with large $\epsilon_t$ is reduced in absolute value, and we obtain a smaller data sequence with less extreme values for all the data points.
    The idea is that for a chosen integer $s$, we define $Q = \lceil T/s \rceil$ and set 
    \begin{equation*}
        X^\ast_q = \frac{1}{s} \sum_{t=(q-1)s+1}^{qs} X_t, \text{ for } \sla{q=1,2,\ldots,Q-1} \quad \text{and} \quad X^\ast_Q = \frac{1}{T-(Q-1)s} \sum_{t=(Q-1)s+1}^T X_t.
    \end{equation*}
    DAIS can then be applied on $X^\ast_t$ to obtain change-points $\hat{r}^\ast_1, \ldots, \hat{r}^\ast_{\hat{N}}$. 
    The change-points of the original time series $X_t$ can be obtained using the transformation $\hat{r}_i = (\hat{r}^\ast_i - 1)s + \lfloor \frac{s+1}{2} \rfloor$, for $i=1,\ldots,\hat{N}$. 
    The choice of the value of $s$ is important.
    On the one hand, taking larger values brings the data closer to Gaussianity. 
    On the other hand, these values of $s$ return smaller data sequences that have been pre-processed more and so we lose on the accuracy of both the location and the number of change-points detected.
    In practice we recommend $s=5$.
    
    The performance of DAIS in the case of heavy-tailed data can be found in Table~\ref{tab: heavy_tailed}.
    For the pre-averaging technique we use $s=5$ and the heavy-tailed noise is chosen to be the Student's t-distribution with degrees of freedom being equal to either 5 or 7.
    Signals (S1) and {(S8)} are piecewise-constant and {(S9)} is piecewise-linear and their description can be found in Appendix \ref{simulations_supplement}.
Even though the accuracy of DAIS is reduced compared to Section~\ref{simulations}, the results indicate a good level of robustness for DAIS at the presence of heavy-tailed noise. The more heavy-tailed the noise, the more the overdetection, which is expected, as spikes caused by the heavy-tailed distribution of $\epsilon_t$ can sometimes be mistakenly detected as change-points.
However, DAIS performs relatively well for Signals (S1) and (S8) even for the case of 5 degrees of freedom for the Student's t-distribution.
\begin{table}[htbp]
\caption{Distribution of $\hat{N} - N$ over 100 simulated data sequences of the Signals (S1), {(S8)} and {(S9)} with noise following the Student's t-distribution with $d=5,7$ degrees of freedom. The average MSE, $d_H$ and computational times including the pre-averaging are also given.} \label{tab: heavy_tailed}
\centering
        \begin{tabular}{ |c|c|c|c|c|c|c|c|c|c|c|c|c|c|}
            \hline
            &&\multicolumn{7}{|c|}{} &&& \\ 
            &&\multicolumn{7}{|c|}{$\hat{N} - N$} & & & \\
            $d$ & Signal & $\leq -3$& $-2$ & $-1$ & 0 & 1 & 2 & $\geq 3$ & MSE & $d_H$ & Time (ms) \\
            \hline
            & (S1)  & 0 &  1 & 0 & 87 & 6 & 6 & 0 & 0.025 & 0.038 & 1.0 \\
            5 & {(S8)} & 0 & 0 & 0 & 70 & 14 & 12 & 4 & 0.016 & 0.081 & 1.9 \\
            & {(S9)} & 0 & 0 & 0 & 37 & 22 & 21 & 20 & 2 & 21.1 & 10 \\
            \hline
            & (S1)  & 0 &  0 & 0 & 92 & 4 & 3 & 1 & 0.022 & 0.023 & 1.0 \\
            7 & {(S8)} & 0 & 0 & 0 & 83 & 6 & 8 & 3 & 0.015 & 0.060 & 0.8 \\
            & {(S9)} & 0 & 0 & 0 & 82 & 8 & 10 & 0 & 2.01 & 21 & 10 \\
            \hline
            \end{tabular}
\end{table}

    \subsection{Multivariate models} \label{sec: multivariate_extension}

    The DAIS algorithm can be extended to multivariate or high dimensional models. The model considered is the following:
    \begin{equation}\label{eq: high_dim_model}
        \boldsymbol{X_t} = \boldsymbol{f_t} + \boldsymbol{\epsilon_t}, t=1,...,T,
    \end{equation}
    where $\boldsymbol{X_t} \in \R^{d\times 1}$ are the observed data and $\boldsymbol{f_t} \in \R^{d\times 1}$ is the d-dimensional deterministic signal with structural changes at certain points.
    In the case of piecewise-constant signals, the structure of $\boldsymbol{f_t}$ is given by 
    $\boldsymbol{f_t} = \boldsymbol{\mu_j}$ for $t \in \{r_{j-1} + 1, \ldots, r_j \}$ and $\boldsymbol{f_{r_j}} \neq \boldsymbol{f_{r_j+1}}$ where $\boldsymbol{\mu_j} \in \R^{d\times 1}$ for $j = 1,2,\ldots,N + 1$.
    For continuous, piecewise-linear signals, $\boldsymbol{f_t} = \boldsymbol{\mu_{j,1}} + \boldsymbol{\mu_{j,2}} t$ for $t \in \{r_{j-1} + 1, \ldots, r_j \}$ and $\boldsymbol{f_{r_j-1}} + \boldsymbol{f_{r_j+1}} \neq 2\boldsymbol{f_{r_j}}$ where $\boldsymbol{\mu_{j,1}}, \boldsymbol{\mu_{j,2}} \in \R^{d\times 1}$.
    To ensure continuity, we need the additional constraint of $\boldsymbol{\mu_{k,1}} + \boldsymbol{\mu_{k,2}} r_k = \boldsymbol{\mu_{k+1,1}} + \boldsymbol{\mu_{k+1,2}} r_k$ for $k\in\{1,2,\ldots,N\}$.
    As with the univariate case, the algorithm can also be applied to more general signal structures.

    The calculation of the largest difference, that was used as a starting point for the univariate signal, needs to be adapted for the multivariate case.
    This can be done using mean-dominant norms $L:\R^d \rightarrow \R$, whose definition can be found in \cite{Carlstein1988}.
    Some examples are
    \begin{equation} \label{mean_dominant}
        L_2(\boldsymbol{X_t}) := \sqrt{\frac{1}{d} \sum_{i=1}^d X_{t,i}^2},
        \quad
        L_\infty(\boldsymbol{X_t}) := \sup_{i=1, \ldots, d} \lvert X_{t,i} \rvert.
    \end{equation}
    The location of the largest difference of the interval $[s,e]$, $1\leq s <e \leq T$ can be calculated as
    \begin{equation} \label{def_largest_diff_multi}
    d^{\text{multi}}_{s,e} = \left\{\begin{array}{ll}
    \textrm{argmax}_{t \in \{s, s+1, \dots, e-1\}}\left\{L( \lvert \boldsymbol{X_{t+1}} - \boldsymbol{X_t} \rvert ) \right\}, &\boldsymbol{f_t} \textrm{ piecewise-constant}, \\
    \textrm{argmax}_{t \in \{s, s+1, \dots, e-2\}} \left \{L( \lvert \boldsymbol{X_{t+2}} - 2\boldsymbol{X_{t+1}} + \boldsymbol{X_t} \rvert ) \right\}, &\boldsymbol{f_t} \textrm{ piecewise-linear.}
    \end{array}\right.
\end{equation}
    We define $\boldsymbol{y_t}$ to be the $\R^{d\times1}$ vector with entries the value of the chosen contrast function, this being \eqref{CUSUM} for piecewise-constant signals and \eqref{phi_definition} for piecewise-linear ones. 
    At each step of the algorithm, the goal is to decide whether the point that maximizes the mean-dominant norm $L(\boldsymbol{y_t})$ is a change-point. 
    The aggregation scheme for the contrast function, as just described, was also proposed by \cite{anastasiou_generalized_2023}.
    Algorithm \ref{alg:MDAIS} provides a pseudocode for DAIS in the case of multivariate signals, which we call MDAIS.
    For a discussion on the choice of the mean-dominant norm, we refer the reader to \cite{ccid, anastasiou_generalized_2023, cho_change-point_2016, cho_multiple-change-point_2015}.

    Theoretical results for this variant of the algorithm about consistency for the number and locations of the estimated change-points can be proven. 
    In the rest of this section, we provide the consistency theorem that concerns the case of a piecewise-constant $\boldsymbol{f_t}$ and the outline of its proof can be found in Appendix~\ref{appendix: proof_multivariate}. 
    However, since the focus of this paper is proposing the novel data-adaptive methodology, we do not delve further in this. 
    Following a similar methodology, theoretical results can also be proven for $\boldsymbol{f_t}$ being continuous, piecewise-linear.
    For the statement of the theorem, 
    we require the following assumption, which is equivalent to Assumption (A2) (used for the case of univariate $f_t$):
    \begin{enumerate}
        \item[(A4)]  The quantities $\delta_T$ and $\underline{f}_T$ are connected by $\sqrt{\delta_T} \underline{f}_T \geq \underline{C}_M \sqrt{\log \left( T d^{1/4} \right)}$, for a large enough constant $\underline{C}_M$.
    \end{enumerate}
    Similarly to the univariate case, the number of change-points $N$ is allowed to grow with $T$ and $d$.
    Note that the threshold constant, $\zeta_{T, d}$ depends on both the length of the data sequence and its dimension.
    \sla{Before providing the statement of the theorem, we generalize the contrast function chosen for the case of piecewise-constant signals, which is the absolute value of the CUSUM statistic as given in \eqref{CUSUM}, to 
    \begin{equation} \label{CUSUM_multiv}
        \left| \tilde{X}_{s,e}^{b,j} \right| = \left| \sqrt{\frac{e-b}{\ell (b-s+1)}}\sum_{t=s}^{b}X_{t,j} - \sqrt{\frac{b-s+1}{\ell (e-b)}}\sum_{t=b+1}^{e}X_{t,j} \right|
    \end{equation}
    for $\ell = e-s+1$, $b\in\{s,\ldots,e-1\}$ and $j \in \{1,\ldots,d\}$.}
    
    \begin{theorem}\label{thm: consistency_multivariate}
        Let $\{ \boldsymbol{X_t} \}_{t=1,\ldots,T}$ follow model \eqref{eq: high_dim_model} with $\boldsymbol{f_t}$ piecewise-constant and let $\boldsymbol{\epsilon_t} \sim \mathcal{N}_d (\boldsymbol{0}, \Sigma)$, where $\Sigma \in \R^{d\times d}$ is positive definite and also that Assumptions (A1) and (A4) hold.
        Let $N$ and $r_j, j=1, 2, \ldots N$ be the number and location of the change-points, while $\hat{N}$ and $\hat{r}_j, j=1, 2, \ldots \hat{N}$ their estimates, sorted in increasing order. 
        In addition, $\boldsymbol{\Delta_j} 
        := \lvert \boldsymbol{f_{r_j + 1}} - \boldsymbol{f_{r_j}} \rvert$ for $j=1, 2, \ldots N$, 
        $\underline{f}_T := \inf_{j=1,\ldots,N} \{ L \left( \boldsymbol{\Delta_j} \right) \}$ and \sla{$\delta_T$ is as in \eqref{eq: def distance between cpts}.}
        Then there exist positive constants $C_1, C_2, C_3, C_4$ independent of $T$ and $d$, such that for $C_1 \sqrt{\log \left( T d^{1/4} \right)} \leq \zeta_{T, d} < C_2 \sqrt{\delta_T}\underline{f}_T$ and for sufficiently large $T$, we obtain
        \begin{equation} \label{main result mean multivariate}
        \mathbb{P}\Biggl( \hat{N} = N, \max_{j=1, 2, \ldots, N} \biggl\{ \left| \hat{r}_j - r_j \right| \alpha_j^{2} \biggl\} 
        \leq C_{3} \log \left( T d^{1/4} \right) \Biggl) \geq 1-\frac{C_{4}}{T},
        \end{equation}
        where, for $\left| \tilde{X}_{s,e}^{b,j}\right|$ as in \eqref{CUSUM_multiv} and $[s_j, e_j]$ the interval where $\hat{r}_j$ is obtained, we denote by $q_j := {\rm{argmax}}_{k=1,\ldots, d} \left| \tilde{X}^{\hat{r}_j, k}_{s_j, e_j} \right|$ and 
        \begin{equation*}
            \alpha_j = \left\{\begin{array}{cc}
                 \lvert f_{r_{j+1}, q_j} - f_{r_j, q_j} \rvert, & \textrm{for } L(\cdot) = L_\infty(\cdot) \\
                L^2 \left( \boldsymbol{\Delta}_j \right) , & \textrm{for } L(\cdot) = L_2(\cdot)
            \end{array} \right..
        \end{equation*}
    \end{theorem}

The outline of the proof of Theorem~\ref{thm: consistency_multivariate} is given in Appendix~\ref{appendix: proof_multivariate}.
Note that this extension of our algorithm is robust to spatial dependence. Theoretical results related to this statement, as well as simulation results can be found in \cite{anastasiou_generalized_2023}, where the extension of ID (\cite{anastasiou2022detecting}), namely the MID algorithm, is introduced.

\begin{algorithm}[htbp]
\caption{MDAIS}\label{alg:MDAIS}
\begin{algorithmic}
\State \textbf{function} MDAIS$\left( s, e, \lambda_{T}, \zeta_T, L(\cdot) \right)$
\If{$e-s<1$}
    \State STOP
\Else
    \State Set $d^{\text{multi}}_{s,e}$ as in \eqref{def_largest_diff_multi}
    \State For $j \in \{1,2,\dots,K\}$ let $I_j=[s_j, e_j]$ as in \eqref{intervals}
    \State $i=1$
    \State \textbf{(Main part)}
    \State $b_{i} = {\rm{argmax}}_{t\in [s_{i}, e_{i})} L \left( C_{s_{i}, e_{i}}^t\left(\boldsymbol{X}\right) \right)$
    \If{$L \left( C_{s_{i}, e_{i}}^{b_i}\left(\boldsymbol{X}\right) \right) > \zeta_T$}
        \State add $b_{i}$ to the list of estimated change-points
        \State MDAIS($s, s_{i}, \lambda_{T}, \zeta_T$)
        \State MDAIS($e_{i}, e, \lambda_{T}, \zeta_T$)
    \Else
        \State $i=i+1$
        \If{$i \leq K$}
            \State Go back to \textbf{(Main part)} and repeat
        \Else
        \State STOP
        \EndIf
    \EndIf
\EndIf
\State \textbf{end function}
\end{algorithmic}
\end{algorithm}

\section{Real data}\label{Real data}

\subsection{Crime data}
In this section, DAIS is applied to real data and the underlying signal is assumed to be piecewise-constant. 
The chosen dataset includes daily crime reports in Montgomery County, Maryland and can be found in \url{https://catalog.data.gov/dataset/crime}. 
We use daily observations starting from 22/01/2020, up to 31/08/2024. 
Since the data involve counts of crimes, they are positive integer-valued and thus a transformation is required to bring them closer to Gaussian data with constant variance. 
This is done using the Anscombe transform (\cite{Anscombe1948}) $\alpha: \mathbb{N} \rightarrow \mathbb{R}$, with $\alpha (x) = 2\sqrt{x+3/8}$. 

The top plot of Figure \ref{fig:crime_comparisons} is a plot of the real data (in black), along with the estimated underlying signal, $\hat{f}_t$, according to the change-points detected by DAIS (plotted in red).  We now attempt to provide a possible explanation about the change-points estimated by DAIS that express the most important movements in the data. 
The first change-point occurs on the $17^{\textrm{th}}$ of March 2020 and corresponds to the first days of positive cases of COVID-19 when the first rules were imposed to limit the spread of the virus among the community in the state of Maryland.
As expected, since people were forced to stay at home, the number of crimes reported dropped significantly. 
The second change-point is on the $28^{\textrm{th}}$ of April 2020, when the official authorities started lifting the restrictions and the number of crimes increased.
Some small decrease around the $11^{\textrm{th}}$ of December 2020 can be explained by people's behavioral change around the holidays. 
The number of crimes returns close to the previous levels, on the $28^{\textrm{th}}$ of May 2021 and a further increase is observed at the end of the summer, $27^{\textrm{th}}$ of August 2021. 
This could possibly be due to people going on holidays out of town, leaving their homes more vulnerable to break-ins. 
As has also been observed for 2020, there is again a reduced number of reported crimes around the Christmas holidays of 2021, starting from the $23^{\textrm{th}}$ of December 2021 until the $17^{\textrm{th}}$ of January 2022. 
It is notable that the next change-point detected is almost a year later, as there is a sudden decrease on the $23^{\textrm{rd}}$ of December 2022, with an increase on the $26^{\textrm{th}}$ of the same month. 
An increase in the daily number of crimes reported is also detected on the $28^{\text{th}}$ of September 2023. 
As in the previous years, there is a drop on the $23^{\textrm{rd}}$ of December 2023 and an increase to a slightly higher average than before on the $23^{\textrm{rd}}$ of January 2024. 
Further increase to the number of daily crimes reported is observed on the $26^{\textrm{th}}$ of February and the $10^{\textrm{th}}$ of April 2024, with a decrease on the $4^{\textrm{th}}$ of June 2024.
It is clear that the average number of crimes has increased in the last years.
It is important to note that there is a clear trend as change-points were detected by DAIS on 11/12/2020, 23/12/2021, 23/12/2022 and 23/12/2023. 
All four are near the Christmas holidays and a decrease in the daily number of crimes reported is noticed, followed by the trend increasing on the 17/01/2022, 26/12/2022 and 22/01/2024.

The fits obtained by three competitors, ID, MOSUM, and NOT, are shown in the bottom plot of Figure~\ref{fig:crime_comparisons}, plotted with blue, yellow and green lines, respectively. 
All methods obtain different signals on the same data, with ID and DAIS providing similar fits and NOT having minor differences with them.
ID and NOT detect 11 and 10 out of the 15 change-points that DAIS detects, respectively.
They both do not detect the change-points in May 2021 and the last 2 change-points detected by DAIS, in April and June of 2024.
ID also does not detect the change-point just before Christmas 2020, while NOT does not detect the 2 change-points around Christmas 2022.
MOSUM detects only 5 change-points, with the only change-point around the Christmas period being the one in 2020, which ID missed.
By detecting change-points around the same period every year, we could argue that DAIS is detecting the seasonality around the Christmas holidays, which none of the other competitors managed to fully capture.

\begin{figure}[htbp]
\centering
\includegraphics[scale=0.7]{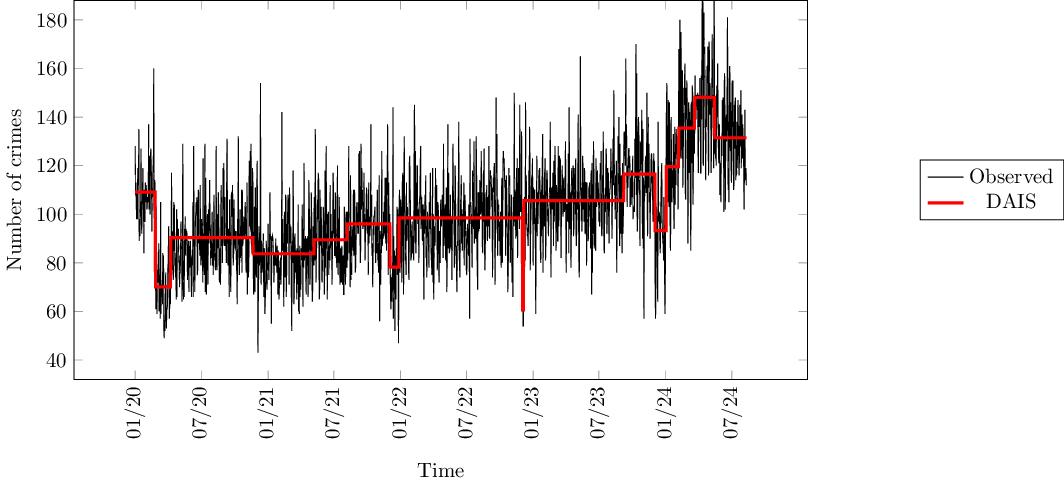} 
\includegraphics[scale=0.7]{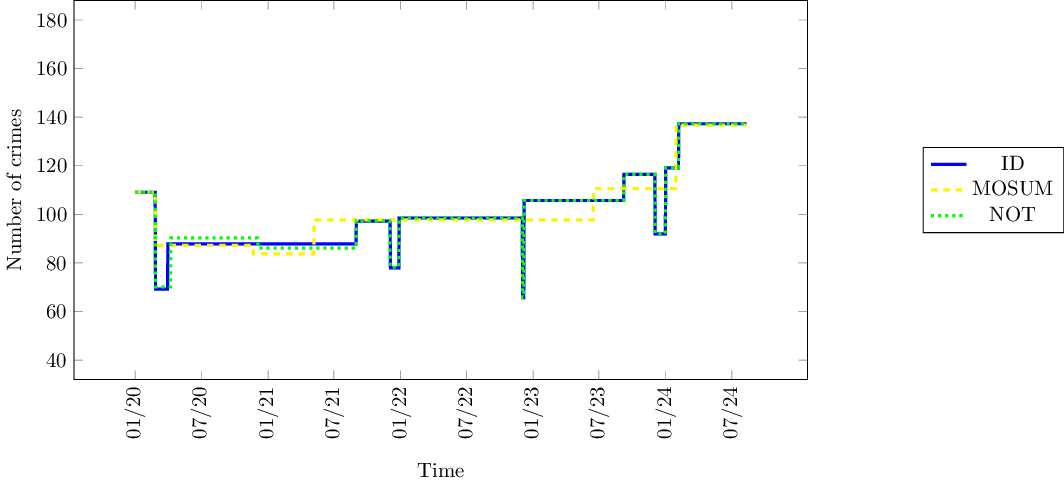}
\caption{\label{fig:crime_comparisons} The observed time series data of daily crime reports in Montgomery County and fitted piecewise-constant mean signal obtained by applying DAIS (red), ID (blue), MOSUM (yellow) and NOT (green) to the data. The x-axis labels correspond to the month and year in an MM/YY format.}
\end{figure}

\subsection{Euro to British pound exchange rate}
The DAIS algorithm for changes in the slope is applied to the Euro to British pound exchange rate. 
The data are the weekly close prices adjusted for splits for the period between 07/04/2014 and 12/09/2024, \slc{obtained from} \url{https://finance.yahoo.com/quote/EURGBP%3DX/}. 
The results can be seen in Figure \ref{fig:EURGBP}. 
The first plot is the observed signal, while the other plots are the estimated piecewise-linear signals from DAIS, ID and CPOP.
The change-points which are discussed in the next paragraph are indicated by dashed lines in the plot of the observed signal.

All three fits are very similar, with ID detecting the least number of change-points (27) and thus having a smoother signal, while DAIS and CPOP detect 42 and 60, respectively. 
The first important change-point, detected only by DAIS, on the week of $5^\textrm{th}$ of January of 2015 is the same week as the Charlie Hebdo shooting, the terrorist attack that occurred at the offices of the satirical weekly newspaper in Paris.
After this, the exchange rate decreased, indicating that the attack had negative effects to the strength of the Euro.
The next important change-point is an increase in the value of the exchange rate starting on the week of $9^\textrm{th}$ of November 2015.
This change-point is detected by both DAIS and CPOP, while ID detects one just one week later.
The change might have been caused by Prime Minister David Cameron's speech on the $10^\textrm{th}$ of November, during which he repeated his commitment to holding a referendum for Brexit before the end of 2017.
A few days later, on the $13^\textrm{th}$, the terrorist attacks in Paris took place. 
Both events had an impact on the exchange rate.
The next change-point, from which a new increase in the observed value begins, is detected in the week of $13^\textrm{th}$ of June 2016, which is the week just before the referendum was held, on the $23^\textrm{rd}$ of June.
It is worth noting that the drop in the week of the $16^\text{th}$ of March of 2020, that is detected by DAIS and CPOP, coincides with the day that the Prime Minister Boris Johnson announced the government would be implementing measures intended to halt the spread of the COVID-19 virus.
The big drop around June of 2022, with change-points detected by all methods in two consecutive weeks, on the $6^\text{th}$ and $13^\text{th}$, falls around the time when fuel prices soared, according to \url{https://www.theguardian.com/theguardian/2022/jun/09}.

\begin{figure}[tbp]
\centering
\includegraphics[scale = 0.55]{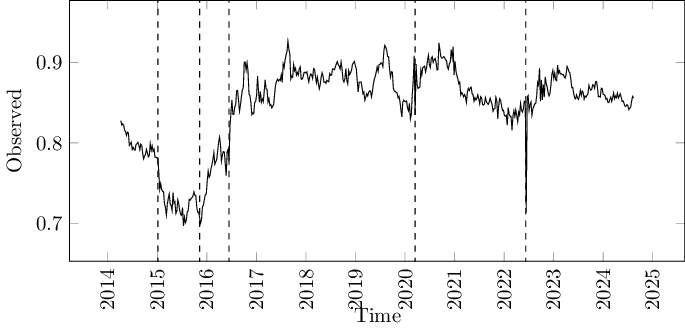}
\includegraphics[scale = 0.55]{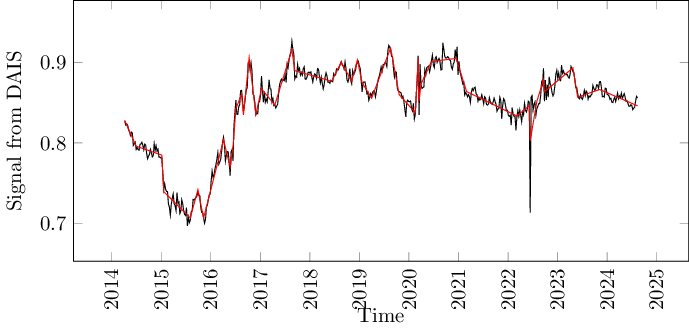}
\includegraphics[scale = 0.55]{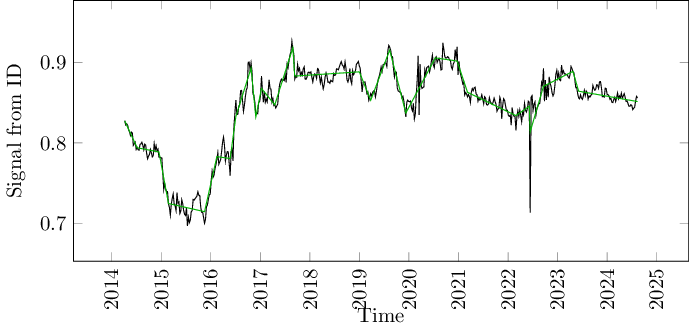}
\includegraphics[scale = 0.55]{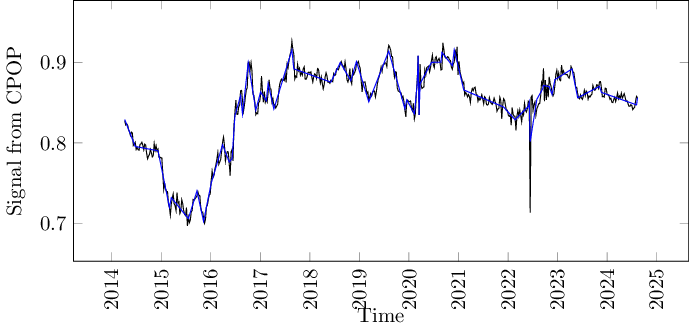}
\caption{\label{fig:EURGBP} The observed time series data of Euro to British pound sterling exchange rate and fitted piecewise-linear signal obtained by applying DAIS (red), ID (green) and CPOP (blue). The x-axis labels correspond to the year.}
\end{figure}

\section{Conclusions}\label{Conclusions}

In this paper, DAIS is introduced, which is a new data-adaptive method for detecting structural changes in a given data sequence. 
The first step of the method is to identify the location of the largest difference in the sequence, as defined in \eqref{def_largest_diff} for the cases of changes in the mean and changes in the slope, in an effort to start searching from an area around the most apparent change-point in the case that there is one. 
This step is what makes our method data-adaptive and differentiates it from the competitors.
Expanding intervals around the location of the largest difference, in absolute value, are used to achieve isolation of the change-points, which enhances the detection power. 
Theoretical results regarding the consistency of the number of change-points detected and the accuracy of their estimated locations are provided.
The simulation results presented
indicate that DAIS is at least as accurate as the state-of-the-art competitors, \sla{exhibiting uniformly good behavior across different signal structures}.
Its data-adaptive nature is advantageous regarding the detection of the true change-points in difficult structures.
DAIS uses intervals around the location of the largest difference and so the change-points are, with high probability, near the midpoint of these intervals, which results in the power of the contrast function being maximized. 
\sla{Such an interval expansion approach around the location of the true change-points, offers also an advantage in terms of the computational complexity of the algorithm, especially in cases with few change-points.}
For the proof of the theoretical properties of DAIS we require that the location of the largest difference is at a location where detection, while the change-point is isolated in an interval, is guaranteed. 
\sla{However, in practice, the change-points can still potentially get detected even when this assumption fails.
Detection could occur in an interval in which the change-point is isolated, but whose end-points are not far enough from the true location of the change-point to guarantee detection.
It could also occur in an interval in which the change-point is not isolated.
In such intervals, as described in both cases, the detection power is reduced, but detection can still occur.
}

\sla{Finally, we consider some extensions of DAIS to more difficult structures. 
This includes heavy tailed noise and temporal dependence.
The case of multivariate signals is also discussed and further analyzed, with the algorithm MDAIS being proposed. Under some assumptions that are equivalent to the ones for the univariate signals, a consistency result is provided for the case of possible changes in the mean vector of an underlying $d-$dimensional signal.
}

\


\begin{appendix}

\section{Signals used in the simulation study} \label{simulations_supplement}
The signals used in the simulations of Section \ref{simulations} are reported below. 
For the case of changes in the mean (piecewise-constancy), the signals $f_t$ in \eqref{model} are the following:
\begin{itemize}
    \item[(S1)] $small\_dist$: sequence of length 1000 with 2 change-points at 485 and 515 with values between change-points 0, 1, 0. The standard deviation is $\sigma=1$. 

     \item[\sla{(S2)}] \sla{$small\_dist2$ piecewise-constant signal of length 135 with 2 change-points at 30, 35 with values between change-points 0, 2.3, 8. The standard deviation is $\sigma=1$. }
     
    \item[{(S3)}] $stairs$: piecewise-constant signal of length 150 with 14 change-points at 10, 20, 30, 40, 50, 60, 70, 80, 90, 100, 110, 120, 130, 140 with values between change-points 1, 2, 3, $\ldots$, 15. The standard deviation is $\sigma=0.3$. 
    
    \item[{(S4)}] $mix$: piecewise-constant signal of length 301 with 9 change-points at 11, 21, 41, 61, 91, 121, 161, 201, 251 with values between change-points 7, -7, 6, -6, 5, -5, 4, -4, 3, -3. The standard deviation is $\sigma=4$. 
    
    \item[{(S5)}] $mix2$: piecewise-constant signal of length 75 with 11 change-points at 5, 12, 17, 25, 31, 38, 44, 50, 56, 61, 67 with values between the change-points 0, 5, 0, 6, 0, 4, 0, 5, 0, 6, 0, 4. The standard deviation is $\sigma = 1$.
    
    \item[{(S6)}] $many\_cpts$: piecewise-constant signal of length 700 with 99 change-points at 7, 14, $\ldots$, 693 with values between change-points 0, 4, 0, 4, $\ldots$, 0, 4. The standard deviation is $\sigma=1$. 

     \item[{(S7)}] $many\_cpts\_long$: piecewise constant signal of length 600 with 119 change-points at 5, 10, 15, $\ldots$, 595 with values between change-points 0, 5, 0, $\ldots$, 0, 5. The standard deviation is $\sigma = 1$.

     \item[{(S8)}] $simple\_signal$: sequence of length 1100 with 1 change-point at 550 with values between change-points 0, 2. The standard deviation is $\sigma=1$.

\end{itemize}

Under the case of continuous, piecewise-linearity, the signals $f_t$ used are:
\begin{itemize}
    \item[{(S9)}] $wave1$: continuous, piecewise-linear signal of length 1408 with 7 change-points at 256, 512, 768, 1024, 1152, 1280, 1344 with changes in slope $-1/64$, $2/64$, $-3/64$, $4/64$, $-5/64$, $6/64$, $-7/64$. The starting intercept is $f_1=1$ and slope $f_2-f_1=1/256$. The standard deviation is $\sigma=1$. 

    \item[{(S10)}] $wave2$: continuous, piecewise-linear signal of length 1500 with 99 change-points at 15, 30, $\ldots$, 1485 with changes in slope -1, 1, -1 $\ldots$, -1. The starting intercept is $f_1=-1/2$ and slope $f_2-f_1=1/40$. The standard deviation is $\sigma=1$.

    \item[{(S11)}] $wave3$: continuous, piecewise-linear signal of length 840 with 119 change-points at 7, 14, 21, $\ldots$, 833 with changes in slope -1, 1, -1 $\ldots$, -1. The starting intercept is $f_1=-1/2$ and slope $f_2-f_1=1/32$. The standard deviation is $\sigma=0.3$. 
\end{itemize}
\section{Proof of Theorem \ref{consistency_theorem}} \label{proofs}
Before proving Theorem \ref{consistency_theorem}, we introduce some more notation. 
For convenience, we denote \eqref{CUSUM} by
    $\Tilde{X}^b_{s,e} = \sqrt{\frac{e-b}{\ell(b-s+1)}}\sum^b_{t=s}X_t - \sqrt{\frac{b-s+1}{\ell(e-b)}}\sum^{e}_{t=b+1}X_t.$
Similarly, we denote the CUSUM of the unobserved signal as $ \Tilde{f}^b_{s,e} = \sqrt{\frac{e-b}{\ell(b-s+1)}}\sum^b_{t=s}f_t - \sqrt{\frac{b-s+1}{\ell(e-b)}}\sum^{e}_{t=b+1}f_t$. 
\sla{We also work under \eqref{eq: def minimum magnitude of change}, \eqref{eq: def magnitude of change} and \eqref{eq: def distance between cpts}.}
For the following proof, the contrast vector \linebreak $\boldsymbol{\psi_{s,e}^b} = \left( \psi_{s,e}^b(1), \psi_{s,e}^b(2), \dots, \psi_{s,e}^b(T_1) \right)$ is defined through the contrast function
\begin{align} \label{psi_definition}
    \psi_{s,e}^b(t)=\left\{\begin{array}{ll}
    \sqrt{\frac{e-b}{\ell(b-s+1)}}, &t =s, s+1, \dots,b, \\
    -\sqrt{\frac{b-s+1}{\ell(e-b)}}, &t =b+1, b+2, \dots,e, \\
    0, &\text{otherwise}
    \end{array}\right.,
\end{align}
where $s\leq b < e$ and $\ell=e-s+1$. 
Notice that for any vector $\boldsymbol{v} = \left(v_1, v_2, \dots, v_{T_1}\right)$, we have that $\langle \boldsymbol{v}, \boldsymbol{\psi_{s,e}^b} \rangle = \tilde{v}_{s,e}^b$. \\

For the proof of Theorem \ref{consistency_theorem}, we require the following Lemma.
\begin{lemma} \label{Lemma_for_main_thm}
Suppose $\boldsymbol{f} = \left( f_1, f_2, \dots, f_{T} \right)^T$ is a piecewise-constant vector. 
Pick any interval $[s,e] \subset [1,T]$ such that $[s,e-1]$ contains exactly one change-point $r$. 
Let $\rho = \left| r-b \right|, \Delta^f = \left| f_{r+1} - f_r \right|,
\eta_L = r-s+1$ and $\eta_R = e-r$. Then,
\begin{equation*}
    \Bigl\|\boldsymbol{\psi_{s,e}^b}\langle\boldsymbol{f},\boldsymbol{\psi_{s,e}^b}\rangle - \boldsymbol{\psi_{s,e}^{r}}\langle\boldsymbol{f},\boldsymbol{\psi_{s,e}^{r}}\rangle\Bigr\|_{2}^2 
    = \left( \tilde{f}_{s,e}^r \right)^2 - \left( \tilde{f}_{s,e}^b \right)^2.
\end{equation*}
In addition,
\begin{enumerate}
    \item for any $r \leq b < e$, $\left( \tilde{f}_{s,e}^r \right)^2 - \left( \tilde{f}_{s,e}^b \right)^2 = \left( \rho \eta_L / (\rho + \eta_L)\right)\left( \Delta^f \right)^2$;
    
    \item for any $s \leq b < r$, $\left( \tilde{f}_{s,e}^r \right)^2 - \left( \tilde{f}_{s,e}^b \right)^2 = \left( \rho \eta_R / (\rho + \eta_R)\right)\left( \Delta^f \right)^2$.
\end{enumerate}
\end{lemma}
\begin{proof}
    See Lemma 4 from \cite{baranowski2019narrowest}.
\end{proof}

The proof of Theorem \ref{consistency_theorem} consists of 6 steps. 
Step 1 is to show that the observed $\lvert \tilde{X}_{s,e}^b \rvert$ is uniformly close to the unobserved $\lvert \tilde{f}_{s,e}^b \rvert$ for all $1\leq s \leq b < e \leq T$. 
This will allow us to extend some results that will be derived for the signal, $f_t$, to the data sequence, $X_t$, in which we are interested. 
In Step 2, we control the distance between $\lvert \tilde{X}_{s,e}^{b_1} \rvert - \lvert \tilde{X}_{s,e}^{b_2} \rvert$ and $\lvert \tilde{f}_{s,e}^{b_1} \rvert - \lvert \tilde{f}_{s,e}^{b_2} \rvert$ for all possible combinations of $s, e, b_1, b_2$, where $1 \leq s < e \leq T$ and $b_1, b_2 \in [s, e)$. 
In Step 3, we show that it suffices to restrict the proof to an interval with a single change-point because each change-point will be isolated in an interval where detection will occur with high probability, as discussed in Section \ref{discussion}. 
In this step, we also show that the estimated change-point $\hat{r}_j$ will be close to the actual change-point $r_j$. 
Since after detection DAIS restarts in intervals with end- (or start-) point the start- (or end-) point of the interval where the detection occurred, in Step 4 we prove that there is no change-point, besides the detected one, in the intervals that are skipped, with probability 1. 
In Step 5, we show that the new intervals used after detection allow for the detection of all the remaining change-points. 
Finally, in Step 6 we show that when there is no change-point in the interval being checked, the algorithm will not have any false detections and will terminate.
\vspace{0.1in}
\\
\begin{proof}
We will prove the more specific result
\begin{equation}
\label{mainresult_theorem}
\mathbb{P}\Biggl( \hat{N} = N, \max_{j=1, 2, \ldots, N} \biggl( \left| \hat{r}_j - r_j \right| \left( \Delta^{f}_j \right)^{2} \biggl) \leq C_{3} \log T \Biggl) \geq 1 - \frac{1}{6\sqrt{\pi}{T}},
\end{equation}
which implies result \eqref{main result mean}.
\vspace{0.1in}
\\
\textbf{Step 1:} Allow us to denote by
\begin{equation}
\label{A_T}
A_{T} = \left\lbrace \max_{s,b,e: 1\leq s \leq b < e \leq T}\left|\tilde{X}_{s,e}^b - \tilde{f}_{s,e}^b\right|\leq \sqrt{8\log T} \right\rbrace.
\end{equation}
We will show that $\Prob\left(A_{T} \right) \geq 1-1/(12\sqrt{\pi}{T})$. 
From \eqref{model_sigma1} and \eqref{CUSUM}, simple steps yield $\tilde{X}_{s,e}^b - \tilde{f}_{s,e}^b = \tilde{\epsilon}_{s,e}^b$, where $\tilde{\epsilon}_{s,e}^b \sim \mathcal{N}(0,1)$. 
Thus, for $Z \sim \mathcal{N}(0,1)$, using the Bonferroni inequality we get that 
\begin{align}
\nonumber \Prob\Bigl((A_{T})^{c}\Bigr) &= \Prob\left(\max_{s,b,e: 1\leq s \leq b < e \leq {T}}\left|\tilde{X}_{s,e}^b - \tilde{f}_{s,e}^b\right| > \sqrt{8\log {T}}\right)\\
\nonumber & \leq \sum_{1\leq s\leq b <e \leq {T}}\Prob\left(\left|\tilde{\epsilon}_{s,e}^b\right|>\sqrt{8\log {T}}\right) \leq \frac{{T}^3}{6}\Prob(|Z|>\sqrt{8\log {T}})\\
\nonumber & = \frac{{T}^3}{3}\Prob\left(Z>\sqrt{8 \log {T}}\right)\leq \frac{{T}^3}{3}\frac{\phi(\sqrt{8\log {T}})}{\sqrt{8\log {T}}} \leq \frac{1}{12\sqrt{\pi}{T}},
\end{align}
where $\phi(\cdot)$ is the probability density function of the standard normal distribution.
\vspace{0.1in}
\\
{\textbf{Step 2:}} For intervals $[s,e)$ that contain only one true change-point $r$, for $\boldsymbol{\psi_{s,e}^b}$ as defined in \eqref{psi_definition}, we denote by
\begin{equation}
\label{B_T}
B_{T} = \left\lbrace \max_{j=1,2,\ldots,N} \max_{\substack{r_{j-1}<s\leq r_j\\r_j < e \leq r_{j+1}\\s\leq b < e}}
\frac{\left|\left\langle\boldsymbol{\psi_{s,e}^b}\langle\boldsymbol{f},\boldsymbol{\psi_{s,e}^b}\rangle - \boldsymbol{\psi_{s,e}^{r}}\langle\boldsymbol{f},\boldsymbol{\psi_{s,e}^{r}}\rangle,\boldsymbol{\epsilon}\right\rangle\right|}{\|\boldsymbol{\psi_{s,e}^b}\langle\boldsymbol{f},\boldsymbol{\psi_{s,e}^b}\rangle - \boldsymbol{\psi_{s,e}^{r}}\langle\boldsymbol{f},\boldsymbol{\psi_{s,e}^{r}}\rangle\|_{2}}\leq \sqrt{8 \log T}\right\rbrace.
\end{equation}
Because 
\begin{equation*}
    \left| \left \langle \boldsymbol{\psi_{s,e}^b} \langle \boldsymbol{f}, \boldsymbol{\psi_{s,e}^b} \rangle - \boldsymbol{\psi_{s,e}^{r_j}} \langle \boldsymbol{f}, \boldsymbol{\psi_{s,e}^{r_j}} \rangle \right \rangle \right|/ \left(\| \boldsymbol{\psi_{s,e}^b} \langle \boldsymbol{f}, \boldsymbol{\psi_{s,e}^b} \rangle - \boldsymbol{\psi_{s,e}^{r_j}} \langle \boldsymbol{f}, \boldsymbol{\psi_{s,e}^{r_j}} \rangle\|_{2} \right)
\end{equation*} 
follows the standard normal distribution, we use a similar approach as in Step 1, to show that 
$\Prob\Bigl(\left(B_{T}\right)^{c}\Bigr) \leq {1}/{12\sqrt{\pi}{T}}$. 
Therefore, Step 1 and Step 2 lead to
\begin{equation}
\nonumber \Prob\left(A_{T} \cap B_{T} \right) \geq 1 - \frac{1}{6\sqrt{\pi} {T}}.
\end{equation}
{\textbf{Step 3:}} This is the main part of our proof. 
From now on, we assume that $A_T$ and $B_T$ both hold. 
The constants we use are
\begin{equation}
    C_1 = \sqrt{C_3} + \sqrt{8}, C_2 = \frac{1}{\sqrt{4n}} - \frac{2\sqrt{2}}{\underline{C}}, C_3 = 2(2\sqrt{2}+4)^2,
\end{equation}
where $\underline{C}$ satisfies Assumption (A2), $\sqrt{\delta_T}\underline{f}_T \geq \underline{C} \sqrt{\log T}$
and \slb{$n \geq 3/2$} as defined in Section \ref{discussion}.
For $j \in \{1, 2, \ldots, N\}$ define $I^L_j$ and $I^R_j$ as in \eqref{intervals_discussion}.
The location of the largest difference in an interval $[s,e]$, $1\leq s<e\leq T$, is defined as $d_{s,e} = \textrm{argmax}_{t \in \{s, s+1, \dots, e-1\}}\{\lvert X_{t+1} - X_{t} \rvert\}$ as in \eqref{def_largest_diff} and for $K^l = \lceil \frac{d_{s,e}-s+1}{\lambda_T} \rceil$, $K^r = \lceil \frac{e-d_{s,e}+1}{\lambda_T} \rceil$, $K^{\max} = \max\{K^l, K^r\}$, define
\begin{align} \label{end-points_proof}
    & c_{m}^l =
    \max\{d_{s,e}-m\lambda_T, s\}, \quad m =0, 1, \dots,K^{\max}
    , \nonumber \\
    & c_{k}^r = 
    \min\{d_{s,e}+k\lambda_T-1,e\}, \quad k =1,2, \dots,K^{\max}.
\end{align}
\slb{Since the length of the intervals in \eqref{intervals_discussion} is $\delta_T/2n$,
$\lambda_T \leq \delta_T/2n$ ensures that there exists at least one $m\in \{0,1,\ldots,K^{\max}\}$ and at least one $k\in \{1,2\ldots,K^{\max}\}$} such that $c_m^l \in I^L_j$ and $c_k^r \in I^R_j$ for all $j \in \{1, 2, \ldots, N\}$ and for all possible locations of the largest difference.
\vspace{0.1in}
\\
At the beginning of DAIS, $s=1, e=T$ and the first change-point that will get detected depends on the value of $d_{1,T}$. 
As already explained in Section \ref{discussion}, for $j\in\{1,\ldots,N-1\}$, the largest difference $d_{s,e}$ will be at most at a distance $\delta_{1,T}^j$ from the nearest change-point $r_j$ or $r_{j+1}$, where $\delta_{1,T}^j \leq \frac{\Tilde{\delta}_j}{2} - \frac{3\delta_T}{4n}$ for $\Tilde{\delta}_j = r_{j+1} - r_j$. 
The first point to get detected will be the point that is closest to the largest difference $d_{1,T}$. 
The interval where the detection of this change-point occurs, cannot contain more than one change-points, as was explained in Section \ref{discussion} and is proved at \eqref{isolation}.
\vspace{0.1in}
\\
We will show that there exists an interval $[c_m^l, c_k^r]$, for 
\sla{$m\in \{0,1,\ldots,K^{\max}\}$ and $k\in \{1,2\ldots,K^{\max}\}$},
such that the first change-point to get detected, $r_J$, is isolated, assuming that $d_{1,T} > r_J$. 
\sla{A similar} approach works for $d_{1,T} \leq r_J$. 
\sla{The differences arise from the fact that the expansions are always performed starting from the right, which implies that if the change-point is to the right of $d_{1,T}$, an uneven number of expansions will be performed before the end-point of the interval lies in \eqref{intervals_discussion}.}
Without loss of generality, we suppose that
$J \in \left\lbrace 1,2,\ldots, N-1 \right\rbrace$.
The result in the case that $J=N$ is easier to prove and so it is skipped, \sla{since, as can be seen by \eqref{eq: assumption_largest_diff}, the change-point will certainly be isolated when detection occurs if $d_{1,T} > r_N$}.
Note that since the closest change-point to the largest difference is $r_J$, it must hold that $\delta_{1,T}^J = d_{1,T}-r_J$.
We are now considering the largest possible interval that may be checked before detection occurs in the sense that the largest number of expansions will be performed.
Detection is guaranteed when checking this interval, but it may occur in any interval smaller than this.
Showing that the change-point will be isolated in the largest interval means that it will \sla{also} be isolated \sla{in any of the smaller intervals}.
Since we are considering $d_{1,T} > r_J$, the worst case scenario, \sla{(meaning when DAIS checks the largest interval)} occurs for $m=k$, 
$c_m^l \in I^L_J$ and $c_k^r - r_J > \delta_T / 2n$, but it does not necessarily hold that $c_k^r \in I^R_J$. So, $r_J - c_m^l \leq \delta_T/n$ and $c_k^r - d_{1,T} + 1= d_{1,T} - c_m^l$.
It follows that
\begin{align}\label{isolation}
    c_k^r - c_m^l 
    & = c_k^r - d_{1,T} + d_{1,T} - c_m^l 
    = 2(d_{1,T} - c_m^l)-1 
    < 2(d_{1,T} - r_J + r_J - c_m^l) \nonumber\\
    & < 2 \Bigl(\delta_{1,T}^J + \frac{\delta_T}{n}\Bigr) 
    \leq \Tilde{\delta}_J + \frac{\delta_T}{2n},
\end{align}
\sla{where we used \eqref{delta upper bound} in the last step.}
Now, since $r_J - c_m^l \leq \delta_T/n$, there can be at most one change-point in the interval $[c_m^l, r_J]$. 
Considering the interval $[r_J, c_k^r]$, using \eqref{isolation} and $r_J - c_m^l > \delta_T/2n$, it holds that $c_k^r - r_J = c_k^r - c_m^l + c_m^l - r_J < \Tilde{\delta}_J$, so $c_k^r < r_{J+1}$ and $r_J$ is isolated in $[r_J, c_k^r]$.
The two results combined prove that the change-point is isolated in the interval $[c_m^l, c_k^r]$.
\vspace{0.1in}
\\
Now, we consider $J\in\{1,\ldots,N\}$.
We will show that for $\tilde{b} = \textrm{argmax}_{c_m^l\leq t < c_k^r} \lvert \tilde{X}^t_{c_m^l, c_k^r} \rvert$, it holds that $\Big| \tilde{X}^{\tilde{b}}_{c_m^l, c_k^r} \Big| > \zeta_T$. 
Using \eqref{A_T}, we have that
\begin{equation} \label{initial_bound}
    \Big| \tilde{X}^{\tilde{b}}_{c_m^l, c_k^r} \Big| \geq \Big| \tilde{X}^{r_J}_{c_m^l, c_k^r} \Big| \geq \Big| \tilde{f}^{r_J}_{c_m^l, c_k^r} \Big| - \sqrt{8\log T}.
\end{equation}
But, 
\begin{align} \label{unobserved_cusum}
    \Big| \tilde{f}^{r_J}_{c_m^l, c_k^r} \Big| &= \Bigg| \sqrt{\frac{c_k^r - r_J} {(c_k^r-c_m^l+1) (r_J - c_m^l + 1)}} (r_J - c_m^l + 1) f_{r_J} \\
    & \hspace{3cm} - \sqrt{\frac{r_J - c_m^l + 1} {(c_k^r-c_m^l+1) (c_k^r - r_J)}} (c_k^r - r_J) f_{r_J+1} \Bigg| \nonumber\\
    & = \Bigg| \sqrt{\frac{(c_k^r - r_J) (r_J - c_m^l + 1)} {(c_k^r-c_m^l+1)}} f_{r_J} - \sqrt{\frac{(r_J - c_m^l + 1) (c_k^r - r_J)} {(c_k^r-c_m^l+1)}} f_{r_J+1} \Bigg| \nonumber \\
    & = \sqrt{\frac{(c_k^r - r_J) (r_J - c_m^l + 1)} {c_k^r-c_m^l+1}} \Delta^f_J \nonumber
    \geq \sqrt{\frac{(c_k^r - r_J) (r_J - c_m^l + 1)} {2\max\{c_k^r-r_J,r_J-c_m^l+1\}}}  \Delta^f_J \\
    & = \sqrt{\frac{\min\{c_k^r-r_J,r_J-c_m^l+1\}}{2}} \Delta^f_J.
\end{align}
If we show that 
\begin{equation} \label{bound_for_min}
    \min\{c_k^r-r_J,r_J-c_m^l+1\}\geq \frac{\delta_T}{2n}
\end{equation}
then, using Assumption (A2) and the results in \eqref{initial_bound} \eqref{unobserved_cusum}, \eqref{bound_for_min}, we have that
\begin{align} \label{proof_of_detection}
    \Big| \tilde{X}^{\tilde{b}}_{c_m^l, c_k^r} \Big| & \geq \sqrt{\frac{\delta_T}{4n}}\Delta^f_J - \sqrt{8\log T}
    \geq \sqrt{\frac{\delta_T}{4n}}\underline{f_T} - \sqrt{8\log T} \nonumber\\
    & = \biggl( \frac{1}{\sqrt{4n}} - \frac{2\sqrt{2\log T}}{\sqrt{\delta_T}\underline{f}_T} \bigg) \sqrt{\delta_T} \underline{f}_T
    \geq \biggl( \frac{1}{\sqrt{4n}} - \frac{2\sqrt{2}}{\underline{C}} \bigg) \sqrt{\delta_T} \underline{f}_T \nonumber\\
    & = C_2 \sqrt{\delta_T} \underline{f}_T > \zeta_T
\end{align}
and thus, the change-point will get detected.
\vspace{0.1in}
\\
But \eqref{bound_for_min} holds as, for $m=k$, it holds that $c_k^r - d_{1,T} + 1 = d_{1,T} - c_m^l$ and so
\begin{itemize}
    \item If $d_{1,T} \leq r_J$, then $c_k^r \in I^R_J$ and so $c_k^r-r_J > \delta_T/2n$ and $r_J - c_m^l + 1 = r_J - d_{1,T} + d_{1,T} - c_m^l + 1 > c_k^r - d_{1,T} \geq c_k^r-r_J > \delta_T/2n$.

    \item If $d_{1,T} > r_J$, then $c_m^l \in I^L_J$ and so $r_J-c_m^l+1 > \delta_T/2n$ and $c_k^r - r_J = c_k^r - d_{1,T} + d_{1,T} - r_J > d_{1,T} - c_m^l - 1 \geq r_J - c_m^l \geq \delta_T/2n$.
\end{itemize}
\vspace{0.1in}
Therefore, we have proved that there will be an interval of the form $[c_{\tilde{m}}^l, c_{\tilde{k}}^r]$, such that the interval contains $r_J$ and no other change-point and $\max_{c_{\tilde{m}}^l \leq t < c_{\tilde{k}}^r} \Big| \tilde{X}^t_{c_{\tilde{m}}^l, c_{\tilde{k}}^r} \Big| > \zeta_T$. 
For $k^{\ast} \in \{1, 2, \ldots K^{\max}\}$ and $m^{\ast} \in \{k^{\ast}-1, k^{\ast}\}$, denote by $c_{m^{\ast}}^l \geq c_{\tilde{m}}^l$ and $c_{k^{\ast}}^r \leq c_{\tilde{k}}^r$ the first left- and right-expanding points, respectively, that this happens and let $b_J = \textrm{argmax}_{c_{m^{\ast}}^l \leq t < c_{k^{\ast}}^r} \Big| \tilde{X}^t_{c_{m^{\ast}}^l, c_{k^{\ast}}^r} \Big|$, with $\Big| \tilde{X}^{b_J}_{c_{m^{\ast}}^l, c_{k^{\ast}}^r} \Big| > \zeta_T$. 
Note that $b_J$ cannot be an estimation of $r_j, j\neq J$, as $r_J$ is isolated in the interval where it is detected. 
Our aim now is to find $\gamma_T > 0$, such that for any $b^{\ast} \in \{c_{m^{\ast}}^l, c_{m^{\ast}}^l + 1, \ldots, c_{k^{\ast}}^r-1\}$ with $\lvert b^{\ast} - r_J\rvert \Bigl( \Delta_J^f \Bigl)^2 > \gamma_T$, we have that
\begin{equation}\label{contradiction}
    \Bigl( \tilde{X}^{r_J}_{c_{m^{\ast}}^l, c_{k^{\ast}}^r} \Bigl)^2 > \Bigl( \tilde{X}^{b^{\ast}}_{c_{m^{\ast}}^l, c_{k^{\ast}}^r} \Bigl)^2.
\end{equation}
Proving \eqref{contradiction} and using the definition of $b_J$, we can conclude that \linebreak $\lvert b_J - r_J\rvert \Bigl( \Delta_J^f \Bigl)^2 \leq \gamma_T$. 
Now, using \eqref{model_sigma1}, it can be shown, for $\boldsymbol{\psi_{s,e}^{b}}$ as defined in \eqref{psi_definition}, that \eqref{contradiction} is equivalent to
\begin{align} \label{equivalent_contradiction}
    \Bigl( \tilde{f}^{r_J}_{c_{m^{\ast}}^l, c_{k^{\ast}}^r} \Bigl)^2 - \Bigl( \tilde{f}^{b^{\ast}}_{c_{m^{\ast}}^l, c_{k^{\ast}}^r} \Bigl)^2 & > 
    \Bigl( \tilde{\epsilon}^{b^{\ast}}_{c_{m^{\ast}}^l, c_{k^{\ast}}^r} \Bigl)^2 - \Bigl( \tilde{\epsilon}^{r_J}_{c_{m^{\ast}}^l, c_{k^{\ast}}^r} \Bigl)^2 \nonumber\\ 
    &+ 2\Big \langle \boldsymbol{\psi_{c_{m^{\ast}}^l, c_{k^{\ast}}^r}^{b^{\ast}}} \langle \boldsymbol{f}, \boldsymbol{\psi_{c_{m^{\ast}}^l, c_{k^{\ast}}^r}^{b^{\ast}}} \rangle - \boldsymbol{\psi_{c_{m^{\ast}}^l, c_{k^{\ast}}^r}^{r_J} } \langle \boldsymbol{f}, \boldsymbol{\psi_{c_{m^{\ast}}^l, c_{k^{\ast}}^r}^{r_J} } \rangle, \boldsymbol{\epsilon} \Big \rangle.
\end{align}
Without loss of generality, assume that $b^{\ast} \in [r_J, c_{k^{\ast}}^r)$ and a similar approach holds when $b^{\ast} \in [c_{m^{\ast}}^l, r_J)$. 
Using Lemma \ref{Lemma_for_main_thm}, we have that for the left-hand side of \eqref{equivalent_contradiction},
\begin{equation} \label{Lambda_definition}
    \Bigl( \tilde{f}^{r_J}_{c_{m^{\ast}}^l, c_{k^{\ast}}^r} \Bigl)^2 - \Bigl( \tilde{f}^{b^{\ast}}_{c_{m^{\ast}}^l, c_{k^{\ast}}^r} \Bigl)^2 = \frac{\lvert r_J - b^{\ast}\rvert (r_J-c_{m^{\ast}}^l+1)}{\lvert r_J - b^{\ast}\rvert + (r_J-c_{m^{\ast}}^l+1)} \Bigl( \Delta_J^f \Bigl)^2 \colon = \Lambda.
\end{equation}
Also, for the right-hand side
\begin{equation} \label{eq_epsilon}
    \Bigl( \tilde{\epsilon}^{b^{\ast}}_{c_{m^{\ast}}^l, c_{k^{\ast}}^r} \Bigl)^2 - \Bigl( \tilde{\epsilon}^{r_J}_{c_{m^{\ast}}^l, c_{k^{\ast}}^r} \Bigl)^2 \leq \max_{s,e,b:s\leq b<e} \Bigl(\tilde{\epsilon}^b_{s,e}\Bigl)^2 \leq 8\log T
\end{equation}
and using Lemma \ref{Lemma_for_main_thm} and using \eqref{B_T},
\begin{align} \label{eq_extra_stuff}
    2\Big \langle \boldsymbol{\psi_{c_{m^{\ast}}^l, c_{k^{\ast}}^r}^{b^{\ast}}} & \langle \boldsymbol{f}, \boldsymbol{\psi_{c_{m^{\ast}}^l, c_{k^{\ast}}^r}^{b^{\ast}}} \rangle - \boldsymbol{\psi_{c_{m^{\ast}}^l, c_{k^{\ast}}^r}^{r_J} } \langle \boldsymbol{f}, \boldsymbol{\psi_{c_{m^{\ast}}^l, c_{k^{\ast}}^r}^{r_J} } \rangle, \boldsymbol{\epsilon} \Big \rangle \nonumber\\
    &\leq 2 \big\| \boldsymbol{\psi_{c_{m^{\ast}}^l, c_{k^{\ast}}^r}^{b^{\ast}}} \langle \boldsymbol{f}, \boldsymbol{\psi_{c_{m^{\ast}}^l, c_{k^{\ast}}^r}^{b^{\ast}}} \rangle - \boldsymbol{\psi_{c_{m^{\ast}}^l, c_{k^{\ast}}^r}^{r_J}} \langle \boldsymbol{f}, \boldsymbol{\psi_{c_{m^{\ast}}^l, c_{k^{\ast}}^r}^{r_J}} \rangle \big\|_2 \sqrt{8\log T} \nonumber\\
    & = 2\sqrt{\Lambda}\sqrt{8\log T} .
\end{align}
Using \eqref{Lambda_definition}, \eqref{eq_epsilon} and \eqref{eq_extra_stuff}, we can conclude that \eqref{equivalent_contradiction} is satisfied if $\Lambda > 8\log T + \sqrt{2} \sqrt{\Lambda}\sqrt{8\log T}$ is satisfied, which has solution
\begin{equation*}
    \Lambda > \left(2\sqrt{2} +4\right)^2 \log T.
\end{equation*}
From \eqref{Lambda_definition} and since
\begin{align*}
    \frac{\lvert r_J - b^{\ast}\rvert (r_J-c_{m^{\ast}}^l+1)}{\lvert r_J - b^{\ast}\rvert + r_J-c_{m^{\ast}}^l+1} 
    & \geq \frac{\lvert r_J - b^{\ast}\rvert (r_J-c_{m^{\ast}}^l+1)} {2 \max\{\lvert r_J - b^{\ast}\rvert, r_J-c_{m^{\ast}}^l+1\}} \\
    & = \frac{\min\{\lvert r_J - b^{\ast}\rvert, r_J-c_{m^{\ast}}^l+1\}}{2},
\end{align*}
we can conclude that
\begin{equation}\label{min_bound2}
    \min\{\lvert r_J - b^{\ast}\rvert, r_J-c_{m^{\ast}}^l+1\} > \frac{2\left(2\sqrt{2} +4\right)^2 \log T}{\Bigl( \Delta_J^f \Bigl)^2} = C_3 \frac{\log T}{\Bigl( \Delta_J^f \Bigl)^2}
\end{equation}
implies \eqref{contradiction}. Now, if
\begin{equation} \label{min_bound3}
    \min\{r_J-c_{m^{\ast}}^l+1, c_{k^{\ast}}^r - r_J\} > C_3 \frac{\log T}{\Bigl( \Delta_J^f \Bigl)^2},
\end{equation}
then \eqref{min_bound2} is restricted to $\lvert r_J - b^{\ast}\rvert \Bigl( \Delta_J^f \Bigl)^2 > C_3 \log T$ and this implies \eqref{contradiction}. So, we conclude that necessarily
\begin{equation} \label{result_step3}
    \lvert r_J - b_J\rvert \Bigl( \Delta_J^f \Bigl)^2 \leq C_3 \log T.
\end{equation}
But \eqref{min_bound3} must be true since if we assume that $\min\{r_J-c_{m^{\ast}}^l+1, c_{k^{\ast}}^r - r_J\} \leq C_3 \log T/\Bigl( \Delta_J^f \Bigl)^2$, then we have that
\begin{align*}
    \Bigl | \tilde{X}^{b_J}_{c_{m^{\ast}}^l, c_{k^{\ast}}^r} \Bigl | 
    & \leq \Bigl | \tilde{f}^{b_J}_{c_{m^{\ast}}^l, c_{k^{\ast}}^r} \Bigl | + \sqrt{8\log T} 
    \leq \Bigl | \tilde{f}^{r_J}_{c_{m^{\ast}}^l, c_{k^{\ast}}^r} \Bigl | + \sqrt{8\log T} \\
    & = \sqrt{\frac{(c_{k^{\ast}}^r - r_J) (r_J - c_{m^{\ast}}^l + 1)} {c_{k^{\ast}}^r-c_{m^{\ast}}^l+1}} \Delta^f_J + \sqrt{8\log T} \\
    & \leq \sqrt{\min\{c_{k^{\ast}}^r - r_J, r_J - c_{m^{\ast}}^l + 1\}} \Delta^f_J + \sqrt{8\log T} \\
    & \leq \bigl( \sqrt{C_3} + 2\sqrt{2} \bigl) \sqrt{\log T}
    = C_1 \sqrt{\log T}
    \leq \zeta_T,
\end{align*}
which contradicts $\Bigl | \tilde{X}^{b_J}_{c_{m^{\ast}}^l, c_{k^{\ast}}^r} \Bigl | > \zeta_T$.
\vspace{0.1in}
\\
Thus, we have proved that for $\lambda_T \leq \delta_T/2n$, working under the assumption that both $A_T$ and $B_T$ hold, there will be an interval $[c_{m^{\ast}}^l, c_{k^{\ast}}^r]$ with $\Bigl | \tilde{X}^{b_J}_{c_{m^{\ast}}^l, c_{k^{\ast}}^r} \Bigl |  > \zeta_T$, where $b_J = \textrm{argmax}_{c_{m^{\ast}}^l \leq t < c_{k^{\ast}}^r} \Big| \tilde{X}^t_{c_{m^{\ast}}^l, c_{k^{\ast}}^r} \Big|$ is the estimated location for the change-point $r_J$ that satisfies \eqref{result_step3}.
\vspace{0.1in}
\\
{\textbf{Step 4:}} After the detection of the change-point $r_J$ at the estimated location $b_J$ in the interval $[c_{m^{\ast}}^l, c_{k^{\ast}}^r]$, the process is repeated in the disjoint intervals $[1, c_{m^{\ast}}^l]$ and $[c_{k^{\ast}}^r, T]$, which contain $r_1, r_2, \ldots, r_{J-1}$ and $r_{J+1}, r_{J+2}, \ldots, r_N$ respectively. 
This means that we do not check if there are any more change-points in the interval $[c_{m^{\ast}}^l, c_{k^{\ast}}^r]$, besides the already detected $r_J$. 
Since $c_{m^{\ast}}^l \geq c_{\Tilde{m}}^l$ and $c_{k^{\ast}}^r \leq c_{\Tilde{k}}^r$, it holds that $[c_{m^{\ast}}^l, c_{k^{\ast}}^r] \subset [c_{\Tilde{m}}^l, c_{\Tilde{k}}^r]$ and the argument of \eqref{isolation} can be used to conclude that $r_J$ is isolated in $[c_{m^{\ast}}^l, c_{k^{\ast}}^r]$.
\vspace{0.1in}
\\
\textbf{Step 5:} After detecting $r_J$, the algorithm will first check the interval $[1, c_{m^{\ast}}^l]$. So, unless $r_J = r_1$ and $[1, c_{m^{\ast}}^l]$ contains no other change-points, the next change-point to get detected will be one of $r_1, r_2, \ldots, r_{J-1}$. 
The location of the largest difference in the interval $[1, c_{m^{\ast}}^l]$, $d_{1, c_{m^{\ast}}^l}$, will again determine which change-point will be detected next. 
If $d_{1, c_{m^{\ast}}^l}$ is at a position that the next detected change-point, as explained at the beginning of Step 3, will be one of $r_1, r_2, \ldots, r_{J-2}$, then we can prove exactly as we did for $r_J$ that there will eventually be an interval with end-points far enough from the change-point that allows detection while the change-point is isolated. 
\vspace{0.1in}
\\
Now, we need to discuss the case that the next change-point to get detected is $r_{J-1}$.
This is the closest change-point to the already detected $r_J$ so the end-point of the interval where DAIS is reapplied to depends on the interval where the detection of $r_J$ occurred. 
As before, for $k_{J-1} \in \{1,\ldots,K^{\max}\}$ and $m_{J-1} \in \{k_{J-1} - 1, k_{J-1}\}$
we will show that $r_{J-1}$ gets detected in $[c_{m_{J-1}^{\ast}}^l, c_{k_{J-1}^{\ast}}^r]$, where $c_{m_{J-1}^{\ast}}^l \geq c_{m_{J-1}}^l$ and $c_{k_{J-1}^{\ast}}^r \leq c_{k_{J-1}}^r \leq c_{m^{\ast}}^l$ and its detection is at location 
\begin{equation*}
    b_{J-1} = \textrm{argmax}_ {c_{m_{J-1}^{\ast}}^l \leq t < c_{k_{J-1}^{\ast}}^r} \Bigl| \tilde{X}^t _ {c_{m_{J-1}^{\ast}}^l, c_{k_{J-1}^{\ast}}^r} \Bigl|,
\end{equation*} 
which satisfies $\bigl| r_{J-1} - b_{J-1}\bigl| \Bigl( \Delta_{J-1}^f \Bigl)^2 \leq C_3 \log T$. 
Firstly, $r_{J-1}$ is isolated in the interval $[c_{m_{J-1}}^l, c_{k_{J-1}}^r]$ using the same argument as in \eqref{isolation}. 
Following similar steps as in \eqref{unobserved_cusum}, we have that for $\tilde{b}_{J-1} = \textrm{argmax}_ {c_{m_{J-1}}^l \leq t < c_{k_{J-1}}^r} \Bigl| \tilde{X}^t _ {c_{m_{J-1}}^l, c_{k_{J-1}}^r} \Bigl|$,
\begin{align} \label{cusum_step5}
    \Bigl| \tilde{X}^{\tilde{b}_{J-1}} _ {c_{m_{J-1}}^l, c_{k_{J-1}}^r} \Bigl|
    & \geq \Bigl| \tilde{f}^{\tilde{r}_{J-1}} _ {c_{m_{J-1}}^l, c_{k_{J-1}}^r} \Bigl| - \sqrt{8\log T} \nonumber\\
    & \geq \sqrt{\frac{\min\{c_{k_{J-1}}^r-r_{J-1},r_{J-1}-c_{m_{J-1}}^l+1\}}{2}} \Delta^f_{J-1}.
\end{align}
Before we show that $\min\{c_{k_{J-1}}^r-r_{J-1},r_{J-1}-c_{m_{J-1}}^l+1\} \geq \delta_T/2n$, we need show that $c_{m^{\ast}}^l$ satisfies $c_{m^{\ast}}^l - r_{J-1}\geq \delta_T/2n$ since $c_{m^{\ast}}^l \geq c_{k_{J-1}}^r$.
It holds that
\begin{equation*}
    c_{m^{\ast}}^l - r_{J-1} 
    = c_{m^{\ast}}^l - r_J + r_J - r_{J-1} 
    \geq c_{\Tilde{m}}^l - r_J + \Tilde{\delta}_{J-1}.
\end{equation*}
We are again concentrating on the worst-case scenario, meaning that both end-points are at least at a distance $\delta_T/2n$ from the change-point.
More precisely, either $c_{\Tilde{m}}^l \in I^L_J$ or $c_{\Tilde{k}}^r \in I^R_J$, and the other end-point could be further than $\delta_T/n$ from the change-point.
If $c_{\Tilde{m}}^l \in I^L_J$, then, since $n \geq 3/2$, the above inequality becomes 
\begin{equation*}
    c_{m^{\ast}}^l - r_{J-1} 
    \geq c_{\Tilde{m}}^l - r_J + \Tilde{\delta}_{J-1}
    > -\frac{\delta_T}{n} + 1 + \delta_T 
    > \frac{n-1}{n} \delta_T 
    > \frac{\delta_T}{2n}.
\end{equation*}
If instead $c_{\Tilde{k}}^r \in I^R_J$, then using \eqref{delta upper bound} and the fact that $d_{1,T} - c_{\Tilde{m}}^l < c_{\Tilde{k}}^r - d_{1,T}$ due to the way the expansions occur,
\begin{align*}
    r_J - c_{\Tilde{m}}^l 
    & = r_J - d_{1,T} + d_{1,T} - c_{\Tilde{m}}^l
    < \delta_{1,T}^{J-1} + c_{\Tilde{k}}^r - d_{1,T}\\
    & \leq \delta_{1,T}^{J-1} + c_{\Tilde{k}}^r - r_J + r_J - d_{1,T} 
    \leq 2 \delta_{1,T}^{J-1} + \frac{\delta_T}{n}\\
    & \leq 2\Bigl(\frac{\Tilde{\delta}_{J-1}}{2} - \frac{3\delta_T}{4n}\Bigr) + \frac{\delta_T}{n} 
    = \Tilde{\delta}_{J-1} - \frac{\delta_T}{2n}
\end{align*}
and so
\begin{equation*}
    c_{m^{\ast}}^l - r_{J-1} 
    \geq c_{\Tilde{m}}^l - r_J + \Tilde{\delta}_{J-1}
    \geq - \Tilde{\delta}_{J-1} + \frac{\delta_T}{2n} + \Tilde{\delta}_{J-1}
    = \frac{\delta_T}{2n}.
\end{equation*}
This means that the right end-point of the interval could eventually satisfy \linebreak $c_{k_{J-1}}^r - r_{J-1} \geq \delta_T/2n$ for some $k_{J-1}$ since the largest it can become is $c_{k_{J-1}}^r = c_{m^{\ast}}^l$. 
As in Step 3, the detection could happen before this is satisfied.
For the left end-point of the interval, the same holds as in Step 3.
In any case, $\min\{c_{k_{J-1}}^r-r_{J-1},r_{J-1}-c_{m_{J-1}}^l+1\} \geq \delta_T/2n$ holds and so, from \eqref{cusum_step5} we have that
\begin{align*}
    \Bigl| \tilde{X}^{\tilde{b}_{J-1}} _ {c_{m_{J-1}}^l, c_{k_{J-1}}^r} \Bigl|
    & \geq \sqrt{\frac{\delta_T}{4n}}\Delta^f_{J-1} - \sqrt{8\log T}
    \geq \sqrt{\frac{\delta_T}{4n}}\underline{f}_Te - \sqrt{8\log T} \\
    & \geq \biggl( \frac{1}{\sqrt{4n}} - \frac{2\sqrt{2}}{\underline{C}} \bigg) \sqrt{\delta_T} \underline{f}_T
    = C_2 \sqrt{\delta_T} \underline{f}_T > \zeta_T.
\end{align*}
Therefore, we have shown that there exists an interval of the form $[c_{\tilde{m}_{J-1}}^l, c_{\tilde{k}_{J-1}}^r]$ with $\max_{c_{\tilde{m}_{J-1}}^l \leq b < c_{\tilde{k}_{J-1}}^r} \Bigl| \tilde{X}^t _ {c_{\tilde{m}_{J-1}}^l, c_{\tilde{k}_{J-1}}^r} \Bigl| > \zeta_T$. 
\vspace{0.1in}
\\
Now, denote $c_{m_{J-1}^{\ast}}^l, c_{k_{J-1}^{\ast}}^r$ the first points where this occurs and $b_{J-1}$ as defined above with $\Bigl| \tilde{X}^{b_J-1} _ {c_{m_{J-1}^{\ast}}^l, c_{k_{J-1}^{\ast}}^r} \Bigl| > \zeta_T$. 
We can show that $\bigl| r_{J-1} - b_{J-1}\bigl| \Bigl( \Delta_{J-1}^f \Bigl)^2 \leq C_3 \log T$, following exactly the same process as in Step 3.
\vspace{0.1in}
\\
After detecting $r_{J-1}$ in the interval $[c_{m_{J-1}^{\ast}}^l, c_{k_{J-1}^{\ast}}^r]$, DAIS will restart on intervals $[1, c_{m_{J-1}^{\ast}}^l]$ and $[c_{k_{J-1}^{\ast}}^r, c^l_{m^\ast}]$. 
DAIS will also check $[c_{k^{\ast}}^r, T]$, as described in Section \ref{methodology}.
Step 5 can be applied to all intervals, as long as there is a change-point. We can conclude that all change-points will get detected, one by one, and their estimated locations will satisfy $\bigl| r_j - b_j\bigl| \Bigl( \Delta_j^f \Bigl)^2 \leq C_3 \log T$, $\forall j \in \{1, 2, \ldots, N\}$. 
There will not be any double detection issues as each interval contains no previously detected change-points.
\vspace{0.1in}
\\
\textbf{Step 6:} After detecting all the change-points at locations $b_1, b_2, \ldots, b_N$ using the intervals $[c_{m_j^{\ast}}^l, c_{k_j^{\ast}}^r]$ for $j\in \{1, \ldots, N\}$, the algorithm will check all intervals of the form $[c_{k_{j-1}^{\ast}}^r, c_{m_j^{\ast}}^l]$ and $[c_{k_j^{\ast}}^r, c_{m_{j+1}^{\ast}}^l]$, with $c_{m_0^{\ast}}^l = 1$ and $c_{k_{N+1}^{\ast}}^r = T$. 
At most $N+1$ intervals of this form, containing no change-points, will be checked. Denoting by $[s^{\ast}, e^{\ast}]$ any of those intervals, we can show that DAIS will not detect any change-point as for $b \in \{s^{\ast}, s^{\ast}+1, \ldots, e^{\ast} - 1\}$,
\begin{equation*}
    \Bigl| \tilde{X}^{b}_{s^{\ast}, e^{\ast}} \Bigl| 
    \leq \Bigl| \tilde{f}^{b}_{s^{\ast}, e^{\ast}} \Bigl|  + \sqrt{8 \log T}
    = \sqrt{8 \log T}
    < C_1 \sqrt{\log T}
    \leq \zeta_T.
\end{equation*}
The algorithm will terminate after not detecting any change-points in all intervals.
\\
\end{proof}

\section{Outline of the proof of Theorem \ref{thm: consistency_multivariate}} \label{appendix: proof_multivariate}

Before providing the outline of the proof, we need some notation. For $1 \leq s \leq b < e \leq T$, we denote by $\boldsymbol{\tilde{X}_{s,e}^{b}} := (\tilde{X}_{s,e}^{b,1}, \tilde{X}_{s,e}^{b,2}, \ldots, \tilde{X}_{s,e}^{b,d})$, where $\tilde{X}_{s,e}^{b,j} = \sqrt{\frac{e-b}{\ell (b-s+1)}}\sum_{t=s}^{b} X_{t,j} - \sqrt{\frac{b-s+1}{\ell (e-b)}}\sum_{t=b+1}^{e} X_{t,j}$, which is equivalent to \eqref{CUSUM}, $\ell = e-s+1$ and $X_{t,j}$ is the $j\textsuperscript{th}$ component of $\boldsymbol{X_t}$, for $j\in\{1, \ldots, d\}$. Furthermore, $\boldsymbol{\tilde{f}_{s,e}^{b}} := (\tilde{f}_{s,e}^{b,1}, \tilde{f}_{s,e}^{b,2}, \ldots, \tilde{f}_{s,e}^{b,d})$, where $\tilde{f}_{s,e}^{b,j} = \sqrt{\frac{e-b}{n(b-s+1)}}\sum_{t=s}^{b}f_{t,j} - \sqrt{\frac{b-s+1}{n(e-b)}}\sum_{t=b+1}^{e}f_{t,j}$ is the value of the CUSUM statistic at the location $b$, when we work in the interval from $s$ up to $e$, for the $j^{th}$ underlying component signal. Also, for $L(\cdot)$ any mean-dominant norm as those in \eqref{mean_dominant}, and for $|\boldsymbol{x}|$ being the vector of the absolute value of the elements of an $\boldsymbol{x} \in \mathbb{R}^{d}$, $C^{b}_{s,e}  = L\left(\left|\boldsymbol{\tilde{X}_{s,e}^{b}}\right|\right)$ and $D_{s,e}^{b} = L\left(\left|\boldsymbol{\tilde{f}_{s,e}^{b}}\right|\right).$
Proving Theorem~\ref{thm: consistency_multivariate} involves the following steps: \\

\noindent \textbf{Step 1:} We show that the values of the mean-dominant norm of the contrast function of the observed and unobserved signals, $L \left( \boldsymbol{\tilde{X}_{s,e}^{b}} \right) $ and $L \left( \boldsymbol{\tilde{f}_{s,e}^{b}} \right)$, are uniformly close, for all $1\leq s \leq b < e \leq T$.
Using the mean-dominance property as given in p.190 of \cite{Carlstein1988}, it holds that
$0 \leq L_1(\boldsymbol{x}) \leq L(\boldsymbol{x}) \leq L_{\infty}(\boldsymbol{x}), \forall \boldsymbol{x}\in (\mathbb{R}^d)^{+}$
and so, for $A_T^*  =\bigg\{\max_{1\leq s \leq b < e\leq T} L\left(\left|\boldsymbol{\tilde{X}_{s,e}^{b}} - \boldsymbol{\tilde{f}_{s,e}^{b}}\right|\right) \leq \sqrt{8\operatorname{log}\left(Td^{\frac{1}{4}}\right)}\bigg\}$ and $A_T  =\bigg\{\max_{1 \leq s \leq b < e \leq T}\left|C^{b}_{s,e}-D_{s,e}^{b}\right|\leq \sqrt{8\operatorname{log}\left(Td^{\frac{1}{4}}\right)}\bigg\}$, it holds that $A^\ast_T \subseteq A_T$.
From the definition of our model in \eqref{eq: high_dim_model}, we have that $\boldsymbol{\epsilon_t} \sim \mathcal{N}_d(\boldsymbol{0}, \Sigma), \forall t \in \left\lbrace 1,\ldots, T\right\rbrace$ and, using the Bonferroni inequality, we can show, following similar steps as those in \cite{anastasiou_generalized_2023}, that $\Prob \left( \left( A^\ast_T \right) \right) \leq 1/(12 \sqrt{\pi} T)$, which implies $\Prob \left( A_T \right) \geq \Prob \left( A^\ast_T \right) \geq 1/(12 \sqrt{\pi} T)$. \\

\noindent \textbf{Step 2:} 
We control the distance between $L \left( \boldsymbol{\tilde{X}_{s,e}^{b_1}} \right) - L \left( \boldsymbol{\tilde{X}_{s,e}^{b_2}} \right)$ and $L \left( \boldsymbol{\tilde{f}_{s,e}^{b_1}} \right) - L \left( \boldsymbol{\tilde{f}_{s,e}^{b_2}} \right)$, for all possible combinations of $s,e,b_1,b_2$, where $1\leq s < e \leq T$.
For $\boldsymbol{f_k} = \left(f_{1,k}, \ldots, f_{T,k}\right)$ and for $[s,e]$ being any interval that contains only one true change-point, namely $r_j$, we denote
\begin{equation}
\label{A_seb}
    A_{s,e}^{b}(k,r_j) := \frac{\left\langle\boldsymbol{\psi_{s, e}^{b}}\left\langle \boldsymbol{f_{k}}, \boldsymbol{\psi_{s, e}^{b}}\right\rangle-\boldsymbol{\psi_{s,e}^{r_{j}}}\left\langle\boldsymbol{f_{k}},\boldsymbol{\psi_{s,e}^{r_{j}}}\right\rangle,\boldsymbol{\epsilon}\right\rangle}{\left\|\boldsymbol{\psi_{s, e}^{b}}\left\langle \boldsymbol{f_{k}}, \boldsymbol{\psi_{s, e}^{b}}\right\rangle-\boldsymbol{\psi_{s,e}^{r_{j}}}\left\langle\boldsymbol{f_{k},\boldsymbol{\psi_{s,e}^{r_{j}}}}\right\rangle\right\|_{2}},
\end{equation}
where $\psi_{s, e}^{b}$ is defined in \eqref{psi_definition}. 
In \eqref{A_seb} above, $k \in \left\lbrace 1, \ldots, d\right\rbrace$, while $\|\cdot\|_2$ is the Euclidean norm. Due to $\epsilon_{t,k}, t = 1,\ldots, T$ following the standard normal distribution for all $k \in \left\lbrace 1,\ldots, d\right\rbrace$ (with the univariate error terms being independent over time, but possibly spatially dependent), then straightforward calculations lead to $A_{s,e}^{b}(k,r_j) \sim \mathcal{N}(0,1), \forall k \in \left\lbrace 1,\ldots,d\right\rbrace$.
For \begin{equation}
\label{set_B_T}
B_{T}=\left\{\max_{\substack{j=1, \ldots, N\\k=1,\ldots,d}}\max_{\substack{
r_{j-1}<s \leq r_{j}\\
r_{j} < e \leq r_{j+1} \\
s \leq b <e}}\left|A_{s,e}^{b}(k,r_j)\right| \leq \sqrt{8\operatorname{log}\left(Td^{\frac{1}{4}}\right)}\right\},
\end{equation} 
we can show that $\Prob\left(B_T\right) \geq 1 - 1/(12\sqrt{\pi}T)$ using a similar procedure as in Step 1.
From Steps 1 and 2, we conclude that $\Prob\left(A_T^* \cap B_T\right) \geq 1 - \frac{1}{6\sqrt{\pi}T}.$\\

\noindent \textbf{Step 3:} 
Working under the assumption that both $A^\ast_T$ (and thus also $A_T$) and $B_T$ hold, we show that detection in an interval where the change-point is isolated will occur and so the proof can be restricted to an interval with a single change-point and we also prove that the estimated and true change-points are close to each other. The values used in this step for the positive constants $C_1, C_2,$ and $C_3$ of Theorem~\ref{thm: consistency_multivariate}, are $C_1 = \sqrt{C_3} + \sqrt{8}, C_2 = \frac{1}{\sqrt{4n}} - \frac{2\sqrt{2}}{\underline{C}_M}, C_3 = 2(2\sqrt{2}+4)^2$.
The order with which the change-points are detected depends on the location of the largest difference at each step of the algorithm and so, without loss of generality, we assume that the first change-point to get detected is $r_J$ for $J\in\{1,\ldots,N\}$.
For $c_m^l, c_k^r$ as defined in \eqref{end-points_proof}, $c_m^l \in I_J^L$ or $c_k^r \in I_J^R$ hold (or both), with $I_J^L, I_J^R$ as defined in \eqref{intervals_discussion}, and showing that the change-point will be isolated in $[c_m^l, c_k^r]$ is the same as for Theorem\ref{consistency_theorem} Step 3.
It can be shown that for $\tilde{b} = {\rm argmax}_{c_m^l \leq t < c_k^r}C_{c_m^l, c_k^r}^{t}$, $C_{c_m^l,c_k^r}^{\tilde{b}} > \zeta_{T,d}$, so $r_J$ can be detected in the interval.

Now, for $k^\ast \in \{1,\ldots,K^{\max}\}$ and $m^\ast \in \{k^\ast-1, k^\ast\}$, we denote by $c_{m^\ast}^l \geq c_m^l$ and $c_{k^\ast}^r \leq c_{k}^r$ the first left- and right-expanding points, respectively, where detection of $r_J$ occurs.
Showing that the estimated and true change-points are close to each other requires working on the two mean-dominant norms of \eqref{mean_dominant} separately.
Denoting by $q_J:= {\rm argmax}_{j=1,\ldots,d} \left| \tilde{X}_{c_{m^\ast}^l,c_{k^\ast}^r}^{b_J,j}\right|$ and $\alpha_J = \begin{cases}
    \lvert f_{r_{J+1}, q_J} - f_{r_J, q_J} \rvert, \;{\rm when}\; L(\cdot) = L_\infty(\cdot) \\
    L_2(\boldsymbol{\Delta_J}), \;{\rm when}\; L(\cdot) = L_2(\cdot)
\end{cases}$, we show that $\left|b_J - r_J \right|\alpha_J^2 \leq C_3 \log\left(Td^{\frac{1}{4}}\right).$
This can be done using a contradiction argument, based on the fact that for any $b^* \in \left\lbrace c_{m^\ast}^l,\ldots,c_{k^\ast}^r - 1 \right\rbrace$ with $\left|b^*-r_J \right| \alpha_J^2 > \log \left( Td^{\frac{1}{4}} \right)$, it holds that
$\left( \tilde{X}_{c_{m^\ast}^l, c_{k^\ast}^r} ^{r_J,q_J}\right)^2 > \left(\tilde{X}_{c_{m^\ast}^l, c_{k^\ast}^r}^{b^*,q_J}\right)^2$, which, using the definition of $b_J$, cannot be true. 
The detailed proof of this is skipped, but we highlight that it basically relies on a combination of the strategy followed in \cite{anastasiou_generalized_2023} that the consistency of the MID algorithm was proved, and the approach followed in our Theorem \ref{consistency_theorem}, Step 3.\\

\noindent \textbf{Step 4:} We show that when the algorithm re-starts, no change-points are left undetected and that the new intervals allow for the detection of the remaining change-points. 
The first part is immediate, as from Step 3 we know that the change-point was isolated in the interval in which detection occurred.
Since the change-point was detected in the interval $[c_{m^\ast}^l, c_{k^\ast}^r]$, MDAIS restarts on $[1, c_{m^\ast}^l]$ and $[c_{k^\ast}^r, T]$.
The proof that detection of $r_j$ occurs is the same as in Step 3 for all $j \in \{1, \ldots, J-2, J+2, \ldots, N \}$.
We focus on $j = J-1$ and similar steps hold for $j = J+1$.
For $k_{J-1} \in \{1,\ldots,K^{\max}\}$ and $m_{J-1} \in \{k_{J-1} - 1, k_{J-1}\}$, we will show that $r_{J-1}$ gets detected in $[c_{m_{J-1}^{\ast}}^l, c_{k_{J-1}^{\ast}}^r]$, where $c_{m_{J-1}^{\ast}}^l \geq c_{m_{J-1}}^l$ and $c_{k_{J-1}^{\ast}}^r \leq c_{k_{J-1}}^r \leq c_{m^{\ast}}^l$.
The upper bound of the last inequality is the reason why $r_{J-1}$ should be considered separately, as the right end-point of the interval where detection will occur is upper bounded by the left end-point of the interval in which $r_J$ was detected.
The detection of $r_{J-1}$ is at $b_{J-1} = {\rm argmax}_{c_{m^\ast_{J-1}}^l \leq t < c_{k^\ast_{J-1}}^r}C_{c_{m^\ast_{J-1}}^l, c_{k^\ast_{J-1}}^r}^t$.
Following similar steps as in Step 3, we have that $C_{c_{m^\ast_{J-1}}^l, c_{k^\ast_{J-1}}^r}^{b_{J-1}} >\zeta_{T,d}$.
Now, with $$\alpha_{J-1} := \begin{cases} 
\lvert f_{r_{J}, q_{J-1}} - f_{r_{J-1}, q_{J-1}} \rvert, \;{\rm when}\; L(\cdot) = L_{\infty}(\cdot)\\L_2(\boldsymbol{\Delta_{J-1}}), \;{\rm when}\; L(\cdot) = L_{2}(\cdot)\end{cases},$$ for $q_{J-1}:= {\rm argmax}_{j=1,2,\ldots,d}\left|\tilde{X}_{c_{m^\ast_{J-1}}^l, c_{k^\ast_{J-1}}^r}^{b_{J-1},j}\right|$, it is straightforward to show, following the same process as Step 3 of the proof, that $\left|b_{J-1} - r_{J-1}\right|\alpha_{J-1}^2 \leq C_3\log \left(Td^{1/4}\right)$.\\

\noindent \textbf{Step 5:} Finally, when no change-points are left, the algorithm terminates.
After detecting all the change-points at locations $b_1, b_2, \ldots, b_N$ using the intervals $[c_{m_j^{\ast}}^l, c_{k_j^{\ast}}^r]$ for $j\in \{1, \ldots, N\}$, the algorithm will check all intervals of the form $[c_{k_{j-1}^{\ast}}^r, c_{m_j^{\ast}}^l]$ and $[c_{k_j^{\ast}}^r, c_{m_{j+1}^{\ast}}^l]$, with $c_{m_0^{\ast}}^l = 1$ and $c_{k_{N+1}^{\ast}}^r = T$. 
At most $N+1$ intervals of this form, containing no change-points will be checked. 
Denoting by $[s^{\ast}, e^{\ast}]$ any of those intervals, for $b \in \{s^{\ast}, s^{\ast}+1, \ldots, e^{\ast} - 1\}$,
\begin{equation*}
    C_{s^*,e^*}^{b} \leq D_{s^*,e^*}^{b} + \sqrt{8\log \left(Td^{\frac{1}{4}}\right)}
    = \sqrt{8\log \left(Td^{\frac{1}{4}}\right)} 
    < C_1\sqrt{\log \left(Td^{\frac{1}{4}}\right)}\leq \zeta_{T,d}
\end{equation*}
and the algorithm terminates.
\end{appendix}
\\

\slc{\textbf{Acknowledgements} We thank the Editor, the Associate Editor, and the three anonymous referees for their helpful comments which have led to a significant improvement of the paper. 
We also thank Piotr Fryzlewicz for constructive conversations. 
This research started whilst S.L. was employed as a Research Special Scientist at the University of Cyprus, supported by the start-up grant that A.A. received from the University of Cyprus. 
Afterwards, S.L. was funded in whole, or in part, by the Luxembourg National Research Fund (FNR), grant reference PRIDE/21/16747448/MATHCODA.}

\textbf{Author contribution} 
A.A. has been regularly reviewing the manuscript and has guided S.L. for the simulations and real data analysis. 
S.L. came up with the idea of the proposed methodology, worked on the real data analysis, performed the simulations, drafted the paper and created the GitHub repository. 
Both authors carried out the proofs of the theoretical results and have carefully reviewed the latest version of the manuscript.

\slc{\textbf{Funding} Open access funding provided by the Cyprus Libraries
Consortium (CLC).
This research started whilst S.L. was employed as a Research Special Scientist at the University of Cyprus, supported by the start-up grant that A.A. received from the University of Cyprus. 
Afterwards, S.L. was funded in whole, or in part, by the Luxembourg National Research Fund (FNR), grant reference PRIDE/21/16747448/MATHCODA.}

\textbf{Data availability} In Section \ref{Real data}, we provide the relevant links where all data
and materials used in the paper can be obtained from.

\section*{Declarations}

\textbf{Competing interests} The authors have no competing interests to declare that are relevant to the content of this article.

\textbf{Ethics approval and consent to participate} Not applicable.

\textbf{Consent for publication} Not applicable.

\textbf{Materials availability} Not applicable.

\slc{\textbf{Code availability} Code and instructions to replicate the results in the paper are available in \url{https://doi.org/10.5281/zenodo.15469890} as well as in \url{https://github.com/Sophia-Loizidou/DAIS}. }

\slc{\textbf{License} This article is licensed under a Creative Commons
Attribution 4.0 International License.}

\bibliographystyle{abbrv}
\bibliography{references}  

\begin{thebibliography}{10}

\bibitem{breakfast_package}
A.~Anastasiou, Y.~Chen, H.~Cho, and P.~Fryzlewicz.
\newblock {\bf{breakfast}}: {M}ethods for fast multiple change-point detection
  and estimation. {{\textsf{R}}} {P}ackage, {V}ersion 2.5.
\newblock pages 1--27, 2024.

\bibitem{ccid}
A.~Anastasiou, I.~Cribben, and P.~Fryzlewicz.
\newblock {Cross-covariance isolate detect: A new change-point method for
  estimating dynamic functional connectivity}.
\newblock {\em Medical Image Analysis}, 75:102252, 2022.

\bibitem{anastasiou2022detecting}
A.~Anastasiou and P.~Fryzlewicz.
\newblock {Detecting multiple generalized change-points by isolating single
  ones}.
\newblock {\em Metrika}, 85(2):141--174, 2022.

\bibitem{anastasiou_generalized_2023}
A.~Anastasiou and A.~Papanastasiou.
\newblock Generalized multiple change-point detection in the structure of
  multivariate, possibly high-dimensional, data sequences.
\newblock {\em Statistics and Computing}, 33(5):94, 2023.

\bibitem{Anscombe1948}
F.~J. Anscombe.
\newblock {T}he {T}ransformation of {P}oisson, {B}inomial and
  {N}egative-{B}inomial {D}ata.
\newblock {\em Biometrika}, 35(3/4):246--254, 1948.

\bibitem{10.2307/2998540}
J.~Bai and P.~Perron.
\newblock {E}stimating and {T}esting {L}inear {M}odels with {M}ultiple
  {S}tructural {C}hanges.
\newblock {\em Econometrica}, 66(1):47--78, 1998.

\bibitem{baranowski2019narrowest}
R.~Baranowski, Y.~Chen, and P.~Fryzlewicz.
\newblock {N}arrowest-{O}ver-{T}hreshold {D}etection of {M}ultiple {C}hange
  {P}oints and {C}hange-{P}oint-{L}ike {F}eatures.
\newblock {\em Journal of the Royal Statistical Society: Series B (Statistical
  Methodology)}, 81(3):649--672, 2019.

\bibitem{wbs_package}
R.~Baranowski and P.~Fryzlewicz.
\newblock Wild {B}inary {S}egmentation for {M}ultiple {C}hange-{P}oint
  {D}etection. {{\textsf{R}}} {P}ackage, {V}ersion 1.4.1.
\newblock pages 1--15, 2024.

\bibitem{brodsky2000non}
B.~Brodsky and B.~Darkhovsky.
\newblock {\em {N}on-{P}arametric {S}tatistical {D}iagnosis {P}roblems and
  {M}ethods}.
\newblock Springer Netherlands, Dordrecht, 2000.

\bibitem{10.1093/biomet/asv031}
H.~Cao and W.~Biao~Wu.
\newblock {Changepoint estimation: another look at multiple testing problems}.
\newblock {\em Biometrika}, 102(4):974--980, 2015.

\bibitem{Carlstein1988}
E.~Carlstein.
\newblock {Nonparametric Change-Point Estimation}.
\newblock {\em The Annals of Statistics}, 16(1):188 -- 197, 1988.

\bibitem{chan2022}
H.-P. Chan and H.~Chen.
\newblock Multi-sequence segmentation via score and higher-criticism tests,
  2017.

\bibitem{10.2307/24310529}
H.~P. Chan and G.~Walther.
\newblock Detection with the scan and the average likelihood ratio.
\newblock {\em Statistica Sinica}, 23(1):409--428, 2013.
\newblock \url{http://www.jstor.org/stable/24310529}.

\bibitem{cho_change-point_2016}
H.~Cho.
\newblock Change-point detection in panel data via double {CUSUM} statistic.
\newblock {\em Electronic Journal of Statistics}, 10(2):2000--2038, 2016.

\bibitem{10.2307/24310145}
H.~Cho and P.~Fryzlewicz.
\newblock Multiscale and multilevel technique for consistent segmentation of
  nonstationary time series.
\newblock {\em Statistica Sinica}, 22(1):207--229, 2012.
\newblock \url{http://www.jstor.org/stable/24310145}.

\bibitem{cho2015multiple}
H.~Cho and P.~Fryzlewicz.
\newblock Multiple-{C}hange-{P}oint {D}etection for {H}igh {D}imensional {T}ime
  {S}eries via {S}parsified {B}inary {S}egmentation.
\newblock {\em Journal of the Royal Statistical Society: Series B (Statistical
  Methodology)}, 77(2):475--507, 2015.

\bibitem{cho_multiple-change-point_2015}
H.~Cho and P.~Fryzlewicz.
\newblock Multiple-{C}hange-{P}oint {D}etection for {H}igh {D}imensional {T}ime
  {S}eries via {S}parsified {B}inary {S}egmentation.
\newblock {\em Journal of the Royal Statistical Society Series B: Statistical
  Methodology}, 77(2):475--507, 2015.

\bibitem{CHO202476}
H.~Cho and C.~Kirch.
\newblock Data segmentation algorithms: Univariate mean change and beyond.
\newblock {\em Econometrics and Statistics}, 30:76--95, 2024.

\bibitem{chu1995mosum}
C.-S.~J. Chu, K.~Hornik, and C.-M. Kaun.
\newblock {MOSUM} tests for parameter constancy.
\newblock {\em Biometrika}, 82(3):603--617, 1995.

\bibitem{7938741}
J.~Ding, Y.~Xiang, L.~Shen, and V.~Tarokh.
\newblock {M}ultiple {C}hange {P}oint {A}nalysis: {F}ast {I}mplementation and
  {S}trong {C}onsistency.
\newblock {\em IEEE Transactions on Signal Processing}, 65(17):4495--4510,
  2017.

\bibitem{10.3150/16-BEJ887}
B.~Eichinger and C.~Kirch.
\newblock {A MOSUM procedure for the estimation of multiple random change
  points}.
\newblock {\em Bernoulli}, 24(1):526 -- 564, 2018.

\bibitem{cpop_package}
P.~Fearnhead and D.~Grose.
\newblock cpop: {D}etecting {C}hanges in {P}iecewise-{L}inear {S}ignals.
\newblock {\em Journal of Statistical Software}, 109(7):1–--30, 2024.

\bibitem{doi:10.1080/10618600.2018.1512868}
P.~Fearnhead, R.~Maidstone, and A.~Letchford.
\newblock {D}etecting {C}hanges in {S}lope {W}ith an ${L}_0$ {P}enalty.
\newblock {\em Journal of Computational and Graphical Statistics},
  28(2):265--275, 2019.

\bibitem{frick_multiscale_2014}
K.~Frick, A.~Munk, and H.~Sieling.
\newblock Multiscale {Change} {Point} {Inference}.
\newblock {\em Journal of the Royal Statistical Society Series B: Statistical
  Methodology}, 76(3):495--580, May 2014.

\bibitem{10.1214/aos/1176347963}
J.~H. Friedman.
\newblock {Multivariate Adaptive Regression Splines}.
\newblock {\em The Annals of Statistics}, 19(1):1 -- 67, 1991.

\bibitem{fryzlewicz2014wild}
P.~Fryzlewicz.
\newblock Wild binary segmentation for multiple change-point detection.
\newblock {\em The Annals of Statistics}, 42(6):2243--2281, 2014.
\newblock \url{https://www.jstor.org/stable/43556493}.

\bibitem{10.1214/17-AOS1662}
P.~Fryzlewicz.
\newblock {Tail-greedy bottom-up data decompositions and fast multiple
  change-point detection}.
\newblock {\em The Annals of Statistics}, 46(6B):3390 -- 3421, 2018.

\bibitem{fryzlewicz2020detecting}
P.~Fryzlewicz.
\newblock Detecting possibly frequent change-points: {W}ild {B}inary
  {S}egmentation 2 and steepest-drop model selection.
\newblock {\em Journal of the Korean Statistical Society}, 49(4):1027--1070,
  2020.

\bibitem{Akilagun_Fast_Optimal}
S.~Gavioli-Akilagun and P.~Fryzlewicz.
\newblock {F}ast and {O}ptimal {I}nference for {C}hange {P}oints in {P}iecewise
  {P}olynomials via {D}ifferencing.
\newblock {\em Electronic Journal of Statistics}, 19(1):593 -- 655, 2025.

\bibitem{Hampel}
F.~R. Hampel.
\newblock {T}he {I}nfluence {C}urve and its {R}ole in {R}obust {E}stimation.
\newblock {\em Journal of the American Statistical Association},
  69(346):383--393, 1974.

\bibitem{HAWKINS2001323}
D.~M. Hawkins.
\newblock Fitting multiple change-point models to data.
\newblock {\em Computational Statistics \& Data Analysis}, 37(3):323--341,
  2001.

\bibitem{1381461}
B.~Jackson, J.~Scargle, D.~Barnes, S.~Arabhi, A.~Alt, P.~Gioumousis, E.~Gwin,
  P.~Sangtrakulcharoen, L.~Tan, and T.~T. Tsai.
\newblock An algorithm for optimal partitioning of data on an interval.
\newblock {\em IEEE Signal Processing Letters}, 12(2):105--108, 2005.

\bibitem{mosum_linear}
H.-S.~O. Joonpyo~Kim and H.~Cho.
\newblock {M}oving {S}um {P}rocedure for {C}hange {P}oint {D}etection under
  {P}iecewise {L}inearity.
\newblock {\em Technometrics}, 66(3):358--367, 2024.

\bibitem{pelt_package}
R.~Killick and I.~A. Eckley.
\newblock changepoint: An {{\textsf{R}}} {P}ackage for {C}hangepoint
  {A}nalysis.
\newblock {\em Journal of Statistical Software}, 58(3):1–19, 2014.

\bibitem{doi:10.1080/01621459.2012.737745}
R.~Killick, P.~Fearnhead, and I.~A. Eckley.
\newblock {O}ptimal {D}etection of {C}hangepoints {W}ith a {L}inear
  {C}omputational {C}ost.
\newblock {\em Journal of the American Statistical Association},
  107(500):1590--1598, 2012.

\bibitem{doi:10.1137/070690274}
S.-J. Kim, K.~Koh, S.~Boyd, and D.~Gorinevsky.
\newblock $\ell_1$ trend filtering.
\newblock {\em SIAM Review}, 51(2):339--360, 2009.

\bibitem{seeded_bs}
S.~Kovács, P.~Bühlmann, H.~Li, and A.~Munk.
\newblock Seeded {B}inary {S}egmentation: {A} general methodology for fast and
  optimal changepoint detection.
\newblock {\em Biometrika}, 110(1):249--256, 2023.

\bibitem{mscp_package}
T.~Levajkovic and M.~Messer.
\newblock {\bf{mscp}}: {M}ultiscale {C}hange {P}oint {D}etection via {G}radual
  {B}andwidth {A}djustment in {M}oving {S}um {P}rocesses. {{\textsf{R}}}
  {P}ackage, {V}ersion 1.0.
\newblock pages 1--6, 2022.

\bibitem{10.1214/22-EJS2101}
T.~Levajković and M.~Messer.
\newblock {Multiscale change point detection via gradual bandwidth adjustment
  in moving sum processes}.
\newblock {\em Electronic Journal of Statistics}, 17(1):70 -- 101, 2023.

\bibitem{10.1214/16-EJS1131}
H.~Li, A.~Munk, and H.~Sieling.
\newblock {FDR-{C}ontrol in {M}ultiscale {C}hange-point {S}egmentation}.
\newblock {\em Electronic Journal of Statistics}, 10(1):918 -- 959, 2016.

\bibitem{Liehrmann22112024}
A.~Liehrmann and G.~Rigaill.
\newblock {Ms.FPOP}: {A} {F}ast {E}xact {S}egmentation {A}lgorithm with a
  {M}ultiscale {P}enalty.
\newblock {\em Journal of Computational and Graphical Statistics}, 0(0):1--11,
  2024.

\bibitem{maeng2023detecting}
H.~Maeng and P.~Fryzlewicz.
\newblock Detecting linear trend changes in data sequences.
\newblock {\em Statistical Papers}, pages 1--31, 2023.

\bibitem{trendSegmentR_package}
H.~Maeng and P.~Fryzlewicz.
\newblock {\bf{trendsegmentR}}: {L}inear {T}rend {S}egmentation. {{\textsf{R}}}
  {P}ackage, {V}ersion 1.3.0.
\newblock pages 1--7, 2023.

\bibitem{maidstone2017optimal}
R.~Maidstone, T.~Hocking, G.~Rigaill, and P.~Fearnhead.
\newblock On optimal multiple changepoint algorithms for large data.
\newblock {\em Statistics and computing}, 27(2):519--533, 2017.

\bibitem{JSSv097i08}
A.~Meier, C.~Kirch, and H.~Cho.
\newblock mosum: {A} {P}ackage for {M}oving {S}ums in {C}hange-{P}oint
  {A}nalysis.
\newblock {\em Journal of Statistical Software}, 97(8):1–42, 2021.

\bibitem{earth_package}
S.~Milborrow, T.~Hastie, R.~Tibshirani, A.~Miller, and T.~Lumley.
\newblock Multivariate {A}daptive {R}egression {S}plines. {{\textsf{R}}}
  {P}ackage, {V}ersion 5.3.4.
\newblock pages 1--53, 2024.

\bibitem{ng1969table}
E.~W. Ng and M.~Geller.
\newblock A {T}able of {I}ntegrals of the {E}rror {F}unctions.
\newblock {\em Journal of Research of the National Bureau of Standards B},
  73(1):1--20, 1969.

\bibitem{NINOMIYA2005237}
Y.~Ninomiya.
\newblock Information criterion for {G}aussian change-point model.
\newblock {\em Statistics \& Probability Letters}, 72(3):237--247, 2005.

\bibitem{stepr_package}
F.~Pein, T.~Hotz, H.~Sieling, and T.~Aspelmeier.
\newblock {\bf{step{R}}}: {M}ultiscale {C}hange-{P}oint {I}nference.
  {{\textsf{R}}} {P}ackage, {V}ersion 2.1-10.
\newblock pages 1--69, 2024.

\bibitem{se-12-2717-2021}
N.~Piana~Agostinetti and G.~Sgattoni.
\newblock Changepoint detection in seismic double-difference data: application
  of a trans-dimensional algorithm to data-space exploration.
\newblock {\em Solid Earth}, 12(12):2717--2733, 2021.

\bibitem{Raimondo1998}
M.~Raimondo.
\newblock {Minimax estimation of sharp change points}.
\newblock {\em The Annals of Statistics}, 26(4):1379 -- 1397, 1998.

\bibitem{rigaill2015pruned}
G.~Rigaill.
\newblock A pruned dynamic programming algorithm to recover the best
  segmentations with $1 $ to ${K}_{\max} $ change-points.
\newblock {\em Journal de la Soci{\'e}t{\'e} Fran{\c{c}}aise de Statistique},
  156(4):180--205, 2015.
\newblock \url{http://www.numdam.org/item/JSFS_2015__156_4_180_0/}.

\bibitem{doi:10.1080/01621459.1993.10476408}
P.~J. Rousseeuw and C.~Croux.
\newblock {A}lternatives to the {M}edian {A}bsolute {D}eviation.
\newblock {\em Journal of the American Statistical Association},
  88(424):1273--1283, 1993.

\bibitem{10.2307/2958889}
G.~Schwarz.
\newblock {E}stimating the {D}imension of a {M}odel.
\newblock {\em The Annals of Statistics}, 6(2):461--464, 1978.
\newblock \url{http://www.jstor.org/stable/2958889}.

\bibitem{SOSACOSTA20182044}
A.~Sosa-Costa, I.~K. Piechocka, L.~Gardini, F.~S. Pavone, M.~Capitanio, M.~F.
  Garcia-Parajo, and C.~Manzo.
\newblock {PLANT}: A {M}ethod for {D}etecting {C}hanges of {S}lope in {N}oisy
  {T}rajectories.
\newblock {\em Biophysical Journal}, 114(9):2044--2051, 2018.

\bibitem{doi:10.1080/00949655.2011.647317}
S.~Spiriti, R.~Eubank, P.~W. Smith, and D.~Young.
\newblock Knot selection for least-squares and penalized splines.
\newblock {\em Journal of Statistical Computation and Simulation},
  83(6):1020--1036, 2013.

\bibitem{TRUONG2020107299}
C.~Truong, L.~Oudre, and N.~Vayatis.
\newblock Selective review of offline change point detection methods.
\newblock {\em Signal Processing}, 167:107299, 2020.

\bibitem{Verzelen_optimal}
N.~Verzelen, M.~Fromont, M.~Lerasle, and P.~Reynaud-Bouret.
\newblock {Optimal change-point detection and localization}.
\newblock {\em The Annals of Statistics}, 51(4):1586 -- 1610, 2023.

\bibitem{vostrikova1981detecting}
L.~Y. Vostrikova.
\newblock Detecting “disorder” in multidimensional random processes.
\newblock In {\em Doklady akademii nauk}, volume 259, pages 270--274. Russian
  Academy of Sciences, 1981.

\bibitem{10.1214/20-EJS1710}
D.~Wang, Y.~Yu, and A.~Rinaldo.
\newblock {Univariate mean change point detection: Penalization, CUSUM and
  optimality}.
\newblock {\em Electronic Journal of Statistics}, 14(1):1917 -- 1961, 2020.

\bibitem{WIGGINS2015346}
P.~Wiggins.
\newblock {A}n {I}nformation-{B}ased {A}pproach to {C}hange-{P}oint {A}nalysis
  with {A}pplications to {B}iophysics and {C}ell {B}iology.
\newblock {\em Biophysical Journal}, 109(2):346--354, 2015.

\bibitem{changepoints_package}
H.~Xu, O.~Padilla, D.~Wang, M.~Li, and Q.~Wen.
\newblock {\bf{changepoints}}: {A} {C}ollection of {C}hange-{P}oint {D}etection
  {M}ethods. {{\textsf{R}}} {P}ackage, {V}ersion 1.1.0.
\newblock pages 1--73, 2022.

\bibitem{YAO1988181}
Y.-C. Yao.
\newblock Estimating the number of change-points via {S}chwarz' criterion.
\newblock {\em Statistics \& Probability Letters}, 6(3):181--189, 1988.

\bibitem{yao1989least}
Y.-C. Yao and S.-T. Au.
\newblock Least-{S}quares {E}stimation of a {S}tep {F}unction.
\newblock {\em Sankhy{\=a}: The Indian Journal of Statistics, Series A}, pages
  370--381, 1989.
\newblock \url{https://www.jstor.org/stable/25050759}.

\bibitem{yu2020review}
Y.~Yu.
\newblock A review on minimax rates in change point detection and localisation,
  2020.

\bibitem{yu_localising_2022}
Y.~Yu, S.~Chatterjee, and H.~Xu.
\newblock Localising change points in piecewise polynomials of general degrees.
\newblock {\em Electronic Journal of Statistics}, 16(1):1855--1890, 2022.

\end{thebibliography}


\begin{thebibliography}{4}
\bibitem[Baranowski, Chen and Fryzlewicz(2019)]{baranowski2019narrowest_supp}
R. Baranowski, Y. Chen and P.  Fryzlewicz (2019).  Narrowest-Over-Threshold Detection of
Multiple Change Points and Change-Point-Like Features. \textit{Journal of the Royal Statistical Society: Series B (Statistical Methodology)}, \textbf{81}:3, 649--672.

\end{thebibliography}
\clearpage
\begingroup
\renewcommand{\thepage}{S\arabic{page}} 
\setcounter{page}{1}                    
\renewcommand{\thesection}{S\arabic{section}} 
\setcounter{section}{0}
\renewcommand{\thefigure}{S\arabic{figure}}
\renewcommand{\thetable}{S\arabic{table}}
\renewcommand{\theequation}{S\arabic{equation}}
\setcounter{figure}{0}
\setcounter{table}{0}
\setcounter{equation}{0}
\begin{center}
  \LARGE \textbf{Supplementary Material for ``Data-adaptive structural change-point detection via isolation''} \\[1em]
  \normalsize
  Andreas Anastasiou\textsuperscript{1}, Sophia Loizidou\textsuperscript{2} \\
  \textsuperscript{1}Department of Mathematics and Statistics, University of Cyprus\\
  \textsuperscript{2}Department of Mathematics, University of Luxembourg \\
\end{center}
\vspace{2em}

\noindent\textbf{Abstract:} In this supplement we provide further simulation results, an investigation of the impact of the values of the expansion parameter, $\lambda_T$ , and the threshold constant in the performance of the algorithm, as well as the step-by-step proof of Theorem 2, which shows the consistency of our method in accurately estimating the true number and the locations of the change-points in the case that the underlying signal ft is continuous, piecewise-linear.

\section{Further simulation results} 
\label{sec:further_simulations}
Some further simulation results are presented in Tables \ref{supp:table:justnoise, long signal} and \ref{table:small_dist2, small_dist3, teeth} for piecewise-constant signals and Tables \ref{table:justnoise_wave} and \ref{table:wave4, wave5, wave6} for continuous, piecewise-linear signals. The signals used are:
\begin{itemize}
    \item[{(S12)}] $justnoise$: sequence of length 6000 with no change-points. The standard deviation is $\sigma=1$.
   
    \item[{(S13)}] $long\_signal$: sequence of length 11000 with 1 change-point at 5500 with values before and after the change-point $0$ and $1.5$. The standard deviation is $\sigma=1$. 
    
    
    \item[{(S14)}] $small\_dist3$: sequence of length 1000 with 6 change-points at 100, 130, 485, 515, 870, 900 with values between change-points 0, 1.5, 0, 1, 0, 1.5, 0. The standard deviation is $\sigma=1$. 
    
    \item[{(S15)}] $teeth$: piecewise-constant signal of length 270 with 13 change-points at 11, 31, 51, 71, 91, 111, 131, 151, 171, 191, 211, 231, 251 with values between change-points 0, 1, 0, 1, $\ldots$, 0, 1. The standard deviation is $\sigma=0.4$. 

    \item[{(S16)}] $justnoise\_wave$: piecewise-linear signal without change-points of length 1000. The starting intercept is $f_1=0$ and slope $f_2-f_1=1$. The standard deviation is $\sigma=1$.

    \item[{(S17)}] $wave4$: continuous, piecewise-linear signal of length 200 with 9 change-points at 20, 40, $\ldots$, 180 with changes in slope 1/6, 1/2, -3/4, -1/3, -2/3, 1, 1/4, 3/4, -5/4. The starting intercept is $f_1=-1$ and slope $f_2-f_1=1/32$. The standard deviation is $\sigma=0.3$. 


    \item[{(S18)}] $wave5$: continuous, piecewise-linear signal of length 350 with 50 change-points at 7, 14, $\ldots$, 343 with changes in slope -2.5, 2.5, $\ldots$, 2.5, -2.5. The starting intercept is $f_1=-0$ and slope $f_2-f_1=1$. The standard deviation is $\sigma=1$.
\end{itemize}

The general conclusions form the tables are that DAIS performs at least as well as the competitors, often having an advantage in computational time and accuracy compared to the best performing competitors. 
Signal {(S12)}, in Table~\ref{supp:table:justnoise, long signal} has no change-points. Most algorithms have very good results with negligible differences in the accuracy in the number of change-points detected.
Similar conclusions hold for Signal {(S13)}, in which the large length of the signal forces some algorithms, especially MSCP, to be extremely slow.
Moreover, comparing DAIS with ID\_th which use the same value for the expansion parameter $\lambda_T$, it is evident that DAIS requires a smaller number of expansions and thus is faster.
Signal
{(S14)} in Table~\ref{table:small_dist2, small_dist3, teeth} indicates that in this difficult structure, where pairs of change-points are close to each other and move towards opposite directions such that they `cancel' each-other out, ID\_th has the best performance, followed by DAIS and MOSUM, in terms of accuracy with respect to the number of change-points detected and the MSE.
For the Signal {(S15)} all algorithms perform similarly.
\sla{
Regarding continuous, piecewise-linear signals, the Signal (S16) in Table \ref{table:justnoise_wave} has constant slope and no change-points. In this signal, all algorithms perform very well, apart from the MARS, TF, and local\_poly methods, which seem to exhibit some overdetection issues. Among the top performing algorithms, DAIS and ID\_ic have the lowest computational time. Regarding the performance of the competing methods in Signals (S17) and (S18) as shown in Table~\ref{table:wave4, wave5, wave6}, DAIS and ID (combined with either the thresholding or the information criterion stopping rule) exhibit very good performance in both signals in terms of accuracy for both the estimated number and the estimated locations of change-points. CPOP has slightly worse (but still quite accurate) performance in both signals. NOT and TS are among the top performing methods in (S17), but not in (S18).}

\begin{table}[tbp]
\caption{Distribution of $\hat{N} - N$ over 100 simulated data sequences of the Signals {(S12)} and {(S13)}. 
    The average MSE and computational times are also given.} \label{supp:table:justnoise, long signal}
\centering
        \begin{tabular}{|l|l|c|c|c|c|c|c|c|c|}
            \hline
            &&\multicolumn{5}{|c|}{} &&& \\ 
            &&\multicolumn{5}{|c|}{$\hat{N} - N$} &&& \\
            Method & Signal & -1 &  0 & 1 & 2 & $\geq3$ & MSE & $d_H$ & Time (s) \\
            \hline
            \textbf{DAIS} && - & \textbf{99} & 0 & 1 & 0 & $1.87 \times 10^{-4}$ & - & 0.237 \\
            ID\_th &   &   - &  92 &         3 &         5 &         0 & $3.59 \times 10^{-4}$ & - & 0.458 \\ 
            \textbf{ID\_ic} &    & - &  \textbf{100} &         0 &         0 &         0 & $\boldsymbol{1.39\times 10^{-4}}$ & - & 0.193 \\ 
            WBS\_th &  &  - &    91 &         8 &         1 &         0 & $2.15 \times 10^{-4}$ & - & 0.050 \\ 
            \textbf{WBS\_ic} && - & \textbf{100} &         0 &         0 &         0 & $\boldsymbol{1.39\times 10^{-4}}$ & - & 0.073 \\ 
            WBS2 & & - &     91 &         5 &         3 &         1 & $2.75 \times 10^{-4}$ & - & 3.570 \\ 
            \textbf{PELT} && - & \textbf{100} & 0 & 0 & 0 & $\boldsymbol{1.39\times 10^{-4}}$ & - & 0.003 \\ 
            \textbf{NOT} & {(S12)} & - &   \textbf{100} &         0 &         0 &         0 & $\boldsymbol{1.39\times 10^{-4}}$ & - & 0.085 \\ 
            MOSUM && - & 94 & 2 & 4 & 0 & $2.96\times 10^{-4}$ & - & 0.028 \\ 
            \textbf{MSCP} && - & \textbf{100} & 0 & 0 & 0 & $\boldsymbol{1.39\times 10^{-4}}$ & - & 83.25 \\ 
            SeedBS\_th && - &     88 &        11 &         1 &         0 & $1.59 \times 10^{-4}$ & - & 0.041 \\ 
            \textbf{SeedBS\_ic} && - &   \textbf{99} &         1 &         0 &         0 & $1.68 \times 10^{-4}$ & - & 0.041 \\ 
            DP\_univar && - & 0 & 0 & 0 & 100 & $372 \times 10^{-4}$ & - & 30.70\\
            \sla{SMUCE} &  & \sla{-} & \sla{78} & \sla{22} & \sla{0} & \sla{0} & \sla{$3.60 \times 10^{-4}$} & \sla{-} & \sla{0.070} \\
            \sla{\textbf{MsFPOP}} & & \sla{-} & \sla{\textbf{97}} & \sla{3} & \sla{0} & \sla{0} & \sla{$2.14 \times 10^{-4}$} & \sla{-} & \sla{0.039} \\
            \hline
            \textbf{DAIS} && 0 & \textbf{99} & 0 & 1 & 0 & $\boldsymbol{4.33 \times 10^{-4}}$ & $43.10 \times 10^{-4}$ &  0.322 \\
            ID\_th & &       0 &        93 &         2 &         5 &         0 & $5.61 \times 10^{-4}$ & $135.00 \times 10^{-4}$ & 0.471 \\  
            \textbf{ID\_ic} &  &      0 &       \textbf{100} &         0 &         0 &         0 & $\boldsymbol{4.37 \times 10^{-4}}$ & $\boldsymbol{1.15 \times 10^{-4}}$ & 0.391 \\ 
            \textbf{WBS\_th} &  &      0 &        \textbf{98} &         2 &         0 &         0 & $\boldsymbol{4.33 \times 10^{-4}}$ & $5.74 \times 10^{-4}$ & 0.083 \\ 
            \textbf{WBS\_ic} && 0 &       \textbf{100} &         0 &         0 &         0 & $\boldsymbol{4.19 \times 10^{-4}}$ & $\boldsymbol{1.06\times 10^{-4}}$ & 0.118 \\ 
            WBS2 &  & 0 &        92 &         3 &         4 &         1 & $5.46 \times 10^{-4}$ & $176.00 \times 10^{-4}$ & 6.380 \\ 
            \textbf{PELT} && 0 & \textbf{100} & 0 & 0 & 0 & $\boldsymbol{4.19 \times 10^{-4}}$ & $\boldsymbol{1.06\times 10^{-4}}$ & 0.004 \\
            \textbf{NOT} & {(S13)} & 0 &       \textbf{100} &         0 &         0 &         0 & $\boldsymbol{4.19 \times 10^{-4}}$ & $\boldsymbol{1.06\times 10^{-4}}$ & 0.127 \\ 
            MOSUM && 0 & 92 & 2 & 6 & 0 & $5.49 \times 10^{-4}$ & $203.60 \times 10^{-4}$ & 0.049 \\
            MSCP && 6 & 93 & 1 & 0 & 0 & $341.35 \times 10^{-4}$ & $334.74 \times 10^{-4}$ & 276.294 \\
             SeedBS\_th && 0 &        90 &        10 &         0 &         0 & $\boldsymbol{4.30 \times 10^{-4}}$ & $223.00 \times 10^{-4}$ & 0.052 \\
            \textbf{SeedBS\_ic} &&   0 &       \textbf{100} &         0 &         0 &         0 & $\boldsymbol{4.19 \times 10^{-4}}$ & $\boldsymbol{1.06\times 10^{-4}}$ & 0.052 \\  
            DP\_univar && 0 & 0 & 0 & 0 & 100 & $368 \times 10^{-4}$ & $4840.00 \times 10^{-4}$ & 190.80\\
            \sla{SMUCE} & & \sla{0} & \sla{90} & \sla{7} & \sla{3} & \sla{0} & \sla{$5.08 \times 10^{-4}$} & \sla{$269.00 \times 10^{-4}$} & \sla{0.117} \\
            \sla{\textbf{MsFPOP}} && \sla{0} & \sla{\textbf{99}} & \sla{1} & \sla{0} & \sla{0} & \sla{$\boldsymbol{4.32\times 10^{-4}}$} & \sla{$23.20\times 10^{-4}$} & \sla{0.068} \\
            \hline
        \end{tabular}
\end{table}

\begin{table}[tbp]
\centering
\caption{Distribution of $\hat{N} - N$ over 100 simulated data sequences of the Signals 
{(S14)} and {(S15)}. 
    The average MSE, $d_H$ and computational times are also given.} \label{table:small_dist2, small_dist3, teeth}
        \begin{tabular}{ |l|l|c|c|c|c|c|c|c|c|c|c|c|c|}
            \hline
            &&\multicolumn{7}{|c|}{} &&& \\ 
            &&\multicolumn{7}{|c|}{$\hat{N} - N$} & & & \\
            Method & Signal & $\leq -3$& $-2$ & $-1$ & 0 & 1 & 2 & $\geq 3$ & MSE & $d_H$ & Time (s) \\
            \hline
            {DAIS} && 0 & 15 & 3 & {80} & 0 & 2 & 0 & 0.0321 & 0.067 & 0.004 \\
             \textbf{ID\_th }& &    0 &     2 &     2 &    \textbf{86} &     5 &     4 &     1 & \textbf{0.0303} & \textbf{0.031} & 0.005 \\  
            ID\_ic &   &  2 &    35 &     0 &    63 &     0 &     0 &     0 & 0.0358 & 0.145 & 0.016 \\ 
            WBS\_th &   &  0 &     4 &     7 &    71 &    14 &     3 &     1 & \textbf{0.0307} & 0.049 & 0.009 \\
            WBS\_ic && 2 &    33 &     1 &    64 &     0 &     0 &     0 & 0.0337 & 0.140 & 0.029 \\ 
            WBS2 &  &  0 &     6 &     5 &    70 &     9 &     6 &     4 & \textbf{0.0306} & 0.048 & 0.504 \\ 
            PELT && 26 & 65 & 0 & 9 & 0 & 0 & 0 & 0.0596 & 0.437 & 0.001 \\
            NOT & {(S14)} &   2 &    33 &     1 &    64 &     0 &   0 &     0 & 0.0332 & 0.139 & 0.043 \\ 
            {MOSUM} && 0 & 13 & 2 & {79} & 3 & 3 & 0 & \textbf{0.0287} & 0.063 & 0.015 \\ 
            MSCP &&    95 &     5 &     0 &     0 &     0 &     0 &     0 & 0.1390 & 0.593 & 2.610 \\ 
            SeedBS\_th &&0 &     4 &    23 &    50 &    15 &     7 &     1 & 0.0338 & 0.049 & 0.019 \\ 
            SeedBS\_ic && 1 &    29 &     2 &    63 &     3 &     1 &     1 & 0.0331 & 0.122 & 0.019 \\ 
            DP\_univar && 0 & 0 & 0 & 2 & 6 & 4 & 88 & 0.0602 & 0.126 & 0.139 \\ 
            \sla{SMUCE} & & \sla{0} & \sla{8} & \sla{39} & \sla{48} & \sla{5} & \sla{0} & \sla{0} & \sla{0.0405} & \sla{\textbf{0.034}} & \sla{0.005} \\
            \sla{MsFPOP} & & \sla{2} & \sla{34} & \sla{1} & \sla{63} & \sla{0} & \sla{0} & \sla{0} & \sla{0.0332} & \sla{0.143} & \sla{0.004} \\
            \hline
            \textbf{DAIS} && 0 & 0 & 1 & \textbf{94} & 4 & 1 & 0 & 0.0272 & 0.010 & 0.002 \\
             ID\_th &  &   0 &     0 &     2 &    88 &    10 &     0 &     0 & 0.0270 & 0.012 & 0.003 \\   
            ID\_ic & &     0 &     0 &     0 &    86 &    11 &     3 &     0 & 0.0272 & 0.011 & 0.008 \\ 
            WBS\_th &  &   0 &     0 &     1 &    84 &    14 &     1 &     0 & \textbf{0.0231} & 0.011 & 0.004 \\ 
            \textbf{WBS\_ic} && 0 &     0 &     0 &    \textbf{93} &     7 &     0 &     0 & \textbf{0.0229} & \textbf{0.008} & 0.010 \\ 
            \textbf{WBS2} &  & 0 &     0 &     1 &    \textbf{96} &     3 &     0 &     0 & \textbf{0.0231} & \textbf{0.008} & 0.096 \\ 
            PELT && 1 & 6 & 3 & 90 & 0 & 0 & 0 & 0.0263 & 0.016 & 0.001 \\ 
            \textbf{NOT} & {(S15)} & 0 &     0 &     0 &    \textbf{94} &     6 &     0 &     0 & \textbf{0.0242} & \textbf{0.008} & 0.018 \\ 
            \textbf{MOSUM} && 0 & 0 & 0 & \textbf{95} & 3 & 2 & 0 &  \textbf{0.0226} & \textbf{0.008} & 0.008 \\ 
            MSCP && 100 & 0 & 0 & 0 & 0 & 0 & 0 &  0.2110 & 0.261 & 0.242 \\ 
            SeedBS\_th && 0 &     2 &     9 &    77 &    10 &     2 &     0 & 0.0283 & 0.018 & 0.026 \\ 
            \textbf{SeedBS\_ic} && 0 &     0 &     0 &    \textbf{93} &     5 &     2 &     0 & \textbf{0.0240} & \textbf{0.009} & 0.026 \\
            DP\_univar && 49 & 18 & 2 & 31 & 0 & 0 & 0 &  0.0889 & 0.210 & 0.003 \\
            \sla{SMUCE} &  & \sla{8} & \sla{13} & \sla{25} & \sla{54} & \sla{0} & \sla{0} & \sla{0} & \sla{0.0598} & \sla{0.029} & \sla{0.002} \\
            \sla{MsFPOP} &  & \sla{100} & \sla{0} & \sla{0} & \sla{0} & \sla{0} & \sla{0} & \sla{0} & \sla{0.2500} & \sla{0.930} & \sla{0.001} \\
            \hline
            \end{tabular}
\end{table}

\begin{table}[tbp]
\caption{Distribution of $\hat{N} - N$ over 100 simulated data sequences of the Signal {(S16)}. 
    The average MSE and computational times are also given.} \label{table:justnoise_wave}
\centering
        \begin{tabular}{|l|l|c|c|c|c|c|c|}
            \hline
            &&\multicolumn{4}{|c|}{} && \\ 
            &&\multicolumn{4}{|c|}{$\hat{N} - N$} && \\
            Method & Signal & 0 & 1 & 2 & $\geq3$ & MSE & Time (s) \\
            \hline
            \textbf{DAIS} && \textbf{100} & 0 & 0 & 0 & $\boldsymbol{2.12 \times 10^{-3}}$ & 0.032 \\
            \textbf{ID\_th} & &   \textbf{100} &       0 &       0 &       0 & $\boldsymbol{2.12\times 10^{-3}}$ & 0.049 \\ 
            \textbf{ID\_ic} && \textbf{100} &       0 &       0 &       0 & $\boldsymbol{2.12\times 10^{-3}}$ &   0.033 \\
            \textbf{CPOP} && \textbf{100} & 0 & 0 & 0 & $2.45\times 10^{-3}$ & 3.220 \\ 
            \textbf{NOT} & {(S16)} & \textbf{98} &       2 &       0 &       0 & $2.46\times 10^{-3}$ &  0.427 \\ 
            MARS && 0 & 100 & 0 & 0 & $5.25\times 10^{-3}$ & 0.006 \\ 
            TF && 11 &      51 &      30 &       8 & $5.29\times 10^{-3}$ & 0.420 \\ 
            \textbf{TS} && \textbf{100} & 0 & 0 & 0 & $\boldsymbol{2.12 \times 10^{-3}}$ & 0.618 \\
            local\_poly && 10 & 5 & 17 & 68 & $6.47 \times 10^{-3}$ & 33.000 \\ 
            \hline
        \end{tabular}
\end{table}

            \begin{table}[tbp]
            \caption{Distribution of $\hat{N} - N$ over 100 simulated data sequences of the Signals {(S17)}
            and {(S18)}. The average MSE, $d_H$ and computational times are also given.} \label{table:wave4, wave5, wave6}
\centering
        \begin{tabular}{|l|l|c|c|c|c|c|c|c|c|c|c|}
            \hline
            &&\multicolumn{7}{|c|}{} &&& \\ 
            &&\multicolumn{7}{|c|}{$\hat{N} - N$} &&& \\
            Method & Signal & $\leq -15$ & $(-15,-2]$ & $-1$ & 0 & 1 & $[2,15)$ & $\geq15$ & MSE & $d_H$ & Time (s) \\
            \hline
            \textbf{DAIS} && - &    0 &    0 &   \textbf{96} &    4 &    0 &    0 & 0.022 & 0.104 & 0.003 \\ 
            \textbf{ID\_th} &  &  - &    0 &    0 &   \textbf{95} &    4 &    1 &    0 & 0.021 & 0.108 & 0.002 \\ 
            ID\_ic &&  - &    0 &    0 &   90 &    9 &    1 &    0 & 0.021 & 0.112 & 0.016 \\ 
            CPOP && - &    0 &    0 &   87 &   12 &    1 &    0 & 0.081 & 0.122 & 0.017 \\ 
            \textbf{NOT} & {(S17)} & - &    0 &    0 &   \textbf{96} &    3 &    1 &    0 & \textbf{0.015} & \textbf{0.088} & 0.065 \\ 
            MARS && - &   84 &   13 &    3 &    0 &    0 &    0 & 3.720 & 2.160 & 0.003 \\ 
            TF && - &    0 &    0 &    0 &    0 &   21 &   79 & 0.024 & 0.433 & 0.097 \\ 
            \textbf{TS} && - &    0 &    0 &   \textbf{96} &    4 &    0 &    0 & 0.094 & 0.198 & 0.110 \\ 
            local\_poly && - &  100 &    0 &    0 &    0 &    0 &    0 & 0.331 & 0.956 & 0.095 \\ 
            \hline
            \textbf{DAIS} && 0 &    1 &    2 &   \textbf{96} &    1 &    0 &    0 & \textbf{0.463} & {0.300} & 0.008 \\ 
            \textbf{ID\_th}  & &  0 &    0 &    2 &   \textbf{97} &    1 &    0 &    0 & \textbf{0.428} & \textbf{0.251} & 0.008 \\ 
            \textbf{ID\_ic} &&  0 &    2 &    0 &   \textbf{97} &    1 &    0 &    0 & 0.476 & 0.287 & 0.045 \\ 
            CPOP && 0 &    0 &    0 &   89 &   11 &    0 &    0 & 1.650 & {0.323} & 0.025 \\ 
            NOT & {(S18)} & 1 &    0 &    3 &   79 &   10 &    7 &    0 & 0.611 & 0.831 & 0.094 \\ 
            MARS && 100 &    0 &    0 &    0 &    0 &    0 &    0 & 6.590 & 48.00 & 0.003 \\ 
            TF && 0 &    0 &    0 &    0 &    0 &   82 &   18 & 1.200 & {0.323} & 0.110 \\ 
            TS && 100 &    0 &    0 &    0 &    0 &    0 &    0 & 6.650 &  49.00 & 0.197 \\ 
            local\_poly && 100 &    0 &    0 &    0 &    0 &    0 &    0 & 6.610 & 5.220 & 0.639 \\ 
            \hline
        \end{tabular}
\end{table}

\section{Investigation of the impact of parameter choice on accuracy}\label{sec: impact_params}
In this section, we investigate the impact of the expansion parameter, $\lambda_T$, and the threshold constant, $C$, in the accuracy of DAIS using a small simulation study.
The setup and signals used are the same as in Section~\ref{simulations} of the main paper and Section~\ref{sec:further_simulations}.
The values used for the expansion parameter are $\lambda_T \in \{1,3,5,10,15,20\}$ and for the threshold constant $C\in\{1.5, 1.6, \ldots,1.9\}$ for the case of piecewise-constant signals and $C\in\{1.9,2,\ldots,2.3\}$ for continuous, piecewise-linear signals.
The results in Table~\ref{tab: increasing_lambda} indicate that our algorithm is robust to small and moderate changes of $\lambda_T$.
Tables~\ref{tab: thr_const_const} and \ref{tab: thr_const_lin} indicate that any value of the threshold constant in the intervals $[1.5,1.9]$ for piecewise-constant and $[1.9,2.3]$ for continuous, piecewise-linear signals has good estimation accuracy.
\begin{table}[!h]
\sla{
\caption{Distribution of $\hat{N} - N$ over 100 simulated data sequences of signals for $\lambda_T \in \{1,3,5,10,15,20\}$, for the piecewise-constant Signals (S1), (S3), (S4), (S12) and (S15), and the continuous, piecewise-linear Signals (S9) and (S10). The average MSE, $d_H$ and computational times are also given.}
\label{tab: increasing_lambda}
\centering
        \begin{tabular}{|l|c|c|c|c|c|c|c|c|c|c|c|}
            \hline
            &&\multicolumn{7}{|c|}{} &&& \\ 
            &&\multicolumn{7}{|c|}{$\hat{N} - N$} &&& \\
            $\lambda_T$ & Signal & $\leq 3$ & -2 & -1 &0 & 1 & 2 & $\geq 3$ & MSE & $d_H$ &Time (s) \\
            \hline
             1 & &     - &     11 &      4 &     82 &      1 &      1 &      1 & 0.013 & 0.076 & 0.010 \\ 
            3 &   &   - &     13 &      2 &     85 &      0 &      0 &      0 & 0.012 & 0.073 & 0.004 \\ 
            5 & (S1)  &   - &     16 &      3 &     79 &      2 &      0 &      0 & 0.013 & 0.095 & 0.003 \\ 
            10 &  &    - &     18 &      2 &     77 &      2 &      1 &      0 & 0.013 & 0.102 & 0.002 \\ 
            15 &  &    - &     21 &      2 &     76 &      0 &      1 &      0 & 0.014 & 0.114 & 0.001 \\ 
            20 &   &   - &     22 &      2 &     73 &      2 &      1 &      0 & 0.014 & 0.122 & 0.001 \\ 
            \hline
            1 & &    0 &     0 &     1 &    95 &     4 &     0 &     0 & 0.022 & 0.009 & 0.003 \\ 
            3 &  &   0 &     0 &     0 &    98 &     2 &     0 &     0 & 0.022 & 0.009 & 0.003 \\ 
            5 & (S3) &   0 &     0 &     2 &    97 &     1 &     0 &     0 & 0.022 & 0.009 & 0.002 \\ 
            10 &  &   0 &     0 &     5 &    95 &     0 &     0 &     0 & 0.025 & 0.011 & 0.002 \\ 
            15 &   &  0 &     0 &     3 &    92 &     5 &     0 &     0 & 0.026 & 0.012 & 0.002 \\ 
            20 &   &  0 &     0 &     3 &    91 &     5 &     1 &     0 & 0.028 & 0.013 & 0.003 \\ 
            \hline
            1 &   & 0 &    1 &    0 &   89 &    9 &    1 &    0 & 1.800 & 0.017 & 0.003 \\ 
            3 &  &  0 &    1 &    0 &   89 &    9 &    1 &    0 & 1.750 & 0.018 & 0.002 \\ 
            5 & (S4) &   0 &    1 &    0 &   92 &    7 &    0 &    0 & 1.790 & 0.016 & 0.002 \\ 
            10 &  &  0 &    1 &    0 &   95 &    4 &    0 &    0 & 1.750 & 0.015 & 0.002 \\ 
            15 & &   0 &    4 &    0 &   91 &    5 &    0 &    0 & 1.830 & 0.018 & 0.002 \\ 
            20 & &   0 &    6 &    0 &   93 &    1 &    0 &    0 & 1.890 & 0.018 & 0.002 \\ 
            \hline
            1 &  &  0 &    0 &    0 &   98 &    2 &    0 &    0 & 0.031 & 0.089 & 0.028 \\ 
            3 &  &  0 &    0 &    0 &   98 &    2 &    0 &    0 & 0.030 & 0.086 & 0.011 \\ 
            5 & (S9) &  0 &    0 &    0 &   99 &    1 &    0 &    0 & 0.029 & 0.085 & 0.007 \\ 
            10 &  &  0 &    0 &    0 &   98 &    2 &    0 &    0 & 0.028 & 0.083 & 0.004 \\ 
            15 &  &  0 &    0 &    0 &   99 &    1 &    0 &    0 & 0.026 & 0.082 & 0.004 \\ 
            20 &  &  0 &    0 &    0 &  100 &    0 &    0 &    0 & 0.025 & 0.072 & 0.003 \\ 
            \hline
            1 &  &  0 &    0 &    0 &   96 &    4 &    0 &    0 & 0.292 & 0.325 & 0.034 \\ 
            3 &  &  0 &    0 &    0 &   98 &    2 &    0 &    0 & 0.265 & 0.297 & 0.024 \\ 
            5 &  (S10) &  0 &    1 &    0 &   97 &    2 &    0 &    0 & 0.246 & 0.299 & 0.023 \\ 
            10 &  &  0 &    0 &    0 &   99 &    1 &    0 &    0 & 0.215 & 0.251 & 0.021 \\ 
            15 &  &  0 &    0 &    0 &  100 &    0 &    0 &    0 & 0.228 & 0.247 & 0.021 \\ 
            20 &  &  1 &    0 &    0 &   81 &   16 &    2 &    0 & 0.375 & 0.385 & 0.021 \\ 
            \hline
            1 &&      - &      - &      - &     99 &      0 &      1 &      0 & $1.9 \times 10^{-4}$ &    - & 0.599 \\ 
            3 &&      - &      - &      - &    100 &      0 &      0 &      0 & $1.5\times 10^{-4}$ &    - & 0.209 \\ 
            5 & (S12) &      - &      - &      - &    100 &      0 &      0 &      0 & $1.5\times 10^{-4}$ &    - & 0.124 \\ 
            10 &&     - &      - &      - &    100 &      0 &      0 &      0 & $1.5\times 10^{-4}$ &    - & 0.062 \\ 
            15 &&      - &      - &      - &    100 &      0 &      0 &      0 & $1.5\times 10^{-4}$ &    - & 0.037 \\ 
            20 &&      - &      - &      - &    100 &      0 &      0 &      0 & $1.5\times 10^{-4}$ &    - & 0.032 \\ 
            \hline
            1 & &   0 &    1 &    0 &   92 &    7 &    0 &    0 & 0.028 & 0.011 & 0.004 \\ 
            3 & &   0 &    1 &    0 &   94 &    5 &    0 &    0 & 0.027 & 0.011 & 0.003 \\ 
            5 & (S15) &  0 &    1 &    0 &   96 &    3 &    0 &    0 & 0.026 & 0.010 & 0.003 \\ 
            10 & &   0 &    1 &    0 &   97 &    1 &    1 &    0 & 0.025 & 0.009 & 0.002 \\ 
            15 &  &  0 &    0 &    0 &  100 &    0 &    0 &    0 & 0.025 & 0.008 & 0.003 \\ 
            20 & &   1 &    0 &    0 &   98 &    1 &    0 &    0 & 0.026 & 0.011 & 0.003 \\ 
            \hline
            \end{tabular}
}
\end{table}

\begin{table}[htbp]
\sla{
\caption{Distribution of $\hat{N} - N$ over 100 simulated data sequences of signals for $C\in\{1.5, 1.6, 1.7, 1.8, 1.9\}$, for the piecewise-constant Signals (S1), (S3), (S4) and (S14). The average MSE, $d_H$ and computational times are also given.\label{tab: thr_const_const}}
\centering
        \begin{tabular}{|l|c|c|c|c|c|c|c|c|c|c|c|}
            \hline
            &&\multicolumn{7}{|c|}{} &&& \\ 
            &&\multicolumn{7}{|c|}{$\hat{N} - N$} &&& \\
            $C$ & Signal & $\leq 3$ & -2 & -1 &0 & 1 & 2 & $\geq 3$ & MSE & $d_H$ &Time (s) \\
            \hline
            1.5 &&    - &    4 &    2 &   68 &   12 &   12 &    2 & 0.015 & 0.090 & 0.004 \\ 
            1.6 &  &  - &    8 &    2 &   85 &    3 &    1 &    1 & 0.013 & 0.062 & 0.004 \\ 
            1.7 & (S1)  &  - &   14 &    2 &   84 &    0 &    0 &    0 & 0.013 & 0.078 & 0.004 \\ 
            1.8 & &   - &   21 &    2 &   77 &    0 &    0 &    0 & 0.014 & 0.113 & 0.004 \\ 
            1.9 &   & - &   28 &    4 &   68 &    0 &    0 &    0 & 0.016 & 0.149 & 0.004 \\
            \hline
            1.5 & &   0 &    0 &    0 &   91 &    8 &    1 &    0 & 0.021 & 0.009 & 0.002 \\ 
            1.6 &  &   0 &    0 &    0 &   95 &    5 &    0 &    0 & 0.021 & 0.008 & 0.002 \\ 
            1.7 & (S3) &    0 &    0 &    0 &   98 &    2 &    0 &    0 & 0.022 & 0.009 & 0.002 \\ 
            1.8 & &    0 &    0 &    3 &   96 &    1 &    0 &    0 & 0.022 & 0.010 & 0.002 \\ 
            1.9 & &    0 &    1 &   10 &   89 &    0 &    0 &    0 & 0.024 & 0.013 & 0.002 \\ 
            \hline
            1.5 &  &  0 &    0 &    0 &   72 &   21 &    3 &    4 & 1.890 & 0.022 & 0.002 \\ 
            1.6 &  &  0 &    0 &    0 &   82 &   15 &    3 &    0 & 1.780 & 0.018 & 0.002 \\ 
            1.7 & (S4) &  0 &    1 &    0 &   89 &    9 &    1 &    0 & 1.750 & 0.018 & 0.002 \\ 
            1.8 &  &  0 &    2 &    0 &   96 &    1 &    1 &    0 & 1.730 & 0.016 & 0.002 \\ 
            1.9 & &   0 &    3 &    0 &   96 &    0 &    1 &    0 & 1.760 & 0.017 & 0.002 \\ 
            \hline
            1.5 & & - &    - &    - &   75 &    6 &   16 &    3 & 0.001 &  - & 0.197 \\ 
            1.6 &  &  - &    - &    - &   93 &    2 &    5 &    0 & 0.001 &  - & 0.222 \\ 
            1.7 & (S12) &   - &    - &    - &   99 &    0 &    1 &    0 & 0.001 &  - & 0.216 \\ 
            1.8 &  &  - &    - &    - &   99 &    0 &    1 &    0 & 0.001 &  - & 0.223 \\ 
            1.9 & &   - &    - &    - &  100 &    0 &    0 &    0 & 0.001 &  - & 0.231 \\ 
            \hline
            1.5 & & 0 &    0 &    0 &   74 &   16 &    6 &    4 & 0.027 & 0.013 & 0.002 \\ 
            1.6 &  &    0 &    0 &    0 &   89 &    7 &    2 &    2 & 0.027 & 0.011 & 0.002 \\ 
            1.7 & (S15) &    0 &    1 &    0 &   94 &    5 &    0 &    0 & 0.027 & 0.011 & 0.002 \\ 
            1.8 & &    1 &    1 &    0 &   96 &    2 &    0 &    0 & 0.027 & 0.013 & 0.002 \\ 
            1.9 & &    1 &    4 &    1 &   93 &    1 &    0 &    0 & 0.028 & 0.017 & 0.002 \\ 
            \hline
            \end{tabular}
}
\end{table}

\begin{table}[htbp]
\sla{
\caption{Distribution of $\hat{N} - N$ over 100 simulated data sequences of signals for $C\in\{1.9, 2, 2.1, 2.2\}$, for the continuous, piecewise-linear Signals (S9) and (S10). The average MSE, $d_H$ and computational times are also given.\label{tab: thr_const_lin}}
\centering
        \begin{tabular}{|l|c|c|c|c|c|c|c|c|c|c|c|}
            \hline
            &&\multicolumn{7}{|c|}{} &&& \\ 
            &&\multicolumn{7}{|c|}{$\hat{N} - N$} &&& \\
            $C$ & Signal & $\leq 3$ & -2 & -1 &0 & 1 & 2 & $\geq 3$ & MSE & $d_H$ &Time (s) \\
            \hline
            1.9 &  &  0 &    0 &    0 &   92 &    8 &    0 &    0 & 0.032 & 0.097 & 0.009 \\ 
            2 &  &  0 &    0 &    0 &   97 &    3 &    0 &    0 & 0.032 & 0.091 & 0.010 \\ 
            2.1 &  (S9) &  0 &    0 &    0 &   98 &    2 &    0 &    0 & 0.030 & 0.087 & 0.010 \\ 
            2.2 &  &  0 &    0 &    0 &   99 &    1 &    0 &    0 & 0.029 & 0.083 & 0.009 \\ 
            2.3 &  &  0 &    0 &    0 &  100 &    0 &    0 &    0 & 0.028 & 0.075 & 0.010 \\ 
            \hline
            1.9 &  &   0 &    0 &    0 &   92 &    7 &    1 &    0 & 0.269 & 0.304 & 0.018 \\ 
            2 &  &   0 &    0 &    0 &   96 &    4 &    0 &    0 & 0.265 & 0.297 & 0.017 \\ 
            2.1 & (S10) &    0 &    0 &    0 &   99 &    1 &    0 &    0 & 0.265 & 0.295 & 0.017 \\ 
            2.2 & &    0 &    1 &    0 &   97 &    2 &    0 &    0 & 0.267 & 0.311 & 0.018 \\ 
            2.3 & &    0 &    1 &    0 &   95 &    4 &    0 &    0 & 0.270 & 0.329 & 0.032 \\ 
            \hline
            \end{tabular}
}
\end{table}
\section{Proof of Theorem~\ref{consistency_theorem_slope}} \label{sec: proof_theorem_slope}
For the proof of Theorem~\ref{consistency_theorem_slope}, we require the following two lemmas.
\begin{lemma} \label{Lemma_for_main_thm_slope1}
Suppose $\boldsymbol{f} = \left( f_1, f_2, \dots, f_{T} \right)^T$ is a piecewise-linear vector and $r_1, r_2, \ldots, r_N$ are the locations of the change-points. 
Suppose $1 \leq s < e \leq T$, such that $r_{j-1} \leq s < r_j < e \leq r_{j+1}$, for some $j=1,2,\ldots, N$. Let $\Delta^f_j = \left| 2f_{r_j} - f_{r_j+1} - f_{r_j-1}\right|$ and $\eta = \min\{r_j-s, e-r_j\}$. Then,
\begin{equation*}
    C^{r_j}_{s,e}(\boldsymbol{f}) = \max_{s<b<e} C^{b}_{s,e}(\boldsymbol{f})\left\{\begin{array}{l}
    \geq \frac{1}{\sqrt{24}} \eta^{3/2} \Delta_j^f, \\
    \leq \frac{1}{\sqrt{3}} (\eta + 1)^{3/2} \Delta_j^f
\end{array}\right.
\end{equation*}
\end{lemma}
\begin{proof}
    See Lemma 5 from \cite{baranowski2019narrowest_supp}.
\end{proof}

\begin{lemma} \label{Lemma_for_main_thm_slope2}
Suppose $\boldsymbol{f} = \left( f_1, f_2, \dots, f_{T} \right)^T$ is a piecewise-linear vector and $r_1, r_2, \ldots, r_N$ are the locations of the change-points. Suppose $1 \leq s < e \leq T$, such that $r_{j-1} \leq s < r_j < e \leq r_{j+1}$, for some $j=1,2,\ldots, N$. Let $\rho = \left| r-b \right|, \Delta^f_j = \left| 2f_{r_j} - f_{r_j+1} - f_{r_j-1}\right|, \eta_L = r_j-s$ and $\eta_R = e-r_j$. Then,
\begin{equation*}
    \|\boldsymbol{\phi_{s,e}^b}\langle\boldsymbol{f},\boldsymbol{\phi_{s,e}^b}\rangle - \boldsymbol{\phi_{s,e}^{r}}\langle\boldsymbol{f},\boldsymbol{\phi_{s,e}^{r}}\rangle\|_{2}^2 = \left( C^{r_j}_{s,e}(\boldsymbol{f}) \right)^2 - \left( C^{b}_{s,e}(\boldsymbol{f}) \right)^2.
\end{equation*}
In addition,
\begin{enumerate}
    \item for any $r_j \leq b < e$, $\left( C^{r_j}_{s,e}(\boldsymbol{f}) \right)^2 - \left( C^{b}_{s,e}(\boldsymbol{f}) \right)^2 = \frac{1}{63}\min\{\rho, \eta_L\}^3 \Big( \Delta^f_j \Big)^2$;
    
    \item for any $s \leq b < r_j$, $\left( C^{r_j}_{s,e}(\boldsymbol{f}) \right)^2 - \left( C^{b}_{s,e}(\boldsymbol{f}) \right)^2 = \frac{1}{63}\min\{\rho, \eta_R\}^3 \Big( \Delta^f_j \Big)^2$.
\end{enumerate}
\end{lemma}
\begin{proof}
    See Lemma 7 from \cite{baranowski2019narrowest_supp}.
\end{proof}

The steps of the proof of Theorem~\ref{consistency_theorem_slope} are the same as for Theorem~\ref{consistency_theorem}, which can be found in Appendix~\ref{proofs}.
\sla{We work under the notation in \eqref{eq: def minimum magnitude of change}, \eqref{eq: def magnitude of change} and \eqref{eq: def distance between cpts}.}
\vspace{0.2in}
\\
\begin{proof}
We will prove the more specific result
\begin{equation}
\label{mainresult_theorem_slope}
\mathbb{P}\Biggl( \hat{N} = N, \max_{j=1, 2, \ldots, N} \biggl( \left| \hat{r}_j - r_j \right| \left( \Delta^{f}_j \right)^{2/3} \biggl) \leq C_{3} (\log T)^{1/3} \Biggl) \geq 1 - \frac{1}{6\sqrt{\pi}{T}},
\end{equation}
which implies result \eqref{main result slope} of the main paper.
\vspace{0.1in}
\\
{\textbf{Steps 1 \& 2:}} 
Similar to Theorem~\ref{consistency_theorem}, denote
\begin{align}\label{A_T_slope}
& A_{T}^{\ast} = \left\lbrace \max_{s,b,e: 1\leq s \leq b < e \leq T}\left|C^{b}_{s,e}(\boldsymbol{X}) - C^{b}_{s,e}(\boldsymbol{f}) \right|\leq \sqrt{8\log T} \right\rbrace \nonumber\\
& B_{T}^{\ast} = \left\lbrace \max_{j=1,2\ldots,N} \max_{\substack{r_{j-1}<s\leq r_j\\r_j < e \leq r_{j+1}\\s\leq b < e}} \frac{\left|\left\langle\boldsymbol{\phi_{s,e}^b}\langle\boldsymbol{f},\boldsymbol{\phi_{s,e}^b}\rangle - \boldsymbol{\phi_{s,e}^{r}}\langle\boldsymbol{f},\boldsymbol{\phi_{s,e}^{r}}\rangle,\boldsymbol{\epsilon}\right\rangle\right|}{\|\boldsymbol{\phi_{s,e}^b}\langle\boldsymbol{f},\boldsymbol{\phi_{s,e}^b}\rangle - \boldsymbol{\phi_{s,e}^{r}}\langle\boldsymbol{f},\boldsymbol{\phi_{s,e}^{r}}\rangle\|_{2}}\leq \sqrt{8 \log T}\right\rbrace.
\end{align}
The same reasoning as in the proof of Theorem~\ref{consistency_theorem} leads to $\Prob\left(A_{T}^{\ast} \right) \geq 1-1/(12\sqrt{\pi}{T})$ and $\Prob\left(B_{T}^{\ast} \right) \geq 1-1/(12\sqrt{\pi}{T})$. Therefore, it holds that
\begin{equation*}
    \Prob\left(A_{T}^{\ast} \cap B_{T}^{\ast} \right) \geq 1-\frac{1}{6\sqrt{\pi}{T}}
\end{equation*}
{\textbf{Step 3:}} From now on, we assume that $A_T^{\ast}$ and $B_T^{\ast}$ both hold. The constants we use are
\begin{equation*}
    C_1 = \sqrt{\frac{2}{3}} C_3^{\frac{3}{2}} + \sqrt{8}, \quad C_2 = \frac{1}{8\sqrt{3n^3}} - \frac{2\sqrt{2}}{C^{\ast}}, \quad C_3 = 63^{\frac{1}{3}}(2\sqrt{2}+4)^\frac{2}{3},
\end{equation*}
where $C^{\ast}$ satisfies Assumption (A3), $\delta_T ^ {3/2}\underline{f}_T \geq C^{\ast} \sqrt{\log T}$ and 
\slb{$n \geq 3/2$}.
As before, for $j \in \{1, 2, \ldots, N\}$ define $I^L_j$ and $I^R_j$ as in \eqref{intervals_discussion}.
For $d_{s,e} = \textrm{argmax}_{t \in \{s, s+1, \dots, e-2\}}\{\lvert X_{t+2} - 2X_{t+1} + X_{t} \rvert\}$ being the location of the largest difference detected in the interval $[s,e]$, $1\leq s<e\leq T$, define $c_{m}^l$ and $c_{k}^r$ as in \eqref{end-points_proof} of the main paper. 
\slb{Since the length of the intervals in \eqref{intervals_discussion} is $\delta_T/2n$,
$\lambda_T \leq \delta_T/2n$ ensures that there exists at least one $m\in \{0,1,\ldots,K^{\max}\}$ and at least one $k\in \{1,2\ldots,K^{\max}\}$}
such that $c_m^l \in I^L_j$ and $c_k^r \in I^R_j$ for all $j \in \{1, 2, \ldots, N\}$.
\vspace{0.1in}
\\
At the beginning of DAIS, $s=1, e=T$ and the first change-point that will get detected depends on the value of $d_{1,T}$. 
As already explained in Section~\ref{discussion}, for $j\in\{1,\ldots,N-1\}$, the largest difference $d_{s,e}$ will be at most at a distance $\delta_{1,T}^j$ from the nearest change-point $r_j$ or $r_{j+1}$, where $\delta_{1,T}^j \leq \frac{\Tilde{\delta}_j}{2} - \frac{3\delta_T}{4n}$ for $\Tilde{\delta}_j = r_{j+1} - r_j$. 
The first point to get detected will be the point that is closest to the largest difference $d_{1,T}$. 
The interval where the detection of this change-point occurs, cannot contain more than one change-points, as was explained in Section~\ref{discussion}.
\vspace{0.1in}
\\
Without loss of generality, we suppose that the first change-point to get detected is $r_J$ for some $J \in \left\lbrace 1,2,\ldots, N\right\rbrace$. 
We can show that there exists an interval $[c_m^l, c_k^r]$, for $m,k \in \{0,1,\ldots,K^{\max}\}$, such that $r_J$ is isolated using exactly the same argument as in (B7) of the main paper. 
We will now show that for $\tilde{b}_J = \textrm{argmax}_{c_m^l\leq t < c_k^r} C^t_{c_m^l, c_k^r}(\boldsymbol{X})$, then $C^{\Tilde{b}_J}_{c_m^l, c_k^r}(\boldsymbol{X}) > \zeta_T$. 
Using \eqref{A_T_slope}, we have that
\begin{equation} \label{initial_bound_slope}
    C^{\Tilde{b}_J}_{c_m^l, c_k^r}(\boldsymbol{X}) 
    \geq C^{r_J}_{c_m^l, c_k^r}(\boldsymbol{X}) 
    \geq C^{r_J}_{c_m^l, c_k^r}(\boldsymbol{f}) - \sqrt{8\log T}.
\end{equation}
From Lemma \ref{Lemma_for_main_thm_slope1}, we have that 
\begin{equation*}
    C^{r_J}_{c_m^l, c_k^r}(\boldsymbol{f}) 
    \geq \frac{1}{\sqrt{24}} \biggl( \min \{r_J-c_m^l, c_k^r - r_J \} \biggr) ^ \frac{3}{2} \Delta_J^f.
\end{equation*}
Showing that
\begin{equation} \label{bound_for_min_slope}
    \min\{c_k^r-r_J,r_J-c_m^l\}\geq \frac{\delta_T}{2n}.
\end{equation}
follows the exact same steps as Step 3 in the proof of Theorem~\ref{consistency_theorem} and will not be repeated. 
Now, using Assumption (A3) and the results in \eqref{initial_bound_slope}, \eqref{bound_for_min_slope}, we have that
\begin{align}\label{slope_observed}
    C^{\Tilde{b}_J}_{c_m^l, c_k^r}(\boldsymbol{X}) 
    & \geq \frac{1}{\sqrt{24}} \left( \frac{\delta_T}{2n} \right)^\frac{3}{2} \Delta^f_J - \sqrt{8\log T}
    \geq \frac{1}{\sqrt{24}} \left( \frac{\delta_T}{2n} \right)^\frac{3}{2}\underline{f_T} - \sqrt{8\log T} \nonumber \\
    & = \biggl( \frac{1}{8\sqrt{3n^3}} - \frac{2\sqrt{2\log T}}{\delta_T^{3/2}\underline{f}_T} \bigg) \delta_T^{3/2} \underline{f}_T
    \geq \biggl( \frac{1}{8\sqrt{3n^3}} - \frac{2\sqrt{2}}{C^{\ast}} \bigg) \delta_T^{3/2} \underline{f}_T \\
    & = C_2 \delta_T^{3/2} \underline{f}_T > \zeta_T \nonumber
\end{align}
and thus, the change-point will get detected.
\vspace{0.1in}
\\
Therefore, there will be an interval of the form $[c_{\tilde{m}}^l, c_{\tilde{k}}^r]$, such that the interval contains $r_J$ and no other change-point and $\max_{c_{\tilde{m}}^l \leq t < c_{\tilde{k}}^r} C^{t}_{c_{\tilde{m}}^l, c_{\tilde{k}}^r}(\boldsymbol{X}) > \zeta_T$. 
For $k^{\ast}, m^{\ast} \in \{0, 1, \ldots K^{\max}\}$, denote by $c_{m^{\ast}}^l \geq c_{\tilde{m}}^l$ and $c_{k^{\ast}}^r \leq c_{\tilde{k}}^r$ the first left- and right-expanding points, respectively, that this happens and let $b_J = \textrm{argmax}_{c_{m^{\ast}}^l \leq t < c_{k^{\ast}}^r} C^{t}_{c_{{m^{\ast}}}^l, c_{{k^{\ast}}}^r}(\boldsymbol{X})$, with $C^{b_J}_{c_{{m^{\ast}}}^l, c_{{k^{\ast}}}^r}(\boldsymbol{X}) > \zeta_T$. 
Note that $b_J$ cannot be an estimation of $r_j, j\neq J$, as $r_J$ is isolated in the interval where it is detected. 
Our aim now is to find $\Tilde{\gamma}_T > 0$, such that for any $b^{\ast} \in \{c_{m^{\ast}}^l, c_{m^{\ast}}^l + 1, \ldots, c_{k^{\ast}}^r-1\}$ with $\lvert b^{\ast} - r_J\rvert \Bigl( \Delta_J^f \Bigl)^2 > \Tilde{\gamma}_T$, we have that
\begin{equation}\label{contradiction_slope}
    \Bigl( C^{r_J}_{c_{{m^{\ast}}}^l, c_{{k^{\ast}}}^r}(\boldsymbol{X}) \Bigl)^2 > \Bigl( C^{b^{\ast}}_{c_{{m^{\ast}}}^l, c_{{k^{\ast}}}^r}(\boldsymbol{X}) \Bigl)^2.
\end{equation}
Proving \eqref{contradiction_slope} and using the definition of $b_J$, we can conclude that $\lvert b_J - r_J\rvert \Bigl( \Delta_J^f \Bigl)^2 \leq \Tilde{\gamma}_T$. 
Now, using \eqref{model_sigma1}, it can be shown, for $\boldsymbol{\phi_{s,e}^{b}}$ as defined in \eqref{phi_definition}, that \eqref{contradiction_slope} is equivalent to
\begin{align} \label{equivalent_contradiction_slope}
    \Bigl( C^{r_J}_{c_{{m^{\ast}}}^l, c_{{k^{\ast}}}^r}(\boldsymbol{f}) \Bigl)^2 - \Bigl( & C^{b^{\ast}}_{c_{{m^{\ast}}}^l, c_{{k^{\ast}}}^r}(\boldsymbol{f}) \Bigl)^2
    > \Bigl( C^{b^{\ast}}_{c_{{m^{\ast}}}^l, c_{{k^{\ast}}}^r}(\boldsymbol{\epsilon}) \Bigl)^2 - \Bigl( C^{r_J}_{c_{{m^{\ast}}}^l, c_{{k^{\ast}}}^r}(\boldsymbol{\epsilon}) \Bigl)^2 \nonumber\\ 
    &+ 2\Big \langle \boldsymbol{\phi_{c_{m^{\ast}}^l, c_{k^{\ast}}^r}^{b^{\ast}}} \langle \boldsymbol{f}, \boldsymbol{\phi_{c_{m^{\ast}}^l, c_{k^{\ast}}^r}^{b^{\ast}}} \rangle - \boldsymbol{\phi_{c_{m^{\ast}}^l, c_{k^{\ast}}^r}^{r_J}} \langle \boldsymbol{f}, \boldsymbol{\phi_{c_{m^{\ast}}^l, c_{k^{\ast}}^r}^{r_J}} \rangle, \boldsymbol{\epsilon} \Big \rangle.
\end{align}
Without loss of generality, assume that $b^{\ast} \in [r_J, c_{k^{\ast}}^r)$ and a similar approach holds when $b^{\ast} \in [c_{m^{\ast}}^l, r_J)$. 
Denote
\begin{equation} \label{Lambda_definition_slope}
    \Lambda \colon = \Bigl( C^{r_J}_{c_{{m^{\ast}}}^l, c_{{k^{\ast}}}^r}(\boldsymbol{f}) \Bigl)^2 - \Bigl( C^{b^{\ast}}_{c_{{m^{\ast}}}^l, c_{{k^{\ast}}}^r}(\boldsymbol{f}) \Bigl)^2.
\end{equation}
For the right-hand side of \eqref{equivalent_contradiction_slope} using \eqref{A_T_slope},
\begin{equation} \label{eq_epsilon_slope}
    \Bigl( C^{b^{\ast}}_{c_{{m^{\ast}}}^l, c_{{k^{\ast}}}^r}(\boldsymbol{\epsilon}) \Bigl)^2 - \Bigl( C^{r_J}_{c_{{m^{\ast}}}^l, c_{{k^{\ast}}}^r}(\boldsymbol{\epsilon}) \Bigl)^2 
    \leq \max_{s,e,b:s\leq b<e} \Bigl( C^{b}_{s, e}(\boldsymbol{\epsilon}) \Bigl)^2 - \Bigl( C^{r_J}_{c_{{m^{\ast}}}^l, c_{{k^{\ast}}}^r}(\boldsymbol{\epsilon}) \Bigl)^2 
    \leq 8\log T.
\end{equation}
Using Lemma \ref{Lemma_for_main_thm_slope2}, \eqref{A_T_slope} and \eqref{Lambda_definition_slope}, we have that for the left-hand side,
\begin{align} \label{eq_extra_stuff_slope}
    2\Big \langle \boldsymbol{\phi_{c_{m^{\ast}}^l, c_{k^{\ast}}^r}^{b^{\ast}} }
    & \langle \boldsymbol{f}, \boldsymbol{\phi_{c_{m^{\ast}}^l, c_{k^{\ast}}^r}^{b^{\ast}}} \rangle - \boldsymbol{\phi_{c_{m^{\ast}}^l, c_{k^{\ast}}^r}^{r_J}} \langle \boldsymbol{f}, \boldsymbol{\phi_{c_{m^{\ast}}^l, c_{k^{\ast}}^r}^{r_J}} \rangle, \boldsymbol{\epsilon} \Big \rangle \nonumber\\
    &\leq 2 \big\| \boldsymbol{\phi_{c_{m^{\ast}}^l, c_{k^{\ast}}^r}^{b^{\ast}}} \langle \boldsymbol{f}, \boldsymbol{\phi_{c_{m^{\ast}}^l, c_{k^{\ast}}^r}^{b^{\ast}}} \rangle - \boldsymbol{\phi_{c_{m^{\ast}}^l, c_{k^{\ast}}^r}^{r_J}} \langle \boldsymbol{f}, \boldsymbol{\phi_{c_{m^{\ast}}^l, c_{k^{\ast}}^r}^{r_J}} \rangle \big\|_2 \sqrt{8\log T} \nonumber\\
    & = 2\sqrt{\Lambda}\sqrt{8\log T} .
\end{align}
Using \eqref{Lambda_definition_slope}, \eqref{eq_epsilon_slope} and \eqref{eq_extra_stuff_slope}, we can conclude that \eqref{equivalent_contradiction_slope} is satisfied if $\Lambda > 8\log T + \sqrt{2} \sqrt{\Lambda}\sqrt{8\log T}$ is satisfied, which has solution
\begin{equation*}\label{lambda_bound}
    \Lambda > \left(2\sqrt{2} +4\right)^2 \log T.
\end{equation*}
Using Lemma \ref{Lemma_for_main_thm_slope2}, we can conclude that
\begin{align}\label{lambda_expression}
    & \Lambda > \left(2\sqrt{2} +4\right)^2 \log T \nonumber\\
    & \Leftrightarrow \frac{1}{63} \left( \min\{\lvert r_J-b^{\ast} \rvert, r_J - c_{m^{\ast}}^l\} \right)^3 \Bigl( \Delta_J^f \Bigr)^2 > \left(2\sqrt{2} +4\right)^2 \log T \nonumber\\
    & \Leftrightarrow \min\{\lvert r_J-b^{\ast} \rvert, r_J - c_{m^{\ast}}^l\} > \frac{\left(63\log T\right)^{1/3} \left( 2\sqrt{2} +4 \right)^{2/3}}{\Bigl( \Delta_J^f \Bigr)^{2/3}} = \frac{C_3\left(\log T\right)^{1/3}}{\Bigl( \Delta_J^f \Bigr)^{2/3}}
\end{align}
Now, if for sufficiently large T
\begin{equation} \label{min_bound3_slope}
    \min\{r_J - c_{m^{\ast}}^l, c_{k^{\ast}}^r - r_J\} 
    > 2^{1/3}C_3\frac{\left(\log T\right)^{1/3}}{\Bigl( \Delta_J^f \Bigr)^{2/3}}-1,
\end{equation}
it follows that 
\begin{equation*} \label{required_bound_slope}
    \min\{r_J - c_{m^{\ast}}^l, c_{k^{\ast}}^r - r_J\} 
    > C_3\frac{\left(\log T\right)^{1/3}}{\Bigl( \Delta_J^f \Bigr)^{2/3}},
\end{equation*}
and we can deduce that \eqref{lambda_expression} is restricted to
\begin{equation*}
    \lvert r_J-b^{\ast} \rvert 
    > C_3\frac{\left(\log T\right)^{1/3}}{\Bigl( \Delta_J^f \Bigr)^{2/3}}
\end{equation*}
which implies \eqref{contradiction_slope}. 
So, we conclude that necessarily
\begin{equation} \label{result_step3_slope}
    \lvert r_J - b_J\rvert \Bigl( \Delta_J^f \Bigl)^{2/3} 
    \leq C_3 \left(\log T \right)^{1/3}.
\end{equation}
But \eqref{min_bound3_slope} must be true since if we assume that 
\begin{equation*}
    \min\{r_J - c_{m^{\ast}}^l, c_{k^{\ast}}^r - r_J\} \leq 2^{1/3}C_3\frac{\left(\log T\right)^{1/3}}{\Bigl( \Delta_J^f \Bigr)^{2/3}}-1,
\end{equation*}
then, using Lemma \ref{Lemma_for_main_thm_slope1}, we have that
\begin{align*}
    C^{b_J}_{c_{{m^{\ast}}}^l, c_{{k^{\ast}}}^r}(\boldsymbol{X}) 
    & 
    \leq C^{r_J}_{c_{{m^{\ast}}}^l, c_{{k^{\ast}}}^r}(\boldsymbol{f})  + \sqrt{8\log T} \\
    & \leq \frac{1}{\sqrt{3}}\left( \min\{r_J - c_{m^{\ast}}^l, c_{k^{\ast}}^r - r_J\} + 1 \right) ^ {3/2} \Delta^f_J + \sqrt{8\log T}\\
    & \leq \frac{1}{\sqrt{3}}\left( 2^{1/3}C_3\frac{\left(\log T\right)^{1/3}}{\Bigl( \Delta_J^f \Bigr)^{2/3}} \right) ^ {3/2} \Delta^f_J + \sqrt{8\log T} \\ 
    & = \sqrt{\frac{2}{3}} C_3^{3/2}\sqrt{\log T} + \sqrt{8\log T} \\
    & = \left( \sqrt{\frac{2}{3}} C_3^{3/2} +\sqrt{8}\right) \sqrt{\log T}= C_1 \sqrt{\log T}
    \leq \zeta_T,
\end{align*}
which contradicts $C^{b_J}_{c_{{m^{\ast}}}^l, c_{{k^{\ast}}}^r}(\boldsymbol{X}) > \zeta_T$.
\vspace{0.1in}
\\
Thus, we have proved that for $\lambda_T \leq \delta_T/2n$, working under the assumption that both $A_T^{\ast}$ and $B_T^{\ast}$ hold, there will be an interval $[c_{m^{\ast}}^l, c_{k^{\ast}}^r]$ with $C^{b_J}_{c_{{m^{\ast}}}^l, c_{{k^{\ast}}}^r}(\boldsymbol{X}) > \zeta_T$, where $b_J = \textrm{argmax}_{c_{m^{\ast}}^l \leq t < c_{k^{\ast}}^r} C^{t}_{c_{{m^{\ast}}}^l, c_{{k^{\ast}}}^r}(\boldsymbol{X})$ is the estimated location for the change-point $r_J$ that satisfies \eqref{result_step3_slope}.
\vspace{0.1in}
\\
{\textbf{Step 4:}} After the detection of the change-point $r_J$ at the estimated location $b_J$ in the interval $[c_{m^{\ast}}^l, c_{k^{\ast}}^r]$, the process is repeated in the disjoint intervals $[1, c_{m^{\ast}}^l]$ and $[c_{k^{\ast}}^r, T]$.
Proving that there is no other change-point in the interval $[c_{m^{\ast}}^l, c_{k^{\ast}}^r]$ can be done in exactly the same way as Step 4 in the case of piecewise-constant signals and will not be repeated here.
\vspace{0.1in}
\\
{\textbf{Step 5:}} After detecting $r_J$, the algorithm will first check the interval $[1, c_{m^{\ast}}^l]$. 
So, unless $r_J = r_1$ and $[1, c_{m^{\ast}}^l]$ contains no other change-points, the next change-point to get detected will be one of $r_1, r_2, \ldots, r_{J-1}$. 
The location of the largest difference in the interval $[1, c_{m^{\ast}}^l]$, $d_{1, c_{m^{\ast}}^l}$, will again determine which change-point will be detected next and as in Step 5 for piecewise-constant signals, we only need to consider the case when the next change-point to get detected is $r_{J-1}$. 
\vspace{0.1in}
\\
Now, concentrating on the case that the next change-point to get detected is $r_{J-1}$, we mention that since this is the closest change-point to the already detected $r_J$, we need to make sure that detection is possible. 
As before, for $k_{J-1} \in \{1,\ldots,K^{\max}\}$ and $m_{J-1} \in \{k_{J-1} - 1, k_{J-1}\}$
we will show that $r_{J-1}$ gets detected in \linebreak $[c_{m_{J-1}^{\ast}}^l, c_{k_{J-1}^{\ast}}^r]$, where $c_{m_{J-1}^{\ast}}^l \geq c_{m_{J-1}}^l$ and $c_{k_{J-1}^{\ast}}^r \leq c_{k_{J-1}}^r \leq c_{m^{\ast}}^l$ and its detection is at location 
\begin{equation*}
    b_{J-1} = \textrm{argmax}_ {c_{m_{J-1}^{\ast}}^l \leq t < c_{k_{J-1}^{\ast}}^r} C^t _ {c_{m_{J-1}^{\ast}}^l, c_{k_{J-1}^{\ast}}^r}\left( \boldsymbol{X}\right),
\end{equation*} 
which satisfies $\bigl| r_{J-1} - b_{J-1}\bigl| \Bigl( \Delta_{J-1}^f \Bigl)^{2/3} \leq C_3 \left(\log T\right)^{1/3}$. 
Firstly, $r_{J-1}$ is isolated in the interval $[c_{m_{J-1}}^l, c_{k_{J-1}}^r]$ using the same argument as in \eqref{isolation} of the main paper. 
Using Lemma \ref{Lemma_for_main_thm_slope2}, we have that for $\tilde{b}_{J-1} = \textrm{argmax}_ {c_{m_{J-1}}^l \leq t < c_{k_{J-1}}^r} C^t _ {c_{m_{J-1}}^l, c_{k_{J-1}}^r} \left( \boldsymbol{X}\right)$,
\begin{align} \label{step 5}
    C ^{\tilde{b}_{J-1}} _ {c_{m_{J-1}}^l, c_{k_{J-1}}^r} & \left( \boldsymbol{X}\right) 
    \geq C ^{\tilde{r}_{J-1}} _ {c_{m_{J-1}}^l, c_{k_{J-1}}^r} \left( \boldsymbol{X}\right)
    \geq C^{\tilde{r}_{J-1}} _ {c_{m_{J-1}}^l, c_{k_{J-1}}^r}\left( \boldsymbol{f}\right) - \sqrt{8\log T} \nonumber\\
    & \geq \frac{1}{\sqrt{24}} \left(\min\{r_{J-1}-c_{m_{J-1}}^l,c_{k_{J-1}}^r-r_{J-1}\}\right)^{3/2} \Delta^f_{J-1}- \sqrt{8\log T}.
\end{align}
Before we show that $\min\{c_{k_{J-1}}^r-r_{J-1},r_{J-1}-c_{m_{J-1}}^l\} \geq \delta_T/2n$, we need show that $c_{m^{\ast}}^l$ satisfies $c_{m^{\ast}}^l - r_{J-1}\geq \delta_T/2n$
The proof is exactly the same as in Step 5 of the proof of Theorem 1. 
It can be deduced that the right end-point of the interval will satisfy $c_{k_{J-1}}^r - r_{J-1} > \delta_T/2n$ for some $k_{J-1}$. 
For the left end-point of the interval, the same holds as in Step 3.
In any case, $\min\{c_{k_{J-1}}^r-r_{J-1},r_{J-1}-c_{m_{J-1}}^l\} > \delta_T/2n$ holds and so, from \eqref{step 5}, using exactly the same calculations as \eqref{slope_observed}, we have that
\begin{align*}
    C^{\tilde{b}_{J-1}} _ {c_{m_{J-1}}^l, c_{k_{J-1}}^r} \left( \boldsymbol{X}\right)
    = C_2 (\delta_T)^{3/2} \underline{f}_T > \zeta_T.
\end{align*}
Therefore, we have shown that there exists an interval of the form $[c_{\tilde{m}_{J-1}}^l, c_{\tilde{k}_{J-1}}^r]$ with $\max_{c_{\tilde{m}_{J-1}}^l \leq b < c_{\tilde{k}_{J-1}}^r} C^t _ {c_{\tilde{m}_{J-1}}^l, c_{\tilde{k}_{J-1}}^r}\left( \boldsymbol{X}\right) > \zeta_T$. 
\vspace{0.1in}
\\
Now, denote $c_{m_{J-1}^{\ast}}^l, c_{k_{J-1}^{\ast}}^r$ the first points where this occurs and $b_{J-1}$ as defined above with $C^{b_J-1} _ {c_{m_{J-1}^{\ast}}^l, c_{k_{J-1}^{\ast}}^r} \left( \boldsymbol{X}\right) > \zeta_T$. 
We can show that $\bigl| r_{J-1} - b_{J-1}\bigl| \Bigl( \Delta_{J-1}^f \Bigl)^{2/3} \leq C_3 \left(\log T\right)^{1/3}$, following exactly the same process as in Step 3.
\vspace{0.1in}
\\
After detecting $r_{J-1}$ in the interval $[c_{m_{J-1}^{\ast}}^l, c_{k_{J-1}^{\ast}}^r]$, DAIS will restart on intervals $[1, c_{m_{J-1}^{\ast}}^l]$ and $[c_{k_{J-1}^{\ast}}^r, c^l_{m^\ast}]$. 
Step 5 can be applied to all intervals, as long as there is a change-point. 
We can conclude that all change-points will get detected, one by one, and their estimated locations will satisfy $\bigl| r_j - b_j\bigl| \Bigl( \Delta_j^f \Bigl)^{2/3} \leq C_3 \left(\log T\right)^{1/3}$, $\forall j \in \{1, 2, \ldots, N\}$. 
There will not be any double detection issues as each interval contains no previously detected change-points.
\vspace{0.1in}
\\
{\textbf{Step 6:}} After detecting all the change-points at locations $b_1, b_2, \ldots, b_N$ using the intervals $[c_{m_j^{\ast}}^l, c_{k_j^{\ast}}^r]$ for $j\in \{1, \ldots, N\}$, the algorithm will check all intervals of the form $[c_{k_{j-1}^{\ast}}^r, c_{m_j^{\ast}}^l]$ and $[c_{k_j^{\ast}}^r, c_{m_{j+1}^{\ast}}^l]$, with $c_{m_0^{\ast}}^l = 1$ and $c_{k_{N+1}^{\ast}}^r = T$. 
At most $N+1$ intervals of this form, containing no change-points will be checked. 
Denoting by $[s^{\ast}, e^{\ast}]$ any of those intervals, we can show that DAIS will not detect any change-point as for $b \in \{s^{\ast}, s^{\ast}+1, \ldots, e^{\ast} - 1\}$,
\begin{equation*}
    C^{b}_{s^{\ast}, e^{\ast}} \left( \boldsymbol{X}\right)
    \leq C^{b}_{s^{\ast}, e^{\ast}} \left( \boldsymbol{f}\right)  + \sqrt{8 \log T}
    = \sqrt{8 \log T}
    < C_1 \sqrt{\log T}
    \leq \zeta_T.
\end{equation*}
The algorithm will terminate after not detecting any change-points in all intervals.
\\
\end{proof}

\endgroup
\end{document}